\documentclass[11pt]{article}

\usepackage{amsmath,amssymb,amsthm,appendix,bbm,bm,color,enumerate,fullpage,longtable,lineno,mathrsfs,mathtools,nccmath,subcaption,tabularx,titling,setspace}
\usepackage[all]{xy}
\usepackage{natbib}
\usepackage{bibunits}
\usepackage{authblk}

\newtheorem{thm}{Theorem}
\newtheorem{lem}{Lemma}
\newtheorem{prop}{Proposition}
\newtheorem{cor}{Corollary}
\theoremstyle{remark} 
	\newtheorem{rem}{Remark}
	
\theoremstyle{definition} 
	\newtheorem{mydef}{Definition} 
	\newtheorem{exmp}{Example}

\newcommand{\abs}[1]{\left| #1 \right|}
\newcommand{\norm}[1]{\left\| #1 \right\|}
\newcommand{\BigO}[1]{\mathcal{O}{\textstyle\left( #1\right)}}

\newcommand{\defn}{:=}

\newcommand{\ie}{\textit{i.e.,}\,}
\newcommand{\eg}{\textit{e.g.,}\,}
\newcommand{\etc}{\textit{etc.}\, }
\newcommand{\nb}{\textit{n.b.,}\, }
\newcommand{\vs}{\textit{vs.}\;}

\newcommand{\nocontentsline}[3]{}
\newcommand{\tocless}[2]{\bigskip\bgroup\let\addcontentsline=\nocontentsline#1{#2}\egroup}

\newcolumntype{L}{>{$}l<{$}} \newcolumntype{C}{>{$}c<{$}}

\begin{document}

\title{\textbf{Pathogen evolution: \\slow and steady spreads the best}}

\author[1]{Todd L. Parsons}
\author[1,2]{Amaury Lambert}
\author[3]{Troy Day}
\author[4]{Sylvain Gandon}

\affil[1]{Laboratoire de Probabilit\'es et Mod\`eles Al\'eatoires (LPMA), UPMC Univ Paris 06, CNRS UMR 7599,
Paris, France}
\affil[2]{Center for Interdisciplinary Research in Biology (CIRB), Coll\`ege de France, CNRS UMR 7241, INSERM U1050, PSL Research University, Paris, France}
\affil[3]{Department of Biology, Queen's University, Kingston, Canada}
\affil[4]{Centre d'Ecologie Fonctionnelle et Evolutive (CEFE), CNRS UMR 5175, Universit\'e de Montpellier--Universit\'e Paul-Val\'ery Montpellier--EPHE, Montpellier, France}
\date{\today}

\maketitle

\noindent E-mails: TLP - todd.parsons@upmc.fr; AL - amaury.lambert@upmc.fr; TD - tday@mast.queensu.ca; SG -  sylvain.gandon@cefe.cnrs.fr
\bigskip

\defaultbibliography{VirStocBibliography}
\defaultbibliographystyle{apalike}

\begin{bibunit}

\begin{abstract}
The theory of life history evolution provides a powerful framework to understand the evolutionary dynamics of pathogens in both epidemic and endemic situations. This framework, however, relies on the assumption that pathogen populations are very large and that one can neglect the effects of demographic stochasticity. Here we expand the theory of life history evolution to account for the effects of finite population size on the evolution of pathogen virulence. We show that demographic stochasticity introduces additional evolutionary forces that can qualitatively affect the dynamics and the evolutionary outcome. We discuss the importance of the shape of pathogen fitness landscape and host heterogeneity on the balance between mutation, selection and genetic drift. In particular, we discuss scenarios where finite population size can dramatically affect classical predictions of deterministic models. This analysis reconciles Adaptive Dynamics with population genetics in finite populations and thus provides a new theoretical toolbox to study life-history evolution in realistic ecological scenarios.  
\end{abstract}


\tocless\section{Introduction}
\label{sec:intro}

Why are some pathogens virulent and harm their hosts while others have no effect on host fitness? Our ability to understand and predict the evolutionary dynamics of pathogen virulence has considerable implications for public-health management \citep{dieckmann2005adaptive,bull2014,gandon2016}. A classical explanation for pathogen virulence involves trade-offs with other pathogen life-history traits. If certain components of pathogen fitness, such as a high transmission rate or a long duration of transmission, necessarily require that the pathogen incidentally harm its host then virulence is expected to evolve \citep{frank1996models}. A now classical way to develop specific predictions from this hypothesis is to adopt the Adaptive Dynamics formalism \citep{geritz1998evolutionarily,Metz1992,dieckmann2005adaptive,frank1996models}. This approach relies on the assumption that the mutation rate is small so that the epidemiological dynamics occur on a faster timescale than the evolutionary dynamics \citep{anderson1992infectious,frank1996models,alizon2009virulence,cressler2016adaptive}. Under simple epidemiological assumptions (no co-infections with different genotypes) the evolutionarily stable level of virulence maximizes the basic reproduction ratio $R_{0}$ of the pathogen (but see \eg \citet{nowak1994superinfection,van1995dynamics} for more complex epidemiological scenarios). \\

Adaptive Dynamics models allow one to predict long-term evolution but they tell us little about what we should expect to observe if epidemiological and evolutionary processes occur on a similar timescale. Novel theoretical approaches have therefore been developed to address this issue \citep{lenski1994evolution,frank1996models,day2004general,day2006insights,bull2008invasion}. These studies have revealed that, in addition to tradeoffs, the nature of the epidemiological dynamics (e.g., epidemic spread versus endemic disease) can also dictate the type of pathogen that will evolve. For example, pathogens with relatively high virulence can be selected for during epidemic disease spread whereas pathogens with a lower virulence may outcompete such high virulence strains in endemic diseases \citep{lenski1994evolution,frank1996models,day2004general,berngruber2013evolution}. \\

The above-mentioned theory allows one to determine the level of virulence expected to evolve under a broad range of epidemiological scenarios but it still suffers from the fundamental shortcoming of being a deterministic theory. Pathogen population size, however, can be very small (e.g. at the onset of an epidemic or after a vaccination campaign) and demographic stochasticity is likely to affect both the epidemiological and evolutionary dynamics of the disease. If all that such stochasticity did was to introduce random noise then the predictions of deterministic theory would likely suffice. However, several recent studies have demonstrated that this is not the case. For example, \citet{kogan2014} and \citet{humplik2014evolutionary} each used different theoretical approaches to demonstrate that finite population size tends to select for lower virulence and transmission. Likewise, \citet{read2007stochasticity} analyzed the effect of finite population size in a complex epidemiological model with unstable epidemiological dynamics and showed that finite population size could induce an evolutionary instability that may either lead to selection for very high or very low transmission.

Taken together, the existing literature presents a complex picture of the factors that drive virulence evolution and it remains unclear how all of these factors are related to one another and how they might interact. In this paper we develop a very general theory of pathogen evolution that can be used to examine virulence evolution when all of the above-mentioned factors are at play. First, we use an individual based description of the epidemiological process to derive a stochastic description of the evolutionary epidemiology dynamics of the pathogen. This theoretical framework is used to pinpoint the effect of finite population size on the interplay between epidemiology and evolution. Second, we analyze this model under the realistic assumption that the rate of mutation is small so that pathogen evolution can be approximated by a sequence of mutation fixations. We derive the probability of fixation of a mutant pathogen under both weak and strong selection regimes, and for different epidemiological scenarios. Third, we use this theoretical framework to derive the stationary distribution of pathogen virulence resulting from the balance between mutation, selection and genetic drift. This yields new predictions regarding the effect of the shape of pathogen fitness landscape, the size of the population and sources of host heterogeneities on long-term evolution of the pathogen. As the question of virulence evolution can be viewed as a specific example of the more general notion of life history evolution \citep{stearns1992,roff2002} our results should be directly applicable to other life history traits and other organisms as well.

%These studies clearly show that demographic stochasticity can strongly alter the predictions of deterministic models. Yet, it is often difficult to reconcile these different results derived under very different assumptions and using different analytical or numerical approaches. Our objective is to develop a theoretical framework that may help relate multiple perspectives on pathogen evolution. First, we present a general formulation of the model that takes into account finite pathogen population size. When population size is assumed to be very large we recover the deterministic case but finite population size introduces demographic stochasticity and genetic drift. Second, we derive the probability of fixation of a mutant pathogen introduced at the fixation is used to obtain the stationary distribution of pathogen virulence under the effects of mutation, selection and drift. We compare these results with the predictions of fully deterministic models and discuss the effects of demographic stochasticity on pathogen evolution under different epidemiological scenarios.
 
\tocless\section{Model}

We use a classical SIR epidemiological model where hosts can either be susceptible, infected or recovered.  The number of each of these types of hosts is denoted by $N_S$, $N_I$, and $N_R$ respectively.  Because we are interested in the effect of demographic stochasticity the model is derived from a microscopic description of all the events that may occur in a finite host population of total size $N_T=N_S+N_I+N_R$ (the derivation of the model is detailed in  Supplementary Information). It will be useful to explicitly specify the size of the habitat in which the population lives (e.g., the area of the habitat) and so we denote this by the parameter $n$. 

We use $\lambda$ to denote the rate at which new susceptible hosts enter the population \textit{per unit area} and therefore the total rate is given by $\lambda n$. We focus on the case of frequency-dependent transmission; i.e., new infections occur at rate $\frac{\beta}{N_T} N_S N_I$ where $\beta$ is a constant quantifying the combined effects of contact rate among individuals and the probability of pathogen transmission given an appropriate contact occurs. Note, however, that other forms of transmission (\eg density dependent transmission,\citep{mccallum2001should}) yield qualitatively similar results \citep{Parsons2012}. For simplicity we also assume that already infected hosts cannot be reinfected by another pathogen strain (\ie no co-infections). All hosts are assumed to suffer a constant per capita death rate of $\delta$ and infected hosts die at per capita rate $\alpha$ and they recover at per capita rate $\gamma$. Finally, to study pathogen evolution we need to introduce genetic variation in the parasite population. Therefore we consider $d$ pathogen strains which differ in transmission rate $\beta_i$ and virulence $\alpha_i$,with $i\in\{1,...,d\}$. Likewise we use the subscripted variable $N_{I_i}$ to denote the number of hosts infected with strain $i$.

The dynamical system resulting from the above assumptions is a continuous-time Markov process tracking the number of individuals of each type of host. To progress in the analysis we use a diffusion approximation and work with host densities defined as $S=N_S/n$, $I_i=N_{I_i}/n$ and $N=N_T/n$ and we define the total density of infected hosts as $I=\sum_{i=1}^{d}I_i$. When $n$ is sufficiently large these variables can be approximated using a continuous state space and so this model can be described by a system of stochastic differential equations (see Supplementary Information, \S 3).

\tocless\subsection{Deterministic evolution}

\noindent In the limit where the habitat size (and thus the host population size) gets large, demographic stochasticity becomes unimportant and the epidemiological dynamics are given by the following system of ordinary differential equations:
\begin{equation}
\begin{split}
\dot{S}&=\lambda-\frac{\bar{\beta}}{N} SI-\delta S\\
\dot{I}&=\frac{\bar{\beta}}{N} SI-\left(\delta+\bar{\alpha}+\gamma\right)I \\
\dot{N}&=\lambda-\delta N-\bar{\alpha}Y \label{epidemioEq}\\
\end{split}
\end{equation}

\noindent The bars above $\alpha$, $\beta$ and $\gamma$ refer to the mean of the transmission and the virulence distributions of the pathogen population (\textit{i.e.}\, $\bar{\alpha} = \frac{\sum_{i=1}^{d} \alpha_{i} I_{i}}{I}$). In the absence of the pathogen the density of hosts equilibrates at $S_0= \frac{\lambda}{\delta}$. A monomorphic pathogen population ($m=1$, $\bar{\beta}=\beta$ and $\bar{\alpha}=\alpha$) is able to invade this equilibrium if its basic reproduction ratio is $R_0= \frac{\beta}{\delta+\alpha+\gamma} >1$. If this condition is fulfilled the system reaches an endemic equilibrium where $\frac{S_{eq}}{N_{eq}}=\frac{1}{R_0}$, $\frac{I_{\text{eq}}}{{N}_{eq}} = \frac{\delta}{\delta+\gamma} \left(1-\frac{1}{R_0} \right)$ and ${N}_{eq}= \frac{\lambda (\delta+\gamma)}{\delta(\beta-\alpha)}{R}_0 $.

When several strains are present in the population the evolutionary dynamics of the pathogen can be tracked with \citep{day2006insights,day2007applying}:
\begin{equation}
\dot{p_i} =p_i(r_i-\bar{r}) \label{evolEq}\\
\end{equation}
\noindent where $p_i =\frac{I_{i}}{I}$ is the frequency of pathogen $i$. The quantity $r_i=\beta_i \frac{S}{N}-\left(\delta+\alpha_i+\gamma\right)$ is the instantaneous per capita growth rate of strain $i$ and $\bar{r}=\sum_{i=1}^{d} p_i r_i$ is the average per capita growth rate of the pathogen population. When $m=2$ only two strains are competing (a wild-type, strain 1, and a mutant, strain 2) the change in frequency $p_2$ of the mutant strain is given by:
\begin{equation}
\dot{p_2} =p_1p_2\left(\frac{S}{N}\Delta\beta-\Delta\alpha\right) \label{evolEq2strains}\\
\end{equation}
\noindent where $\Delta\beta=\beta_2-\beta_1$ and $\Delta\alpha=\alpha_2-\alpha_1$ are the effects of the mutation on transmission and virulence, respectively.

The above formalization can be used to understand the evolution of pathogen life-history under different scenarios. First, under the classical Adaptive Dynamics assumption that mutation rate is very small one may use a separation of time scales where the epidemiological dynamics reach an endemic equilibrium (set by the resident pathogen, strain 1) before the introduction of a new variant (strain 2) by mutation. In this case evolution favours the strain with the highest basic reproduction ratio: $R_{(0,i)}= \frac{\beta_i}{\delta+\alpha_i+\gamma}$. In other words, evolution favours strains with higher transmission rates and lower virulence. According to the tradeoff hypothesis, however, transmission and virulence cannot evolve independently. For example, the within-host growth rate of pathogens is likely to affect both traits and result in a functional trade-off between transmission and virulence \citep{anderson1992infectious,frank1996models,alizon2009virulence,cressler2016adaptive}. Under this assumption equation \eqref{evolEq2strains} can be used to predict the evolutionary stable virulence strategy (Figure 1). The above model can also be used to predict virulence evolution when the evolutionary and epidemiological dynamics occur on a similar time scale \citep{day2006insights,day2007applying,gandon2007evolutionary}. For instance, these models can be used to understand virulence evolution during an epidemic \citep{lenski1994evolution,frank1996models,day2004general,berngruber2013evolution}. In this case, a pathogen strain $i$ with a lower $R_{0}$ may outcompete other strains if its instantaneous growth rate, $r_i$, is higher.

\begin{figure}[h]
\centering
    \includegraphics[width=0.85\linewidth]{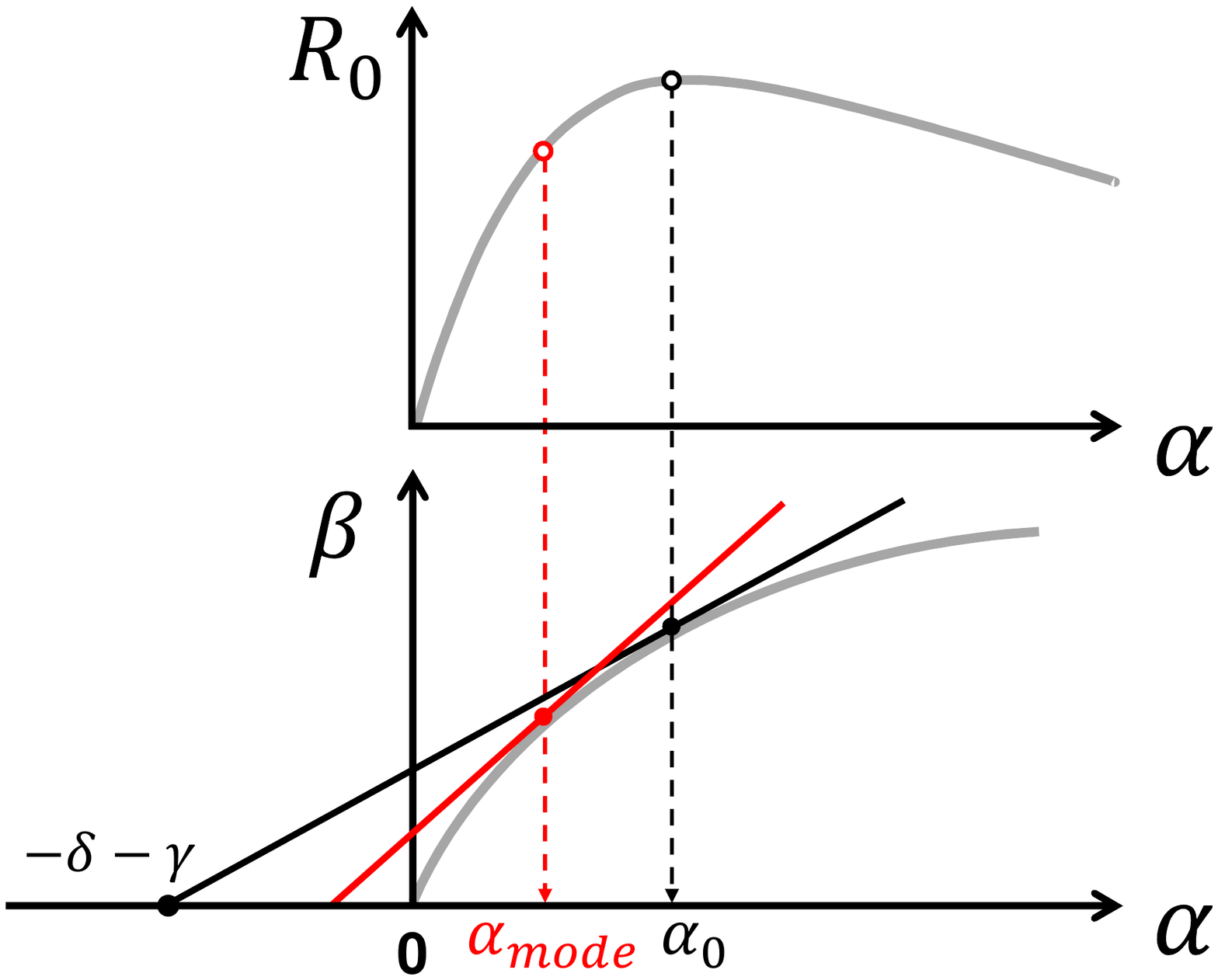} 
\caption{Schematic representation of the effect of finite population size on the
evolution of pathogen virulence. The grey line in the top figure represents the effect of
pathogen virulence, $\alpha$, on $R_0$ (for an asymmetric fitness function). The grey line in the
bottom figure represents the effect of pathogen virulence, $\alpha$, on pathogen transmission,
$\beta$. In the deterministic version of our model the marginal value theorem can be used to find the evolutionary stable (ES)
pathogen virulence, $\alpha_{0}$ (dashed black arrow). In this model ES virulence maximizes $R_0$
in the absence of demographic stochasticity. Finite population size modifies selection
and favours pathogen strategies with lower virulence (see equation \eqref{ESScond}). The mode of the stationary distribution of pathogen virulence is indicated by a dashed red arrow, $\alpha_{\text{mode}}$ (see equation \eqref{StatDisMode}). This geometrical construction indicates that finite population size is expected to favour \textit{slower} strains even if they have a lower $R_0$.}
\end{figure}

\tocless\subsection{Stochastic evolution}

\noindent Finite population size introduces demographic stochasticity and the epidemiological dynamics can be described by the following system of (It\^o) stochastic differential equations:
\begin{equation}
\begin{split}
	dS &=\left(\lambda-\frac{\bar{\beta}}{N} SI-\delta S \right) dt+  \sqrt{\frac{\lambda}{n}} dB_{1}
	- \sqrt{\frac{\delta S}{n}} dB_{2} 
	- \sqrt{\frac{\bar{\beta} SI}{n N}} dB_{3} \\
	dI &=\left(\frac{\bar{\beta}}{N} SI-\left(\delta+\bar{\alpha}+\bar{\gamma}\right)I \right) dt
	+ \sqrt{\frac{\bar{\beta} SI}{n N}} dB_{3} 
	- \sqrt{\frac{(\delta+\bar{\alpha})I}{n}} dB_{4} 
	- \sqrt{\frac{\bar{\gamma}I}{n}} dB_{5}
  \label{epidemioEqStoc}\\
	dN &=\left(\lambda-\delta N-\bar{\alpha}I\right) dt+ \sqrt{\frac{\lambda}{n}} dB_{1}
	- \sqrt{\frac{\delta S}{n}} dB_{2} 
	- \sqrt{\frac{(\delta+\bar{\alpha})I}{n}} dB_{4}   \\
\end{split}
\end{equation}

\noindent where $B_{1},\ldots,B_{5}$ are independent Brownian motions. As expected, when $n\rightarrow\infty$ this set of stochastic differential equations reduces to the deterministic equations in \eqref{epidemioEq}. \\

In finite populations the pathogen, and indeed the host population itself, are destined to extinction with probability 1. The time it takes for this to occur, however, depends critically on the parameter values. For example, in a monomorphic pathogen population (\ie, $m=1$), if $R_{0}$ is larger than one the size of the pathogen population reaches a quasi-stationary distribution which is approximately normal. The mean of this distribution is of order $n$ and its standard deviation of order $\sqrt{n}$ \citep{naasell2001extinction, naasell2007extinction}. The extinction time from the quasi-stationary distribution increases exponentially with $n$ \citep{Barbour1976,naasell2001extinction,naasell2007extinction} and so, in the remainder of the paper we will assume that $n$ is large enough so that we can focus on the dynamics conditional on non-extinction.          

As in the deterministic case, one can study evolutionary dynamics by focusing on the change in strain frequencies. We obtain a stochastic differential equation analogous to \eqref{evolEq} (see Supplementary Information, \S 4):
\begin{equation}
	dp_{i} =\left(p_{i}(r_{i}-\bar{r})-\frac{1}{n I}p_{i}(v_{i}-\bar{v})\right)\, dt
		+ \frac{1}{\sqrt{n I}} \sum_{j=1}^{d} (\delta_{ij}-p_{i})\sqrt{v_{j}p_{j}}\, dB_{j}, 
\label{evolEqStoc}\\
\end{equation}
where $v_i=\beta_i \frac{S}{N}+\left(\delta+\alpha_i+\gamma\right)$ is the variance in the growth rate of strain $i$ (while $r_i$ is the mean) and $\bar{v}=\sum_{i=1}^{d}p_i v_i$ is the average variance in growth rate of the pathogen population.  The first term in equation \eqref{evolEqStoc} is analogous to \eqref{evolEq}. The second term shows that finite population size (\ie when pathogen population size, as measured by the total density of infected hosts, $n I$ is not too large) can affect the direction of evolution. In contrast with the deterministic model, the evolutionary dynamics are not driven exclusively by the expected growth rate $r_i$ but also by a minimization of the variance. This effect is akin to bet-hedging theory stating that a mutant strategy with lower variance in reproduction may outcompete a resident strategy with a higher average instantaneous growth rate \citep{gillespie1974natural,frank1990evolution}. To better understand this effect it is particularly insightful to examine the case $m=2$ when only two strains are competing and the change in frequency $p_2$ of the mutant strain is given by:
\begin{equation}
	dp_{2} =p_1p_2\left(\frac{S}{N}\Delta\beta\left(1-\frac{1}{n I}\right)-\Delta\alpha\left(1+\frac{1}{n I}\right)\right)\, dt
		+ \sqrt{\frac{p_1p_2}{n I} \left(p_{1}v_{2}+p_{2}v_{1}\right)}\, dB, 
\label{evolEqStoc2strains}\\
\end{equation}
\noindent The first term (the drift term) in equation \eqref{evolEqStoc2strains} is similar to \eqref{evolEq2strains} except for the $\frac{1}{n I}$ terms. Those terms are due to the fact that a transmission (or a death) event of the mutant is associated with a change in the number of mutants as well as an increase (decrease) of the total pathogen population size by one individual. This concomitant variation of pathogen population size affects the effective change of the mutant frequency (relative to the change expected under the deterministic model where population size are assumed to be infinite). This effect decreases the benefit associated with higher transmission and increases the cost of virulence. In the long-term this effect (the drift term in \eqref{evolEqStoc}) is thus expected to select for lower virulence. But this long term evolutionary outcome cannot be described by an evolutionary stable state because demographic stochasticity is also expected to generate noise (the diffusion term in \eqref{evolEqStoc}). Indeed, this stochasticity (\ie genetic drift) may lead to the invasion and fixation of strains with lower per capita growth rates. In the following we fully characterize this complex evolutionary outcome with the stationary distribution of pathogen virulence under different epidemiological scenarios.
\\

\tocless\section{Results}

The above theoretical framework embodied by the stochastic differential equations \eqref{epidemioEqStoc} and \eqref{evolEqStoc} subsume the deterministic model and can be used to study the interplay of all the relevant factors affecting virulence evolution. In the following we will assume that pathogen mutation is rare so that evolution can be described, as in classical Adaptive Dynamics, as a chain of fixation of new pathogen mutations. In contrast with Adaptive Dynamics, however, demographic stochasticity may allow deleterious mutations to go to fixation. The analysis of the effect of finite population size requires specific ways to quantify the stochastic fate of a genotype \citep{proulx2002can}. To determine the fate of a new mutation we need to compute the probability of fixation of a mutant pathogen in a resident population. In the absence of selection the fixation probability of a mutant allele depends only on the demography of the population. When the size of the population is fixed and equal to $N$ the fixation probability of a neutral allele is $1/N$. When the fixation probability of a mutant is higher than neutral it indicates that the mutant is selectively favoured. This is particularly useful in many complex situations where the interplay between selection and genetic drift are difficult to disentangle like time varying demography \citep{OttoWhitlock1997, lambert2006} or spatial structure \citep{rousset2004genetic}. In our model, the difficulty arises from (i) the stochastic demography of the pathogen population and (ii) the fact that pathogen life-history traits feed-back on the epidemiological dynamics and thus on the intensity of genetic drift.\\

\tocless\subsection{Stationary distribution of pathogen virulence at equilibrium}

Here we assume, as in the Adaptive Dynamics framework, that the pathogen mutation rate $\mu$ is so low that the mutant pathogen (strain $2$) arises when the resident population (strain $1$) has reached a stationary equilibrium $n I_{\text{eq}}$ (\ie close to the endemic equilibrium derived in the deterministic model). The $R_0$ of the two strains may be written in the following way: $R_{0,2}=R_{0,1}(1+s)$ where $s$ measures the magnitude of selection.\\

When selection is strong (\ie $s \gg \frac{1}{n}$) the probability of fixation of the mutant when $N_{I_{2}}(0)$ mutants are introduced into a resident population at equilibrium is (see Supplementary Information, \S 5.2):
\begin{equation}
U_{\text{strong}} \approx 1- \left(\frac{R_{0,1}}{R_{0,2}}\right)^{N_{I_{2}}(0)} \approx N_{I_{2}}(0) s,\label{PfixStrong}\\
\end{equation}
which may be obtained by approximating the invading strain by a branching process (see Supplementary Information, \S 7.2 for a rigorous justification).  When the mutant and the resident have similar values of $R_0$ (\ie $s$ is of order $\frac{1}{n}$) selection is weak, and  the derivation of the probability of fixation is a much more difficult problem. The classical population genetics approach under the assumption that population size is fixed (or is characterized by a deterministic trajectory independent of mutant frequency) is to use the diffusion equation of mutant frequency to derive the probability of fixation \citep{OttoWhitlock1997, lambert2006}. But in our model, equation \eqref{evolEq2strains} is not autonomous and is coupled with the epidemiological dynamics. To derive the probability of fixation we use a separation of time scale argument to reduce the dimension of the system (see \citep{Parsons2017} for a discussion of the approach). Indeed, if selection is weak the deterministic component of the model sends the system rapidly to the endemic equilibrium. At this point, it is possible to approximate the change in frequency of the mutant by tracking the dynamics of the projection of the mutant frequency on this manifold (see Supplementary Information, \S 5.3). This one dimensional system can then be used to derive the probability of fixation under weak selection. A first order approximation in $s$ and $\sigma$ is:
\begin{equation}
U_{\text{weak}} \approx p+ \frac{p(1-p)}{2}\left(n I_{\text{eq}} s + \sigma\right)\approx n I_2(0)\left(\frac{1}{n I_{\text{eq}}}+\frac{1}{2}\left(s+\frac{\sigma}{n I_{\text{eq}}}\right)\right) \label{PfixWeak}\\
\end{equation}

\noindent where $p=I_2(0)/I_{\text{eq}}$ and $\sigma=\frac{\beta_1-\beta_2}{\beta_2}$. The first term in \eqref{PfixWeak} is the probability of fixation of a single neutral mutation introduced in a pathogen population at the endemic equilibrium $n I_{\text{eq}}$. The second term takes into account the effect due to selection. First, selection may be driven by differences in $R_0$. Second, even if strains have identical $R_0$ (\ie $s=0$) selection may be driven by $\sigma$ which measures the difference in transmission rate. Note, however, that the effect of $s$ rapidly overwhelms the effect of $\sigma$ as pathogen population size $n I_{\text{eq}}$ becomes large (unless $s$ is of order $\frac{1}{n}$). The probability of fixation given in \eqref{PfixWeak} confirms that evolution tends to push towards higher basic reproductive ratio but when the population size is small other forces may affect the evolutionary outcome. In particular, when $n I_{\text{eq}}$ is small, strains with lower $R_0$ can reach fixation. Figure 2 shows the result of stochastic simulations that confirm the approximations \eqref{PfixStrong} and \eqref{PfixWeak} under different epidemiological scenarios. \\

\begin{figure}[h]
\centering
    \includegraphics[width=0.85\linewidth]{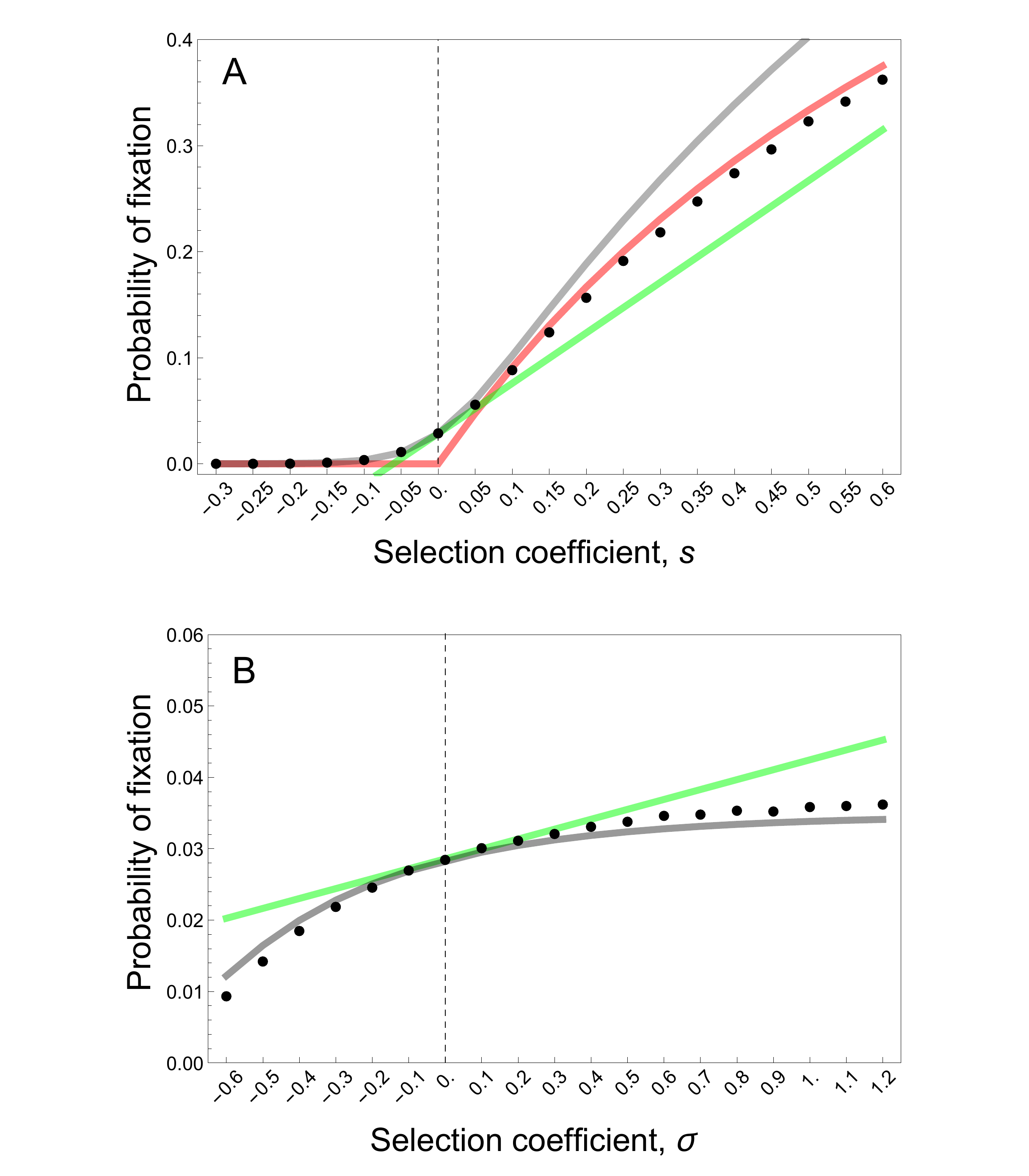} 
\caption{Probability of fixation for (A) different
values of $s$ (strong selection effect) and (B) different values of $\sigma$ (weak selection effect). Simulation results are indicated with a dot, weak selection approximation is indicated with a gray line and its linear approximation (equation $\eqref{PfixWeak}$) is indicated with a green line, the strong selection approximation is indicated with a red line (equation $\eqref{PfixStrong}$). Parameter values of the resident population: $n=100$, $R_0=4$, $\delta=1$, $\alpha=3$, $\gamma=1$, $\lambda=2$, $\beta_1=20$. For the simulation a single mutant (an individual host infected with a mutant pathogen) is introduced at the endemic equilibrium set by the resident pathogen: $S_{\text{eq}}=24$ and $I_{\text{eq}}=35$. $10^6$ simulations are realized for each parameter values and we plot the proportion of the simulations where the mutant goes to fixation.}
\end{figure}

Even though the probability of fixation helps understand the interplay between selection and genetic drift it does not account for any differences in the time to fixation and it is often difficult to measure this probability experimentally as well (but see \citet{Gifford20120310}). What may be more accessible is a characterization of the phenotypic state of the population across different points in time (or in space among replicate populations) - that is, the stationary distribution of the virulence phenotype of the pathogen under the action of mutation, selection and genetic drift \citep{Champagnat+Lambert07,lehmann2012stationary,debarre2016} (Figure 3).\\  

\begin{figure}[h]
\centering
    \includegraphics[width=0.75\linewidth]{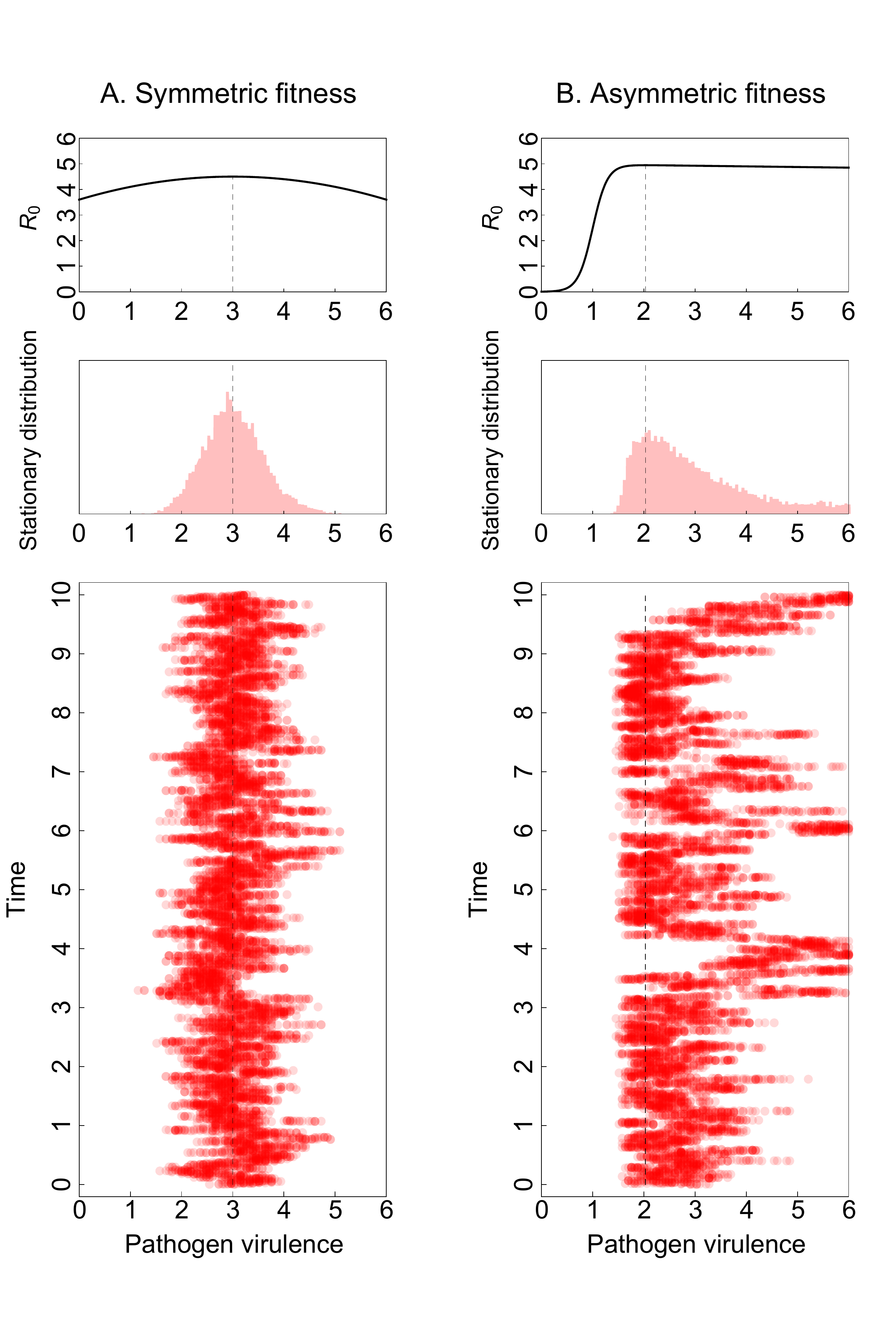} 
\caption{Dynamics of pathogen virulence across time (one time unit on the graph is $10^7$ time steps in the simulation, and a ) and stationary distribution of pathogen virulence for two different fitness landscapes: (A) Symmetric fitness landscape with $\beta(\alpha)=(\delta+\gamma+\alpha) R_{0,\text{max}}\left(1-w(\alpha_0-\alpha)^2\right)$, $R_{0,\text{max}}=4.5$ and $\alpha_0=3$, (B) Asymmetric fitness landscape with $\beta(\alpha)=5(1-0.005 \alpha)/(1+\exp(7(1-\alpha))$. The dashed vertical line indicates the position of $\alpha_0$. Other parameter values: $n=200$,  $\delta=1$, $\alpha=3$, $\gamma=1$, $\lambda=2$, $\mu=0.001$.}
\end{figure}

%The fitness landscape of the pathogen may be defined by the way $R_0$ varies with life-history traits $\alpha$ and $\beta$. When pathogen population size is very large the deterministic model shows that the pathogen should evolve towards a maximization of $R_0$. When the maximum of the fitness landscape  is reached no mutations can invade. In finite pathogen populations, however, the concept of uninvadable strategy is no longer relevant because even the strategy maximizing  $R_0$ can be replaced by a new one with a probability given by equations \eqref{PfixWeak}. Consequently, the phenotypic value of the pathogen is expected to fluctuate over time because new mutations will constantly challenge the resident strategy (Figure 3). The balance between mutation, selection and demographic stochasticity will result in a stationary distribution \citep{champagnat2007evolution,lehmann2012stationary,debarre2016}.\\ 
To derive the stationary distribution of pathogen virulence we first need to impose a trade-off between virulence and contact rate, setting $\beta = \beta(\alpha)$, and introduce the  mutation kernel $K(\alpha_{m},\alpha)$, the probability distribution of mutants with strategy $\alpha_{m}$ from a monomorphic population with strategy $\alpha$. Here we assume that this distribution is Gaussian with a mean equal to the current resident trait value and variance $\nu$. Under the assumption that the mutation rate $\mu$ remains small, pathogen polymorphism is limited to the transient period between the introduction of a mutant and a fixation. The probability of fixation \eqref{PfixWeak} accurately describes the direction of evolution and the evolution of pathogen virulence can then be described by the following Fokker-Planck diffusion equation (see Supplementary Information, \S 6):  
\begin{equation}
\frac{\partial{\psi(\alpha,t)}}{\partial{t}}=-\frac{\mu \nu}{2}  \frac{\partial}{\partial{\alpha}}\left[\left(n I_{\text{eq}}\frac{R'_{0}(\alpha)}{R_{0}(\alpha)}-\frac{\beta'(\alpha)}{\beta(\alpha)}\right) \psi(\alpha,t)\right]+ \frac{\mu \nu}{2} \frac{\partial^2{\psi(\alpha,t)}}{\partial{\alpha^2}}\label{DiffAlpha}\\
\end{equation}
\noindent where $\psi(\alpha,t)$ is the distribution of pathogen virulence and $'$ indicates the derivative with respect to $\alpha$. The drift term of the above equation indicates that deterministic evolution tends to maximize the basic reproduction ratio while finite population size tends to select for lower transmission. Under the classical assumption that pathogen transmission and pathogen virulence are linked by a genetic trade-off one can ask what the level of pathogen virulence is where the drift coefficient is zero. This trait value corresponds to the mode of the stationary distribution of pathogen virulence and is given by the following condition (see Supplementary Information Equation S.45): 
\begin{equation}
\beta'(\alpha)=R_{0}(\alpha)\left(1+\frac{1}{n I_{\text{eq}}(\alpha)-1}\right)\label{ESScond}\\
\end{equation}
\noindent When the pathogen population is very large (\ie $n \rightarrow \infty$) we recover the marginal value theorem while finite population size increases the slope $\beta'(\alpha)$ and reduces the mode of the stationary distribution (see Figure 1). Thus, for a broad range of transmission-virulence trade-off functions, finite population size is expected to decrease virulence and transmission rates. In other words, pathogen avirulence may be viewed as a bet-hedging strategy because even if it reduces the instantaneous growth rate $r_i$, the reduced variance in growth rate $v_i$ is adaptive in finite population size.\\  

\noindent Let us now consider the limiting case when all the pathogen strains have the same $R_0$. This corresponds to a very special case where the fitness landscape is flat. The deterministic model predicts that pathogen life-history variation is neutral near the endemic equilibrium (see \eqref{evolEq}). The probability of fixation \eqref{PfixWeak} shows, however, that selection is \textit{quasi-neutral} and favours pathogens with lower transmission and virulence rates \citep{parsons2007II,Parsons2012,kogan2014,humplik2014evolutionary}. The stationary distribution results from the balance between selection (pushing towards lower values of pathogen traits) and mutation (reintroducing variation). If we focus on virulence and allow variation between a minimal value $\alpha_{min}$ and a maximal value $\alpha_{\text{max}}$ the stationary distribution is (see Supplementary Information Equation S.39):     
\begin{equation}
\psi_{\text{flat}}(\alpha)=
	\frac{1}{\ln{\left(\frac{\delta+\alpha_{\text{max}}+\gamma}{\delta+\alpha_{min}+\gamma}\right)}}
		\frac{1}{(\delta+\alpha+\gamma)}. \label{StatDisFlat}\\
\end{equation}

\noindent It is worth noting that this distribution is independent of the pathogen population size. Indeed, near the endemic equilibrium and when pathogens have the same $R_0$ the probability of fixation \eqref{PfixWeak} is independent of pathogen population size. So this prediction holds even in very large pathogen populations. The time to fixation may, however, be considerably longer in large populations and the assumption that polymorphism is always reduced to the resident and a single mutant may not always hold as pathogen population increases. Yet, stochastic simulations confirm that \eqref{StatDisFlat} correctly predicts the stationary distribution, which is relatively insensitive to pathogen population size but varies with $\delta+\gamma$ (Figure 4A). 

\begin{figure}[h]
\centering
    \includegraphics[width=0.85\linewidth]{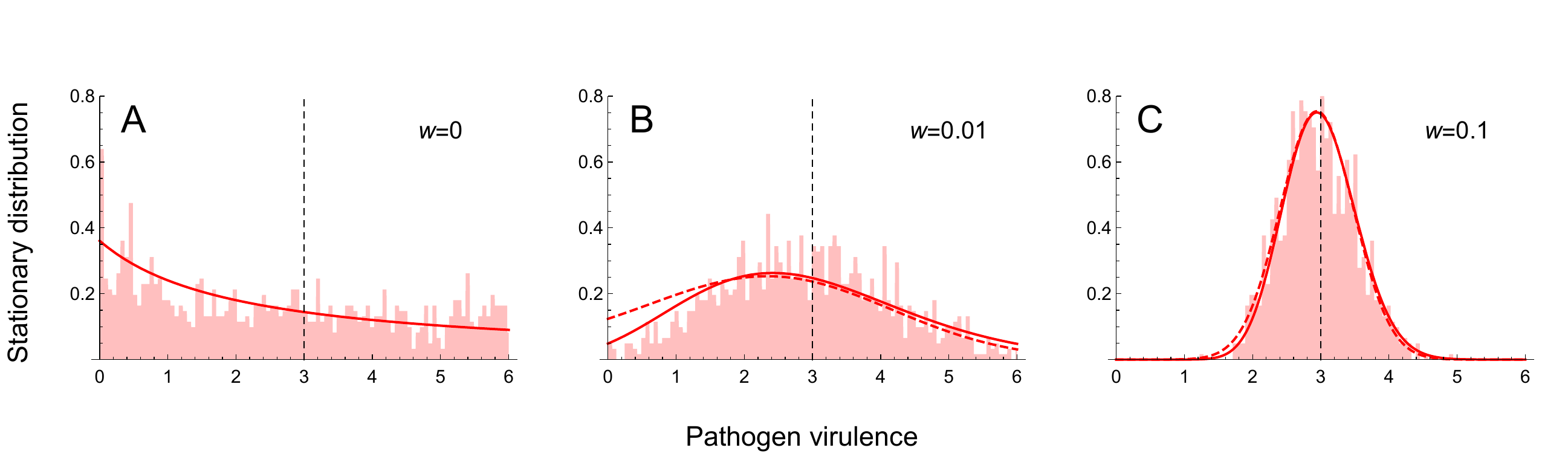} 
\caption{Stationary distribution for symmetric fitness landscapes with increasing strength of selection around the optimum with $\beta(\alpha)=(\delta+\gamma+\alpha) R_{0,\text{max}}\left(1-w(\alpha_0-\alpha)^2\right)$ and $\alpha_0=3$ for three different values of $w$: (A) $w=0$, (B) $0.01$ and (C) $0.1$. Note that when $w=0$ the fitness landscape is flat. The light red histogram indicates results of a stochastic simulation. The red line indicates the stationary distribution of the diffusion approximation (the dashed line indicates the approximation of this distribution, see \eqref{StatDisApprox}). The dashed vertical line indicates the position of $\alpha_0$. Parameter values: $n=200$, $R_{0,\text{max}}=4$, $d=1$, $\alpha=3$, $\gamma=1$, $\lambda=2$,  $\mu=0.001$.}
\end{figure}

Second, we consider a general fitness landscape with a single maximum. It is possible to derive a good approximation for the stationary distribution (see S.44 in the Supplementary Information):
\begin{equation}
\psi_{\text{approx}}(\alpha)=
\frac{\beta(\alpha_{0})}{\beta(\alpha)}	\mathcal{N}(\alpha_{0},\sigma^{2}),
\label{StatDisApprox}\\
\end{equation}

\noindent where $\alpha_{0}$ is the virulence that maximizes $R_{0}$, $I_{\text{eq}}(\alpha_{0})$ is the expected number of infected individuals at the endemic equilibrium when the virulence is $\alpha_{0}$ and $\mathcal{N}(\alpha_{0},\sigma^{2})$ is the Gaussian distribution with mean $\alpha_{0}$ and variance $\sigma^2=\frac{1}{n I_{\text{eq}}(\alpha_{0})|R_{0}''(\alpha_{0})|/(R_{0}(\alpha_{0}))}$. We thus see the effect of the demography is to bias the Gaussian, putting more weight on values of the virulence below $\alpha_{0}$; this becomes more clear when we consider the mode and mean of the (true) stationary distribution
%To better grasp the effect of the different parameters of the model on this distribution we may focus on the mode of the stationary distribution 
(see \S 6.3 in the Supplementary Information):
\begin{equation}
\alpha_{\text{\text{mode}}} =  \alpha_{0} 
		- \frac{\sigma^2}{\delta+\alpha_{0}+\gamma},
\label{StatDisMode}\\
\end{equation}
and
\begin{equation}
\alpha_{\text{\text{mean}}} =  \alpha_{0} - \sigma^2 \left(
		\frac{1}{\delta+\alpha_{0}+\gamma} 
		-\frac{I_{\text{eq}}'(\alpha_{0})}{I_{\text{eq}}(\alpha_{0})} 
		+ \frac{R_{0}'''(\alpha_{0})}{|R_{0}''(\alpha_{0})|}\right)
\end{equation}
respectively. These results indicate that, as expected from the simple optimization approach used above in \eqref{ESScond} and illustrated in Figures 3 and 4, lower pathogen population size tends to decrease pathogen virulence. Yet, the above derivation of the stationary distribution goes far beyond this optimization criterion. First, it predicts accurately the mode of the stationary distribution. In particular it shows that the shape of the fitness landscape may affect the mode of the stationary distribution. The skew of the fitness landscape can have huge effects on the stationary distribution (Figure 3). A positive skew leads to a higher mean virulence and may thus counteract the effect of small pathogen population. In other words, whether demographic stochasticity favours lower of higher virulence depends also on the shape of the fitness landscape. Second, our analysis predicts the amount of variation one may expect to see around this mode. Unlike the criteria used to derive a single optimal strategy our approach predicts accurately the expected variation around this mode (Figures 3 and 4). Note that the population remains monomorphic most of the time (because mutation is assumed to be small) but the variance of the stationary distribution refers to the distribution of phenotypes explored through time (or through space if stochastic evolution is taking place in multiple isolated populations).\\

\subsection{Pathogen evolution after vaccination}

The above analysis relies on the assumption that the epidemiological dynamics are much faster than evolutionary dynamics so that the pathogen population is always at its endemic equilibrium. The present framework can also be used to explore the fate of a mutant pathogen away from this endemic equilibrium. For instance, right after the start of a vaccination campaign the availability of susceptible hosts is going to drop rapidly if the vaccination coverage $f$ is high and if the efficacy of the vaccine is large. This epidemiological perturbation has major consequences on the probability of fixation of a pathogen mutant. Looking at invasion out of endemic equilibrium, we show that a larger vaccination coverage $f$ decreases the probability of fixation of all pathogen mutants (see Supplementary Information \S 5.2.3) but the probability of fixation of strains with low virulence and low transmission rates are less affected than more virulent and transmissible strains. In other words, strains with low turn over rates (\textit{slower} strains) are less likely to be driven to extinction during the drop of the pathogen population size. This extends classical results of populations genetics \citep{OttoWhitlock1997,lambert2006} to situations where the mutations are acting on life history traits and feed back on population dynamics. \textit{Faster} strains are selected for during epidemics while \textit{slower} strains are favoured when the pathogen population is reduced (e.g. after a public health intervention). This is also consistent with the analysis of pathogen evolution based on deterministic models which showed that the direction of selection depends on the epidemiological state of the population \citep{lenski1994evolution,frank1996models,day2004general,lambert2006, day2006insights,day2007applying,berngruber2013evolution}.\\
Vaccination is also meant to induce long-term modifications of the host population. In particular, artificial immunization introduces a heterogeneity between vaccinated and unvaccinated hosts. If the vaccine is perfect, vaccination will act on the epidemiology and will reduce the endemic equilibrium. In a deterministic version of this model, such a perfect vaccine is expected to have no consequences on long-term pathogen evolution \citep{gandon2001,gandon2007evolutionary}. In contrast, when host population size is finite, vaccination is expected to magnify the influence of demographic stochasticity and, as discussed above, to select for lower pathogen virulence and to increase the variance of the stationary distribution (Figure 5). It is also interesting to consider an alternative scenario where vaccinated hosts can be infected but cannot transmit the pathogen. In this case, two types of infected hosts are coexisting: good-quality (na\"ive) hosts and bad-quality (vaccinated) hosts. In this situation, the amount of demographic stochasticity is not governed by the whole pathogen population size but by the size of the population of infected hosts that actually contribute to transmission. In fact all the results derived above apply in this scenario provided that the equilibrium pathogen population size is replaced by the \textit{effective} pathogen size: $n I_{e}=(1-f) n I_{\text{eq}}$. This example illustrates that even for large pathogen population sizes the effect of demographic stochasticity can be important if the \textit{effective} pathogen size is small. When there is variation in reproductive value among individuals the \textit{effective} population size that governs the amount of genetic drift may be substantially lower than the \textit{actual} size of the population \citep{crow1970introduction}. In the context of pathogen evolution this heterogeneity in reproductive value may be driven by variations in infectiousness among hosts. This variation can be induced by public-health interventions (e.g. transmission-blocking vaccines) but it emerges naturally from complex behavioural and/or physiological differences among hosts. For instance, evidence for \textit{superspreading} events where certain individuals can infect unusually large numbers of secondary cases have been found in many human pathogens \citep{woolhouse1997heterogeneities,lloyd2005superspreading}. This heterogeneity is expected to reduce the effective population size and to magnify the effects of demographic stochasticity discussed above. But other factors like temporal fluctuations in population size are also known to modulate the intensity of genetic drift and may affect effective population size \citep{crow1970introduction}. 

\begin{figure}[h]
\centering
    \includegraphics[width=0.85\linewidth]{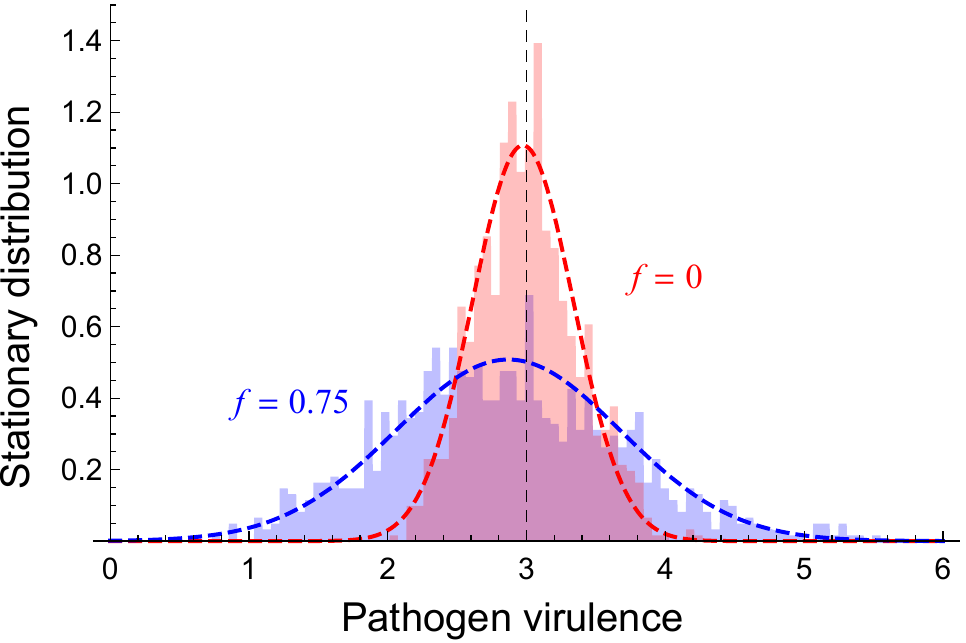} 
\caption{Stationary distribution for symmetric fitness landscapes with increasing vaccination coverage ($f=0$ in red and $f=0.75$ in blue) with $\beta(\alpha)=(\delta+\gamma+\alpha) R_{0,\text{max}}\left(1-w(\alpha_0-\alpha)^2\right)$ and $\alpha_0=3$. The histogram indicates results of a stochastic simulation. The blue and red dashed lines indicate the approximation of the stationary distribution given in \eqref{StatDisApprox} with or without vaccination, respectively. The dashed vertical line indicates the position of $\alpha_0$. Parameter values: $n=1000$, $R_{0,\text{max}}=10$, $\delta=1$, $\alpha=3$, $\gamma=1$, $\lambda=2$, $\mu=0.01$, $w=0.1$}
\end{figure}

\tocless\section{Discussion}

Evolutionary theory has led to the development of different mathematical tools for studying phenotypic evolution in a broad diversity of ecological scenarios \citep{parker1990,roff1993,Otto2007}. For instance, Adaptive Dynamics is a powerful theoretical framework to study life-history evolution when mutation is assumed to be rare so that demographic and evolutionary processes can be decoupled \citep{geritz1998evolutionarily,Metz1992}. This analysis yields evolutionarily stable life-history strategies and captures the ultimate outcome of evolution. But this approach relies on the assumption that population size is infinite and that evolution is deterministic. Finite population size, however, can also affect evolutionary trajectories. In particular, even the fittest genotype can be invaded by a deleterious mutant when population size is reduced. This leads to the collapse of the concept of evolutionarily stable strategy. Here we develop and apply Stochastic Adaptive Dynamics (SAD) \citep{Champagnat+Lambert07,Otto2007, debarre2016}, a new theoretical framework where the evolutionary outcome of life history evolution is studied through the derivation of the stationary distribution of the phenotype under mutation-selection-drift equilibrium. Under the assumption that mutation rate is small, the equilibrium distribution can be derived from a diffusion approximation. In contrast with previous population genetics models, the present framework also allows life-history evolution to affect population size and, consequently, the amount of demographic stochasticity. In other words, this framework retains key features of Adaptive Dynamics but relaxes a major assumption by allowing genetic drift to affect the evolutionary outcome \citep[see also][p.1149]{waxman200520}. As such, our SAD framework is an important step towards a better integration between Adaptive Dynamics and classical population genetics.    
 
We show that finite population size induces a selective pressure towards strains with lower variance in growth rate (but see also \citet{gillespie1974natural,lambert2006}). A simple way to understand this effect is to compare the fate of two strains with the same $R_0$ but with different life-history strategies. The \textit{fast} strain is very transmissible but has a short duration of infection (\eg because of high virulence or high clearance rate). The \textit{slow} strain has a long duration of infection but has a small transmission rate. Since the two strains have the same $R_0$ Adaptive Dynamics predicts that these two strains should coexist. With finite population size, however, the fast strain has a higher probability to go extinct simply because more events happen per unit of time. As in Aesop's Fable ``Slow and steady wins the race'' because the fast strain will reach extinction sooner than the slow strain. Previous studies \citep{Parsons2012,kogan2014,humplik2014evolutionary} pointed out the influence of finite population size on the direction of virulence evolution but they focused mainly on the quasi-neutral case where all the strains have the same $R_0$. \citet{humplik2014evolutionary} did look at scenarios where strains have different $R_0$ but without a derivation of the stationary distribution at mutation-selection-drift equilibrium. We believe that this stationary distribution is key to explore the interaction between finite population size and phenotypic evolution. This distribution yields testable predictions on the mean as well as other moments of the phenotypic distribution. 

The approximation \eqref{StatDisApprox} shows that this distribution is moulded by two main parameters: (i) the pathogen fitness landscape, and (ii) the effective size of the pathogen population. First, the fitness landscape at the endemic equilibrium can be derived from \eqref{evolEqStoc} and depends mainly on the way $R_0$ varies with pathogen life history traits. Under the classical transmission-virulence assumption $R_0$ is maximized for some intermediate virulence. But the shape of the trade-off also affects the shape of the fitness landscape and in particular its symmetry. Second, the effective size $n I_e$ of the pathogen population size depends mainly on the pathogen population size $n I_{\text{eq}}$ at the endemic equilibrium but other factors may reduce the effective pathogen population size as well. For instance, variance in transmission among infected hosts is likely to reduce $n I_e$ below $n I_{\text{eq}}$. One source of heterogeneity in transmissibility may be induced by public-health interventions (\eg vaccination, drug treatments), but intrinsic behavioural or immunological heterogeneities among hosts may induce superspreading transmission routes as well \citep{woolhouse1997heterogeneities,lloyd2005superspreading}. 

When the fitness landscape of the pathogen is symmetric, reducing the effective population size increases the variance of the stationary distribution but decreases also the mean (and the mode) of this distribution. This effect results from the selection for a reduction of the variance identified in \eqref{evolEqStoc}. This is the effect that emerges in the quasi-neutral case. When the fitness landscape is flat this may lead to an important bias towards lower virulence (Figure 4). When the fitness landscape of the pathogen is asymmetric the skewness of the fitness landscape can affect the mean of the stationary distribution when the effective population size  $n I_{e}$ of the pathogen is reduced. More specifically negative (positive) skewness reduces (increases) the mean of the stationary distribution. It is interesting to note that classical functions used to model the trade-off between virulence and transmission tend to generate positive skewness in the fitness landscape \citep{van1995dynamics,frank1996models,alizon2009virulence}. The asymmetry of these fitness functions may thus counteract the effects of stochasticity per se identified in symmetric fitness landscapes. In other words, predictions on the stochastic evolutionary outcome are sensitive to the shape on genetic constraints acting on different pathogen life-history traits. This result is very similar to the deterministic effects discussed in \citep{urban2013asymmetric} on the influence of asymmetric fitness landscapes on phenotypic evolution. Note, however, that the effect analyzed by Urban et al (2013) is driven by environmental effects on phenotypes. In our model, we did not assume any environmental effects and a given genotype is assumed to produce a single phenotype. 

We focused our analysis on the stationary distribution at the endemic equilibrium of this classical SIR model. But we also explore the effect of demographic stochasticity on the transient evolutionary dynamics away from the endemic equilibrium. For instance we recover a classical population genetics result \citep{OttoWhitlock1997,lambert2006} that the probability of fixation of adaptive mutations (\ie with $s \gg \frac{1}{n}$) is increased during epidemics. Beyond the effect of $s$ (\ie differences between the $R_0$ of the two competing strains) differences in life-history traits matter away from the endemic equilibrium. In particular, faster strains have higher probabilities of fixation when the pathogen population is growing (during epidemics) and, conversely, slower strains have higher probabilities of fixation when the pathogen population is crashing. The analysis of scenarios where the epidemiological dynamics is unstable and leads to recurrent epidemics is more challenging but may lead to unexpected evolutionary dynamics \citep{read2007stochasticity}.

We analyzed the effects of demographic stochasticity induced by finite population size but environmental stochasticity may also affect evolution \citep{frank1990evolution,starrfelt2012bet,schreiber2015unifying}. Environmental factors are known to have dramatic impacts on pathogen transmission and it would thus be particularly relevant to expand the current framework to account for the effects of random perturbations of the environment on pathogen evolution \citep{Nguyen2015}. 
 
Another possible extension of this model would be to analyze the effect of demographic stochasticity on the multi-locus dynamics of pathogens. Indeed, the interaction between genetic drift and selection is known to yield complex evolutionary dynamics resulting in the build up of negative linkage disequilibrium between loci. But the analysis of this so-called Hill-Robertson effect is often restricted to population genetics models with fixed population size. The build up of linkage disequilibrium in some epidemiological models has been discussed in some simulation models \citep{althaus2005stochastic,fraser2005hiv}. Our model provides a theoretical framework to explore the effect of finite population size on multi-locus dynamics of pathogens and to generate more accurate predictions on the evolution of drug resistance \citep{day2012evolutionary}. 

Finally, although we have presented our results in the context of pathogen evolution, it is hopefully clear that a very similar theoretical framework could be used to study other examples of life history evolution in the context of demographic stochasticity. Current general life history theory largely neglects the evolutionary consequences of stochasticity arising from small population sizes. Our results suggest that it would be profitable to determine what sorts of insights might be gained for life history evolution more generally by using the type of theoretical framework developed here.

\bigskip

\noindent \textbf{Acknowledgements:} Some of this work was done while TLP was supported by a Fondation Sciences Math\'ematiques de Paris postdoctoral fellowship.  AL thanks the Center for Interdisciplinary Research in Biology (Coll\`ege de France) for funding.  SG thanks the CNRS (PICS and PEPS MPI) for funding and Gauthier Boaglio for his help in the development of the simulation code. Simulations were performed on the Montpellier Bioinformatics Biodiversity cluster.

\clearpage

\putbib

\end{bibunit}

\clearpage

\appendix 

\begin{bibunit}

\setcounter{figure}{0}    
\setcounter{equation}{0}    
\setcounter{page}{1}    
\setcounter{table}{0}

\makeatletter
\renewcommand{\theequation}{S.\@arabic\c@equation}
\renewcommand{\thesection}{\thepart \arabic{section}}
\makeatother

\begin{center}{\Large \textsc{Supplementary Information}\\
\textbf{Pathogen evolution: slow and steady spreads the best}}

\bigskip

Todd L. Parsons, Amaury Lambert, Troy Day and Sylvain Gandon

\bigskip

\today

\end{center}

\begin{spacing}{0.9}
%\maketitle
\tableofcontents

\clearpage

\section*{Glossary of Notation}

\begin{longtable}{L@{\extracolsep{10mm}} l @{\extracolsep{10mm}} C}
	n & system size\\
%	d & number of strains\\
	\lambda^{(n)} n & immigration rate for susceptible individuals; $\lambda^{(n)} = \lambda +\BigO{\frac{1}{n}}$ \\
	\delta^{(n)} & base mortality rate; $\delta^{(n)} = \delta +\BigO{\frac{1}{n}}$\\
	\beta_{i}^{(n)} &  contact rate for strain $i$; $\beta_{i}^{(n)} = \beta_{i} +\BigO{\frac{1}{n}}$\\
	\alpha_{i}^{(n)} & excess mortality for strain $i$; $\alpha_{i}^{(n)} = \alpha_{i} +\BigO{\frac{1}{n}}$\\
	\gamma_{i}^{(n)} & recovery rate for strain $i$; $\gamma_{i}^{(n)} = \gamma_{i} +\BigO{\frac{1}{n}}$\\
	R_{0,i} & $\frac{\beta_{i}}{\delta+\alpha_{i}+\gamma_{i}}$\\
	R^{(n)}_{0,i} & $\frac{\beta^{(n)}_{i}}{\delta^{(n)}+\alpha^{(n)}_{i}+\gamma^{(n)}_{i}} = 
 	R_{0,i}\left(1+\frac{r_{i}}{n}\right) + {\textstyle o\left(\frac{1}{n}\right)}$\\
	S^{(n)}(t) & number of susceptible individuals at time $t$\\
	 I^{(n)}_{i}(t) & number of individuals infected with strain $i$ at time $t$\\
	 R^{(n)}(t) & number of recovered individuals at time $t$\\
	 N^{(n)}(t) & total number of individuals at time $t$\\
	 P^{(n)}_{i}(t) & frequency of individuals infected with strain $i$ at time $t$\\
	 \bm{I}^{(n)}(t) &  $(I^{(n)}_{1}(t),\ldots,I^{(n)}_{d}(t),N^{(n)}(t))$\\
	 \bm{E}^{(n)}(t) & $(S^{(n)}(t),I^{(n)}_{1}(t),\ldots,I^{(n)}_{d}(t))$\\
	 \bar{S}^{(n)}(t) & density of susceptible individuals at time $t$\\
	 \bar{I}^{(n)}_{i}(t) & density of individuals infected with strain $i$ at time $t$\\
	 \bar{R}^{(n)}(t) & density of recovered individuals at time $t$\\
	 \bar{N}^{(n)}(t) & total density of individuals at time $t$\\
	 \bar{\bm{I}}^{(n)}(t) & $(\bar{I}^{(n)}_{i}(t),\ldots,\bar{I}^{(n)}_{d}(t))$\\
	 \bar{\bm{E}}^{(n)}(t) & $(\bar{S}^{(n)}(t),\bar{I}^{(n)}_{i}(t),\ldots,\bar{I}^{(n)}_{d}(t),\bar{N}^{(n)}(t))$\\
	 \hat{S}^{(n)}(t) & density of susceptible individuals at time $nt$\\
	 \hat{I}^{(n)}_{i}(t) & density of individuals infected with strain $i$ at time $nt$\\
	 \hat{R}^{(n)}(t) & density of recovered individuals at time $nt$\\
	 \hat{N}^{(n)}(t) & total density of individuals at time $nt$\\
	 \hat{\bm{I}}^{(n)}(t) & $(\hat{I}^{(n)}_{1}(t),\ldots,\hat{I}^{(n)}_{d}(t))$\\
	 \hat{\bm{E}}^{(n)}(t) & $(\hat{S}^{(n)}(t),\hat{I}^{(n)}_{1}(t),\ldots,\hat{I}^{(n)}_{d}(t),\hat{N}^{(n)}(t))$\\
	 S(t) & asymptotic density of susceptible individuals at time $t$\\
	 I_{i}(t) & asymptotic density of individuals infected with strain $i$ at time $t$\\
	R(t) & asymptotic density of recovered individuals at time $t$\\
	 N(t) & total asymptotic density of individuals at time $t$\\
	 \bm{E}^{\star,i} & endemic equilibrium with resident strain $i$\\ 
	 \bm{I}(t) & $(I_{1}(t),\ldots,I_{d}(t))$\\ 	
	 \bm{E}(t) & $(S(t),I_{1}(t),\ldots,I_{d}(t),N(t))$\\
	 \hat{S}(t) & asymptotic density of susceptible individuals in slow time limit\\
	 \hat{I}_{i}(t) & asymptotic density of individuals infected with strain $i$ in slow time limit\\
	 \hat{R}(t) & asymptotic  density of recovered individuals in slow time limit\\
	 \hat{N}(t) & asymptotic total density of individuals in slow time limit\\
	 \hat{\bm{I}}(t) & $(\hat{I}_{1}(t),\ldots,\hat{I}_{d}(t))$\\
	 \hat{\bm{E}}(t) & $(\hat{S}(t),\hat{I}_{1}(t),\ldots,\hat{I}_{d}(t),\hat{N}(t))$\\
	 P_{i}(t) & asymptotic frequency of individuals infected with strain $i$ in slow time limit\\
\end{longtable}
\end{spacing}

\section{Introduction}

In this SI, we derive the results in the main text.  Where suitable references exist in the literature, we keep the discussion informal, sketching how the results are obtained and referring to the appropriate references for rigorous proofs.  Where they do not, we first give a heuristic derivation for a broader audience, whilst deferring the proofs to the end.

\section{A Stochastic Epidemiological Model with Multiple Pathogen Strains}

We consider a family of random processes $\left(S^{(n)}(t),I^{(n)}_{1}(t),\ldots,I^{(n)}_{d}(t),R^{(n)}(t)\right)$, indexed by a parameter $n$, the ``system size'' \citep{vanKampen1992}, which plays a role similar to the census population size in population genetics (see \eg \citet{Ewens1979,Durrett2009,Etheridge2011}). Similarly to those fixed-population models, we will consider the asymptotic behaviour of our model when $n$ is large.  $S^{(n)}(t)$, $I^{(n)}_{1}(t),\ldots,I^{(n)}_{d}(t)$, and $R^{(n)}(t)$ are the number of susceptible individuals, individuals infected with strain $i=1,\ldots,d$, and recovered individuals, respectively.   We will write $N^{(n)}(t)$ for the total population size at time $t$, so that 
\[
	N^{(n)}(t)  = S^{(n)}(t) + I^{(n)}_{1}(t) + \cdots + I^{(n)}_{d}(t) + R^{(n)}(t).
\]

Using this notation, our compartmental model for the epidemic is represented graphically in Figure \ref{COMPARTMENT}.

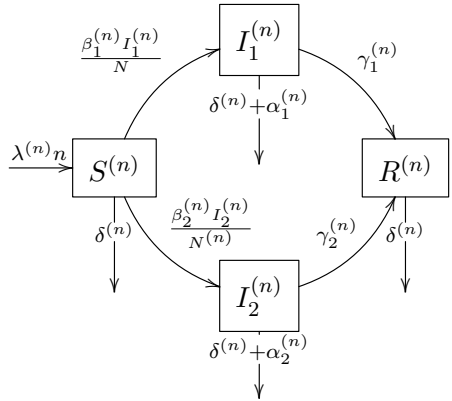
\begin{figure}[h]
\[
	\xymatrix{& & *++[F]{I^{(n)}_{1}} \ar@/^1pc/[dr]^{\gamma^{(n)}_{1}}  \ar[d]|{\delta^{(n)}+\alpha^{(n)}_{1}} & \\
	\ar[r]^{\hspace{-.5cm}\lambda^{(n)} n} & *++[F]{S^{(n)}} \ar@/^1pc/[ur]^{\frac{\beta^{(n)}_{1} I^{(n)}_{1}}{N}} \ar@/_1pc/[dr]^{\frac{\beta^{(n)}_{2} I^{(n)}_{2}}{N^{(n)}}} 
	\ar[d]|{\delta^{(n)}}  & & *++[F]{R^{(n)}}  \ar[d]|{\delta^{(n)}} \\
	& & *++[F]{I^{(n)}_{2}} \ar@/_1pc/[ur]^{\hspace{2cm}\gamma^{(n)}_{2}}  \ar[d]|{\delta^{(n)}+\alpha^{(n)}_{2}} & \\	 
	& & &}
\]
\caption[Compartmental model of a two-strain SIR epidemic]{Compartmental model of a two-strain SIR epidemic.  Arrows indicate transitions between states and are labelled with the corresponding transition rate.  Arrows into empty space indicate deaths.}
\label{COMPARTMENT}
\end{figure}

Equivalently, we may describe our model as a continuous-time Markov chain $\bm{E}$ taking values in $\mathbb{N}^{d+2}$ with transition rates given in Table \ref{RATES}. When a transition occurs at time $t$, we will distinguish between the value $\bm{E}(t-)$ of the Markov chain before the transition and its value $\bm{E}(t)$ after the transition.

All parameters given in Table \ref{RATES} may depend on $n$, but are assumed to have a constant value to first approximation in $n$
\begin{equation}\label{ASSUMPTIONS}
	\alpha^{(n)}_{i} = \alpha_{i} + \BigO{\frac{1}{n}}, 
	\quad \beta^{(n)}_{i} = \beta_{i}+ \BigO{\frac{1}{n}}, 
	\quad \gamma^{(n)}_{i} = \gamma_{i}+ \BigO{\frac{1}{n}},
	\quad \delta^{(n)} = \delta + \BigO{\frac{1}{n}}, 
	\quad \lambda^{(n)} = \lambda + \BigO{\frac{1}{n}}
\end{equation}

Simple calculations using the master equation tell us that in the absence of infected individuals,
the expected value of $N^{(n)}(t)$ is
\[ 	
	\mathbb{E}\left[N^{(n)}(t)\right] = e^{-\delta^{(n)} t} N^{(n)}(0)
		+ \frac{\lambda^{(n)}}{\delta^{(n)}} \left(1-e^{-\delta^{(n)} t}\right) n,
\]
which approaches an equilibrium value of $\frac{\lambda^{(n)}}{\delta^{(n)}} n$ as $t \to \infty$.  Thus, to first approximation, the total population size is proportional to $n$. %if we assume that the arrival rate of new susceptibles is equal to the death rate of uninfected individuals, then the expected number of susceptibles in a disease-free population (and thus the total population size) will eventually equilibrate at $n$.  

\begin{table}
\begin{tabular}{L@{\extracolsep{10mm}}C}
\text{Transition} & \text{Rate}\\
\hline
S^{(n)}(t-) \to S^{(n)}(t) = S^{(n)}(t-)+1 & \lambda^{(n)} n\\
S^{(n)}(t-) \to S^{(n)}(t) = S^{(n)}(t-)-1 & \delta^{(n)} S^{(n)}(t-)\\
S^{(n)}(t-) \to S^{(n)}(t) = S^{(n)}(t-)-1, I^{(n)}_{i}(t-) \to I^{(n)}_{i}(t) = I^{(n)}_{i}(t) + 1 & \frac{\beta^{(n)}_{i} S^{(n)}(t-) I^{(n)}_{i}(t-) }{N^{(n)}(t-)}\\
I^{(n)}_{i}(t-) \to I^{(n)}_{i}(t) = I^{(n)}_{i}(t) - 1 & (\delta^{(n)}+\alpha^{(n)}_{i})I^{(n)}_{i}(t-)\\
I^{(n)}_{i}(t-) \to I^{(n)}_{i}(t) = I^{(n)}_{i}(t) - 1, R^{(n)}(t-) \to R^{(n)}(t) = R^{(n)}(t-)+1 & \gamma^{(n)}_{i} I^{(n)}_{i}(t-)\\
R^{(n)}(t-) \to R^{(n)}(t) = R^{(n)}(t-)-1 & \delta^{(n)} R^{(n)}(t-)
\end{tabular}
\caption{Transition Rates}\label{RATES}
\end{table}

If one knows the values of $S^{(n)}(t)$ and $\bm{I}^{(n)}(t) \defn (I^{(n)}_{1}(t),\ldots,I^{(n)}_{d}(t))$, then given one of $R^{(n)}(t)$ or $N^{(n)}(t)$ one can determine the other.  For our purposes, it is more convenient to track the total population size,  and consider the epidemic 
\[
	\bm{E}^{(n)}(t) \defn \left(S^{(n)}(t),\bm{I}^{(n)}(t),N^{(n)}(t)\right).
\]

In what follows, rather than working with $\bm{E}^{(n)}(t)$ we will focus on the rescaled process 
\[
	\bar{S}^{(n)}(t) \defn \frac{1}{n} S^{(n)}(t), 
	\quad \bar{I}^{(n)}_{i}(t) \defn \frac{1}{n} I^{(n)}_{i}(t), 
	\quad \bar{N}^{(n)}(t) \defn \frac{1}{n} N^{(n)}(t),
\]
and
\[
	\bar{\bm{E}}^{(n)}(t) \defn \left(\bar{S}^{(n)}(t),\bar{\bm{I}}^{(n)}(t),\bar{N}^{(n)}(t)\right).
\]
 $\bar{\bm{E}}^{(n)}(t)$ has the advantage of being a density dependent population process \citep{Kurtz1970,Kurtz1971,Kurtz1978,Kurtz1981} as generalized in \citet{Pollett1990}: the transition rates in \eqref{RATES} depend only on the densities $\bar{S}^{(n)}(t),\bar{\bm{I}}^{(n)}(t),\bar{N}^{(n)}(t)$ and not on the absolute numbers of individuals.  As we discuss below, density dependent population processes have a number of nice features, including a law of large numbers and central limit theorems. 
 
 \begin{rem} 
To simplify our subsequent use of subscripts, we will consider $\bar{\bm{E}}^{(n)}(t)$ as a process taking values in $\mathbb{R}^{d+2}$, the space of points 
\[
	\bm{x} = (x_{0},x_{1},\ldots,x_{d},x_{d+1}),
\]
and use $\bar{S}^{(n)}(t)$ and $\bar{E}^{(n)}_{0}(t)$, \etc interchangeably. 
\end{rem}

\section{Stochastic Differential Equation Formulation}

Here, we introduce a very convenient way of writing our Markov chain as the solution to a stochastic integral equation with the help of simple Poisson processes. 

A Poisson process $P$ is a Markov process making jumps of $+1$ exclusively, and such that $P(0)=0$. A Poisson process $P$ is a called a simple Poisson process if it jumps at constant rate 1. In this case, $(P(at))$ is a Poisson process with rate $a$. This can be generalized by noting that $(P(\int_0^t a(s)\, ds))$ is a time-inhomogeneous Poisson process which jumps at rate $a(t)$ at time $t$. Similarly, there is a unique continuous-time Markov chain $X$ satisfying
\[
X(t) = x_0 + P\left(\int_0^t f(X(s-)) \, ds\right)
\]
and when $X(t-)=x$, $X$ jumps to $x+1$ at rate $f(x)$.

Then it is not difficult to extend this (see Chapter 6, \S 4 in \citet{Ethier+Kurtz86} for details) to our Markov process as follows: 
\begin{align*}
 	S^{(n)}(t) &= \begin{multlined}[t] S^{(n)}(0) + P_{\bm{e}_{0} + \bm{e}_{d+1}}(n\lambda^{(n)} t) \\
		- P_{-\bm{e}_{0} - \bm{e}_{d+1}}\left(\int_{0}^{t} \delta^{(n)} S^{(n)}(s)\, ds\right) 
		- \sum_{i=1}^{d} P_{-\bm{e}_{0} + \bm{e}_{i}}\left(\int_{0}^{t} \frac{\beta^{(n)}_{i} S^{(n)}(s) I^{(n)}_{i}(s) }{N^{(n)}(s)}\, ds\right) \end{multlined}\\
	I^{(n)}_{i}(t) &= \begin{multlined}[t] I^{(n)}_{i}(0) + P_{-\bm{e}_{0} + \bm{e}_{i}}\left(\int_{0}^{t} \frac{\beta^{(n)}_{i} S^{(n)}(s) I^{(n)}_{i}(s) }{N^{(n)}(s)}\, ds\right) \\
		- P_{-\bm{e}_{i} - \bm{e}_{d+1}}\left(\int_{0}^{t}  (\delta^{(n)} + \alpha^{(n)}_{i}) I^{(n)}_{i}(s)\, ds\right)
			 - P_{-\bm{e}_{i}}\left(\int_{0}^{t}  \gamma^{(n)}_{i} I^{(n)}_{i}(s)\, ds\right) \end{multlined}\\
	N^{(n)}(t) &= \begin{multlined}[t] N^{(n)}(0) + P_{\bm{e}_{0} + \bm{e}_{d+1}}(n\lambda^{(n)} t) - P_{-\bm{e}_{0} - \bm{e}_{d+1}}\left(\int_{0}^{t} \delta^{(n)} S^{(n)}(s)\, ds\right) \\
		- P_{-\bm{e}_{d+1}}\left(\int_{0}^{t} \delta^{(n)} \left(N^{(n)}(s) - \sum_{i=1}^{d} I^{(n)}_{i}(s) - S^{(n)}(s)\right)\, ds\right)\\
		- \sum_{i=1}^{d} P_{-\bm{e}_{i} - \bm{e}_{d+1}}\left(\int_{0}^{t} (\delta^{(n)} + \alpha^{(n)}_{i}) I^{(n)}_{i}(s)\, ds\right). \end{multlined}
\end{align*}
where all the processes $P_{\bm{l}}(t)$ are independent, simple Poisson processes, indexed by the corresponding jump $\bm{l}$ of the Markov process $(\bm{E}^{(n)}(t))$ and $\bm{e}_{i}$ is the element of $\mathbb{R}^{d+2}$ with zeros everywhere except a 1 at row $i$.

Changing variables, we get
\begin{align*}
 	\bar{S}^{(n)}(t) &= \begin{multlined}[t] \bar{S}^{(n)}(0) 
	+ \frac{1}{n} P_{\bm{e}_{0} + \bm{e}_{d+1}}(n\lambda^{(n)} t) \\
	- \frac{1}{n} P_{-\bm{e}_{0} - \bm{e}_{d+1}}\left(n \int_{0}^{t} \delta^{(n)} \bar{S}^{(n)}(s)\, ds\right) 
	- \sum_{i=1}^{d} \frac{1}{n} P_{-\bm{e}_{0} + \bm{e}_{i}}\left(n \int_{0}^{t} 
		\frac{\beta^{(n)}_{i} \bar{S}^{(n)}(s) \bar{I}^{(n)}_{i}(s) }{\bar{N}^{(n)}(s)}\, ds\right) 
	\end{multlined}\\
	\bar{I}^{(n)}_{i}(t) &= \begin{multlined}[t] \bar{I}^{(n)}_{i}(0) 
	+ \frac{1}{n} P_{-\bm{e}_{0} + \bm{e}_{i}}\left(n \int_{0}^{t} \frac{\beta^{(n)}_{i} \bar{S}^{(n)}(s) \bar{I}^{(n)}_{i}(s) }{\bar{N}^{(n)}(s)}\, ds\right) \\
		- \frac{1}{n} P_{-\bm{e}_{i} - \bm{e}_{d+1}}\left(n \int_{0}^{t}  (\delta^{(n)} + \alpha^{(n)}_{i}) \bar{I}^{(n)}_{i}(s)\, ds\right)
			 - \frac{1}{n} P_{-\bm{e}_{i}}\left(n \int_{0}^{t}  \gamma^{(n)}_{i} \bar{I}^{(n)}_{i}(s)\, ds\right) \end{multlined}\\
	\bar{N}^{(n)}(t) &= \begin{multlined}[t] \bar{N}^{(n)}(0) + \frac{1}{n} P_{\bm{e}_{0} + \bm{e}_{d+1}}(n\lambda^{(n)} t) - \frac{1}{n} P_{-\bm{e}_{0} - \bm{e}_{d+1}}\left(n \int_{0}^{t} \delta^{(n)} \bar{S}^{(n)}(s)\, ds\right) \\
		- \frac{1}{n} P_{-\bm{e}_{d+1}}\left(n \int_{0}^{t} \delta^{(n)} \left(\bar{N}^{(n)}(s) - \sum_{i=1}^{d} \bar{I}^{(n)}_{i}(s) - \bar{S}^{(n)}(s)\right)\, ds\right)\\
		- \sum_{i=1}^{d} \frac{1}{n} P_{-\bm{e}_{i} - \bm{e}_{d+1}}\left(n \int_{0}^{t} (\delta^{(n)} + \alpha^{(n)}_{i}) \bar{I}^{(n)}_{i}(s)\, ds\right). \end{multlined}
\end{align*}
 %This representation is simply a formalisation of the Gillespie algorithm commonly used to simulate Markov chains.  In the Gillespie algorithm, one samples exponential random variables with the jump rates of the Markov process to obtain the time until the next event; those exponential times are simply the inter-event times for a Poisson process.  Intuitively, the Poisson processes are event clocks that ring (\ie return the value 1) when an event with the corresponding rate has occurred.

This formalism is useful because it will allow us to write each r.h.s. as the sum of a deterministic trend and of a stochastic term with zero expectation. 

Recall that the marginal value $P(t)$ of a simple Poisson process at time $t$ is a Poisson random variable with parameter $t$. In particular, $P(t)-t$ has mean 0 and variance $t$. So if we write 
\[
	\tilde{P}(t) \defn P(t) - t,
\]
we are writing $P(t)$ as the sum of a deterministic trend $t$ and of a stochastic term $\tilde{P}(t)$ with mean 0. If we come back to the example of the Markov process $X$ jumping at rate $f(X)$, we can write
\[
X(t) = x_0 + \int_0^t f(X(s-)) \, ds +M(t),
\]
where we have set 
\[
M(t):= \tilde{P}\left(\int_0^t f(X(s-)) \, ds\right).
\]
In addition, since the increments $\tilde{P}(t+s)-\tilde{P}(t)$ are independent of the past before $t$, have mean 0 and variance $s$, we can write the last equation in differential form
\[
dX(t) =  f(X(t-)) \, dt + dM(t), 
\]
with $dM(t) = U(t) - f(X(t-))\, dt$, where $U(t)$ equals 1 iff $P$ jumps at $\int_0^t f(X(s-)) \, ds$ and equals 0 otherwise. In particular, conditional on $X(t-)=x$, $dM(t)$ has mean 0 and variance $f(x)\,dt$. 
Thus, we also recover the infinitesimal variation of $X$ as the sum of an infinitesimal trend in the dynamics and of a stochastic fluctuation term with zero expectation.

Now let us return to our initial process. We adopt the same notation as previously, for example $\tilde{P}_{-\bm{e}_{0} - \bm{e}_{d+1}}(t) = P_{-\bm{e}_{0} - \bm{e}_{d+1}}(t) - t$ and 
\[
	M^{(n)}_{-\bm{e}_{0} - \bm{e}_{d+1}}(t) \defn \tilde{P}_{-\bm{e}_{0} - \bm{e}_{d+1}}\left(n \int_{0}^{t} \delta^{(n)} \bar{S}^{(n)}(s)\, ds\right).
\]

\begin{prop}
\label{prop:SDE}
The infinitesimal variation of $\bar{\bm{E}}^{(n)}(t)$ can be written as the sum of an infinitesimal deterministic trend and of a stochastic fluctuation term with zero expectation:
\[
d\bar{S}^{(n)}(t)  = {F}^{(n)}_{0}\left(\bar{\bm{E}}^{(n)}(t)\right) dt + \frac{1}{n} dM^{(n)}_{\bm{e}_{0} + \bm{e}_{d+1}}(t) -\frac{1}{n} dM^{(n)}_{-\bm{e}_{0} - \bm{e}_{d+1}}(t)
-\frac{1}{n}  \sum_{i=1}^{d} dM^{(n)}_{-\bm{e}_{0} + \bm{e}_{i}}(t)
\]
\[
d\bar{I}^{(n)}_{i}(t) =  {F}^{(n)}_{i}\left(\bar{\bm{E}}^{(n)}(t)\right) dt + \frac{1}{n} dM^{(n)}_{-\bm{e}_{0} + \bm{e}_{i}}(t)-\frac{1}{n} dM^{(n)}_{-\bm{e}_{i} - \bm{e}_{d+1}}(t)-\frac{1}{n} dM^{(n)}_{-\bm{e}_{i}}(t)
\]
\[
d\bar{N}^{(n)}(t) = {F}^{(n)}_{d+1}\left(\bar{\bm{E}}^{(n)}(t)\right) dt + \frac{1}{n}dM^{(n)}_{\bm{e}_{0} + \bm{e}_{d+1}}(t)
	- \frac{1}{n}dM^{(n)}_{-\bm{e}_{0} - \bm{e}_{d+1}}(t)	-\frac{1}{n} \sum_{i=1}^{d}dM^{(n)}_{-\bm{e}_{i} - \bm{e}_{d+1}}(t)- \frac{1}{n} dM^{(n)}_{-\bm{e}_{d+1}}(t),
\]
%\begin{equation}\label{SDEE}
%\begin{multlined}
%	\bar{\bm{E}}^{(n)}(t) = \bar{\bm{E}}^{(n)}(0)
%	+ \int_{0}^{t} \bm{F}^{(n)}\left(\bar{\bm{E}}^{(n)}(s)\right)\, ds
%	+ \frac{1}{n} (\bm{e}_{0} + \bm{e}_{d+1})M^{(n)}_{\bm{e}_{0} + \bm{e}_{d+1}}(t)
%	- \frac{1}{n} (\bm{e}_{0} + \bm{e}_{d+1}) M^{(n)}_{-\bm{e}_{0} - \bm{e}_{d+1}}(t)\\
%	+ \frac{1}{n}  \sum_{i=1}^{d} (\bm{e}_{i} - \bm{e}_{0}) M^{(n)}_{-\bm{e}_{0} + \bm{e}_{i}}(t)
%	- \frac{1}{n} \sum_{i=1}^{d} (\bm{e}_{i} + \bm{e}_{d+1}) M^{(n)}_{-\bm{e}_{i} - \bm{e}_{d+1}}(t)
%	- \frac{1}{n} \sum_{i=1}^{d} \bm{e}_{i} M^{(n)}_{-\bm{e}_{i}}(t)
%	- \frac{1}{n} \bm{e}_{d+1} M^{(n)}_{-\bm{e}_{d+1}}(t),
%\end{multlined}
%\end{equation}
where
 \begin{align*}
	F^{(n)}_{0}(\bm{x}) &= \lambda^{(n)} - \left(\sum_{i=1}^{d} \beta^{(n)}_{i} \frac{x_{i}}{x_{d+1}}  + \delta^{(n)}\right) x_{0}\\
	F^{(n)}_{i}(\bm{x}) &= \left(\beta^{(n)}_{i}\frac{x_{0}}{x_{d+1}} - (\delta^{(n)}+\alpha^{(n)}_{i}+\gamma^{(n)}_{i})\right) x_{i}\\
	F^{(n)}_{d+1}(\bm{x}) &= \lambda^{(n)} - \delta^{(n)} x_{d+1} - \sum_{i=1}^{d} \alpha^{(n)}_{i} x_{i},
\end{align*}
and $dM^{(n)}_{\bm{e}_{0} + \bm{e}_{d+1}}(t)$, $dM^{(n)}_{-\bm{e}_{0} - \bm{e}_{d+1}}(t)$, $dM^{(n)}_{-\bm{e}_{0} + \bm{e}_{i}}(t)$, $dM^{(n)}_{-\bm{e}_{i} - \bm{e}_{d+1}}(t)$, $dM^{(n)}_{-\bm{e}_{i}}(t)$ and $dM^{(n)}_{-\bm{e}_{d+1}}(t)$, are independent infinitesimal noise terms with mean zero and respective infinitesimal variances\\ 
$n\rho^{(n)}_{\bm{e}_{0} + \bm{e}_{d+1}}\left(\bar{\bm{E}}^{(n)}(t)\right)\, dt$, 
$n\rho^{(n)}_{-\bm{e}_{0} - \bm{e}_{d+1}}\left(\bar{\bm{E}}^{(n)}(t)\right)\, dt$, 
$n\rho^{(n)}_{-\bm{e}_{0} + \bm{e}_{i}}\left(\bar{\bm{E}}^{(n)}(t)\right)\, dt$, 
$n\rho^{(n)}_{-\bm{e}_{i} - \bm{e}_{d+1}}\left(\bar{\bm{E}}^{(n)}(t)\right)\, dt$,\\
$n\rho^{(n)}_{-\bm{e}_{i}}\left(\bar{\bm{E}}^{(n)}(t)\right)\, dt$ and 
$n\rho^{(n)}_{-\bm{e}_{d+1}}\left(\bar{\bm{E}}^{(n)}(t)\right)\, dt$, where
 \begin{align*}
\rho^{(n)}_{\bm{e}_{0} + \bm{e}_{d+1}}(\bm{x})  &=\lambda^{(n)}\\
\rho^{(n)}_{-\bm{e}_{0} - \bm{e}_{d+1}}(\bm{x})&= \delta^{(n)} x_0\\
\rho^{(n)}_{-\bm{e}_{0} + \bm{e}_{i}}(\bm{x})  &=\beta^{(n)}_{i}\frac{x_{0}\, x_i}{x_{d+1}} \\
\rho^{(n)}_{-\bm{e}_{i} - \bm{e}_{d+1}}(\bm{x})  &=(\delta^{(n)} + \alpha^{(n)}_{i}) x_i\\
\rho^{(n)}_{-\bm{e}_{i}}(\bm{x})  &=\gamma^{(n)}_{i} x_i\\
\rho^{(n)}_{-\bm{e}_{d+1}}(\bm{x})  &=\delta^{(n)} \left(x_{d+1}- \sum_{i=0}^{d} x_i\right).
\end{align*}
\end{prop}

\begin{rem}
Note that $n\rho^{(n)}_{\bm{l}}\left(\bar{\bm{E}}^{(n)}(t)\right)$, is rate at which the Markov process  $(\bm{E}^{(n)}(t))$ makes a jump $\bm{l}$.
\end{rem}

Setting
\begin{multline}
\label{eqn:Mn}
	\bm{M}^{(n)}(t) := (\bm{e}_{0} + \bm{e}_{d+1})M^{(n)}_{\bm{e}_{0} + \bm{e}_{d+1}}(t)
	- (\bm{e}_{0} + \bm{e}_{d+1}) M^{(n)}_{-\bm{e}_{0} - \bm{e}_{d+1}}(t)
	+ \sum_{i=1}^{d} (\bm{e}_{i} - \bm{e}_{0}) M^{(n)}_{-\bm{e}_{0} + \bm{e}_{i}}(t)\\
	- \sum_{i=1}^{d} (\bm{e}_{i} + \bm{e}_{d+1}) M^{(n)}_{-\bm{e}_{i} - \bm{e}_{d+1}}(t)
	- \sum_{i=1}^{d} \bm{e}_{i} M^{(n)}_{-\bm{e}_{i}}(t)
	- \bm{e}_{d+1} M^{(n)}_{-\bm{e}_{d+1}}(t),
\end{multline}
the result in the proposition can be written more compactly as
\begin{equation}
\label{SDEE}
	d\bar{\bm{E}}^{(n)}(t) = %\bar{\bm{E}}^{(n)}(0)
	 \bm{F}^{(n)}\left(\bar{\bm{E}}^{(n)}(t)\right)\, dt + \frac{1}{n} d\bm{M}^{(n)}(t),
\end{equation}
where 
\[
	\bm{F}^{(n)}(\bm{x}) =(F^{(n)}_{0}(\bm{x}),F^{(n)}_{1}(\bm{x}),\ldots \\ F^{(n)}_{d}(\bm{x}),F^{(n)}_{d+1}(\bm{x})).
\]
Thus, the function $\bm{F}^{(n)}(\bm{x})$ describes the infinitesimal trend in the dynamics, whereas the terms $M^{(n)}_{\bm{l}}(t)$ capture the de-trended fluctuations corresponding to each type of possible event.  This equation is analogous to an It\^o SDE, only now the driving noise is the discontinuous $\bm{M}^{(n)}(t)$, rather than the more familiar Brownian motion. 

We note that for $i=1,\ldots,d$,
\[
	M^{(n)}_{i}(t) = M^{(n)}_{-\bm{e}_{0} + \bm{e}_{i}}(t) - M^{(n)}_{-\bm{e}_{i} - \bm{e}_{d+1}}(t)
		- M^{(n)}_{-\bm{e}_{i}}(t),		
\]
so that $M^{(n)}_{i}(t)$ is independent of $M^{(n)}_{j}(t)$ for all $1 \leq i \neq j \leq d$, whereas
\begin{equation}\label{QVMi}
	\mathbb{E}\left[dM^{(n)}_{i}(t)^{2}\right] = \left(\beta^{(n)}_{i}\frac{S^{(n)}(t)}{N^{(n)}(t)} 
	+ (\delta^{(n)}+\alpha^{(n)}_{i}+\gamma^{(n)}_{i})\right) I^{(n)}_{i}(t),
\end{equation}
two facts that will prove useful in what follows.
 
We define similarly
\[
\rho_1^{(n)}:=\rho^{(n)}_{-\bm{e}_{0} - \bm{e}_{d+1}}, \quad \rho^{(n)}_2:=\rho^{(n)}_{\bm{e}_{0} + \bm{e}_{d+1}}, \quad \rho^{(n)}_D:=\rho^{(n)}_{-\bm{e}_{d+1}} 
\] 
and for $j=1,\ldots, d$
\[
\rho_{3j}^{(n)}:=\rho_{S,j}^{(n)}, \quad \rho_{3j+1}^{(n)}:=\rho_{j,-}^{(n)}, \quad \rho_{3j+2}^{(n)}:=\rho_{j,R}{(n)}. 
\] 

%From the last set of equations, it is easy to see that 
In particular, it makes sense to define $\bm{a}^{(n)}$ as the infinitesimal variance-covariance matrix of $\bar{\bm{E}}^{(n)}(t)$ by
\begin{equation}\label{A}
	 a^{(n)}_{ij}\left(\bar{\bm{E}}^{(n)}(t)\right) \,dt := 
	 	\text{\rm Cov}\left[d\bar{E}^{(n)}_{i}(t), d\bar{E}^{(n)}_{j}(t) \right]= \frac{1}{n^2}\text{\rm Cov}\left[dM^{(n)}_{i}(t), d{M}^{(n)}_{j}(t)\right]. %\middle\vert	\bar{\bm{E}}^{(n)}(t-) = \bm{x} \right],
\end{equation}
Let us compute $\bm{a}^{(n)}$. Because all distinct terms in the definition \eqref{eqn:Mn} of $M^{(n)}$ are independent, all cross terms vanish. For example, for any $1\le i\le d$
\begin{align*}
\text{\rm Cov}&\left[d\bar{E}^{(n)}_{0}(t),d\bar{E}^{(n)}_{i}(t)\right] \\
	&= \begin{multlined} \frac{1}{n^2}\text{\rm Cov}\left[dM^{(n)}_{\bm{e}_{0} + \bm{e}_{d+1}}(t) - dM^{(n)}_{-\bm{e}_{0} - \bm{e}_{d+1}}(t) -  \sum_{j=1}^{d} dM^{(n)}_{S,j}(t),\right.\\
\left. dM^{(n)}_{-\bm{e}_{0} + \bm{e}_{i}}(t)-dM^{(n)}_{-\bm{e}_{i} - \bm{e}_{d+1}}(t)- dM^{(n)}_{-\bm{e}_{i}}(t)\right]
	\end{multlined}\\
	&= \frac{1}{n^2}\text{\rm Cov}\left[
-dM^{(n)}_{-\bm{e}_{0} + \bm{e}_{i}}(t), dM^{(n)}_{-\bm{e}_{0} + \bm{e}_{i}}(t)\right]\\
&= -\frac{1}{n^2}\text{\rm Var}\left[dM^{(n)}_{-\bm{e}_{0} + \bm{e}_{i}}(t)\right] = -\frac{1}{n}\rho^{(n)}_{-\bm{e}_{0} + \bm{e}_{i}}(\bar{\bm{E}}^{(n)}(t))\, dt ,
\end{align*}
which shows that $a_{0i}^{(n)}(\bm{x}) =  -\frac{1}{n}\rho^{(n)}_{-\bm{e}_{0} + \bm{e}_{i}}(\bm{x})$.
We should not be surprised by the fact that this infinitesimal covariance is negative. Each event where a susceptible is infected by strain $i$ has the effect of simultaneously decreasing the number of susceptibles and increasing the number of individuals infected by strain $i$.

By similar calculations, it is easy to see that 
$$
a_{0,d+1}^{(n)}(\bm{x}) = \frac{1}{n}\rho^{(n)}_{\bm{e}_{0} + \bm{e}_{d+1}}\left(\bm{x}\right) +\frac{1}{n}\rho^{(n)}_{-\bm{e}_{0} - \bm{e}_{d+1}}\left(\bm{x}\right),
$$
$$
a_{d+1,d+1}^{(n)}(\bm{x}) = \frac{1}{n}\rho^{(n)}_{\bm{e}_{0} + \bm{e}_{d+1}}\left(\bm{x}\right) +\frac{1}{n}\rho^{(n)}_{-\bm{e}_{0} - \bm{e}_{d+1}}\left(\bm{x}\right)+\frac{1}{n}\sum_{i=1}^d\rho^{(n)}_{-\bm{e}_{i} - \bm{e}_{d+1}}\left(\bm{x}\right)+ \frac{1}{n}\rho^{(n)}_{-\bm{e}_{d+1}}\left(\bm{x}\right),
$$
and for any $1\le i\le d$,
$$
a_{i,d+1}^{(n)}(\bm{x}) =\frac{1}{n}\rho^{(n)}_{-\bm{e}_{i} - \bm{e}_{d+1}}\left(\bm{x}\right).
$$
Also for any $1\le i\not=j\le d$, $a_{ij}^{(n)}(\bm{x}) =0$
%\[
%\text{\rm Cov}\left[d\bar{E}^{(n)}_{i}(t), d\bar{E}^{(n)}_{j}(t)\right]=0
%\]
and 
\begin{align*}
a_{jj}^{(n)}(\bm{x}) = \frac{1}{n}\rho^{(n)}_{S,j}\left(\bm{x}\right) +\frac{1}{n}\rho^{(n)}_{j,-}\left(\bm{x}\right) + \frac{1}{n}\rho^{(n)}_{j,R}\left(\bm{x}\right).
%\text{\rm Cov}\left[d\bar{E}^{(n)}_{j}(t), d\bar{E}^{(n)}_{j}(t)\right] &= \frac{1}{n^2}\text{\rm Var}\left[dM^{(n)}_{S,j}(t)\right] +\frac{1}{n^2}\text{\rm Var}\left[dM^{(n)}_{j,-}(t)\right] +\frac{1}{n^2}\text{\rm Var}\left[dM^{(n)}_{j,R}(t)\right] \\
%	&= \frac{1}{n}\rho^{(n)}_{S,j}\left(\bar{\bm{E}}^{(n)}(t)\right)\, dt +\frac{1}{n}\rho^{(n)}_{j,-}\left(\bar{\bm{E}}^{(n)}(t)\right)\, dt + \frac{1}{n}\rho^{(n)}_{j,R}\left(\bar{\bm{E}}^{(n)}(t)\right)\, dt.
\end{align*}
In other words, for any $0\le i,j\le d+1$,
\begin{equation}\label{AA}
	a^{(n)}_{ij}(\bm{x}) = \begin{cases}
	\frac{1}{n} \left(\frac{\beta_{i}^{(n)} x_{0}x_{j}}{x_{d+1}} + (\delta^{(n)} + \alpha_{j}^{(n)} + \gamma_{j}^{(n)})x_{j} \right)
		& \text{if $i=j$, }\\
	0 & \text{otherwise.}
	\end{cases}
%a^{(n)}_{ij}(\bm{x}) = \frac 1n \sum_{k=1}^D\sigma^{(n)}_{ik}(\bm{x})\sigma^{(n)}_{jk}(\bm{x}).
\end{equation}
Equivalently,
\begin{equation}\label{AAA}
	\bm{a}^{(n)}(\bm{x}) = \frac{1}{n} \sum_{\bm{l}} \bm{l}^{\top}\bm{l} \rho_{\bm{l}}(\bm{x}),
\end{equation}
where the sum is over all possible jumps $\bm{l}$ of the Markov process $(\bm{E}^{(n)}(t))$.

\subsection{Obtaining Equation (3)}

Note that similarly to Brownian motion, the infinitesimal mean of $\frac 1{\sqrt{n}} M^{(n)}_{\bm{l}}(t)$ during the time interval $dt$ is zero and its infinitesimal variance is  $\rho^{(n)}_{\bm{l}} \left(\tilde{\bm{E}}^{(n)}(t)\right) dt$., whereas the jump size $\frac 1{\sqrt{n}}$ tends to 0 as $n \to \infty$ (and thus, $\frac 1{\sqrt{n}} M^{(n)}_{\bm{l}}(t)$ is approximately continuous for large $n$), we see that  for large values of $n$, this noise is approximately equal to a Brownian motion with the same variance:
\[
	\frac 1{\sqrt{n}} M^{(n)}_{\bm{l}}(t) 
	\approx \sqrt{\rho^{(n)}_{\bm{l}} \left(\tilde{\bm{E}}^{(n)}(t)\right)}\, dB_{\bm{l}}(t),
\]
where all Brownian motions $B_{\bm{e}_{0} + \bm{e}_{d+1}}$, $B_{-\bm{e}_{0} - \bm{e}_{d+1}}$, $B_{-\bm{e}_{0} + \bm{e}_{i}}$, $B_{-\bm{e}_{i} - \bm{e}_{d+1}}$, $B_{-\bm{e}_{i}}$ and $B_{-\bm{e}_{d+1}}$ are independent.

This allows us to rewrite the results in Proposition \ref{prop:SDE} in the form of the following diffusion approximation for $n$ large, 

%\textcolor{blue}{We may want to say here that this approximation holds for $n$ is large, don't we?} 
\begin{multline*}
%\begin{medsize} 
d\bar{S}^{(n)}(t)  \approx {F}^{(n)}_{0}\left(\bar{\bm{E}}^{(n)}(t)\right) dt + \frac{1}{\sqrt{n}} \sqrt{\rho^{(n)}_{\bm{e}_{0} + \bm{e}_{d+1}}\left(\bar{\bm{E}}^{(n)}(t)\right)}\,dB_{\bm{e}_{0} + \bm{e}_{d+1}}(t) \\-\frac{1}{\sqrt{n}}\sqrt{\rho^{(n)}_{-\bm{e}_{0} - \bm{e}_{d+1}}\left(\bar{\bm{E}}^{(n)}(t)\right)} \, dB_{-\bm{e}_{0} - \bm{e}_{d+1}}(t) -\frac{1}{\sqrt{n}}  \sum_{i=1}^{d} \sqrt{\rho^{(n)}_{-\bm{e}_{0} + \bm{e}_{i}}\left(\bar{\bm{E}}^{(n)}(t)\right)}\, dB_{-\bm{e}_{0} + \bm{e}_{i}}(t)
\end{multline*}
\begin{multline*}
d\bar{I}^{(n)}_{i}(t) \approx  {F}^{(n)}_{i}\left(\bar{\bm{E}}^{(n)}(t)\right) dt +  \frac{1}{\sqrt{n}} \sqrt{\rho^{(n)}_{-\bm{e}_{0} + \bm{e}_{i}}\left(\bar{\bm{E}}^{(n)}(t)\right)} dB_{-\bm{e}_{0} + \bm{e}_{i}}(t)\\- \frac{1}{\sqrt{n}} \sqrt{\rho^{(n)}_{-\bm{e}_{i} - \bm{e}_{d+1}}\left(\bar{\bm{E}}^{(n)}(t)\right)} dB_{-\bm{e}_{i} - \bm{e}_{d+1}}(t)- \frac{1}{\sqrt{n}} \sqrt{\rho^{(n)}_{-\bm{e}_{i}}\left(\bar{\bm{E}}^{(n)}(t)\right)} dB_{-\bm{e}_{i}}(t)
\end{multline*}
\begin{multline*}
d\bar{N}^{(n)}(t) \approx {F}^{(n)}_{d+1}\left(\bar{\bm{E}}^{(n)}(t)\right) dt +  \frac{1}{\sqrt{n}} \sqrt{\rho^{(n)}_{\bm{e}_{0} + \bm{e}_{d+1}}\left(\bar{\bm{E}}^{(n)}(t)\right)}dB_{\bm{e}_{0} + \bm{e}_{d+1}}(t)\\ -  \frac{1}{\sqrt{n}} \sqrt{\rho^{(n)}_{-\bm{e}_{0} - \bm{e}_{d+1}}\left(\bar{\bm{E}}^{(n)}(t)\right)}dB_{-\bm{e}_{0} - \bm{e}_{d+1}}(t)	- \frac{1}{\sqrt{n}}\sum_{i=1}^d \sqrt{\rho^{(n)}_{-\bm{e}_{i} - \bm{e}_{d+1}}\left(\bar{\bm{E}}^{(n)}(t)\right)} dB_{-\bm{e}_{i} - \bm{e}_{d+1}}(t)\\
-  \frac{1}{\sqrt{n}} \sqrt{\rho^{(n)}_{-\bm{e}_{d+1}}\left(\bar{\bm{E}}^{(n)}(t)\right)}dB_{-\bm{e}_{d+1}}(t)
%\end{medsize} 
\end{multline*}
(see \citet{Kurtz1978} for a rigorous statement). 

Setting
 \[
 	\bar{I}^{(n)}(t) := \sum_{l=1}^{d} \bar{I}^{(n)}_{l}(t),
\]
so that
\[
	d\bar{I}^{(n)}(t) = \sum_{l=1}^{d} d\bar{I}^{(n)}_{l}(t), 
\]
and combining independent Brownian motions, we obtain equation (3) in the main text (\nb to simplify notation in the main text, we use $X$, $Y$ and $Z$ in lieu of $\bar{S}^{(n)}(t)$, $\bar{I}^{(n)}(t)$ and
$\bar{N}^{(n)}(t)$).  

We shall not use this diffusion approximation in the sequel, where we continue to consider the process with discrete jumps, \eqref{SDEE}.

\section{It\^o's Formula and Derivation of Equation (4)}\label{DERIVATION}

As a first application of the SDE representation, we apply It\^o's formula with jumps to our process to obtain an SDE for the proportion of each strain.  

To motivate this, suppose we had a deterministic differential equation 
\[
	\dot{\bm{Y}}(t) = \bm{f}(\bm{Y}(t))
\]
and we let  $X(t)$ be a deterministic real function of $\bm{Y}(t)$, say
\[
	X(t):=g(\bm{Y}(t))
\]
where $g:\mathbb{R}^{d+2}\to \mathbb{R}$ is assumed to be continuously differentiable.% Of course $X^{(n)}$ converges as 
%$n\to\infty$ to $X(t):=g\left(\bm{E}(t)\right)$. If we were only interested in the limiting deterministic case

Then, applying the chain rule, we derive a differential equation satisfied by $X(t)$:
\[
	\dot{X}(t) = \sum_{j=0}^{d+1}\frac{\partial g}{\partial x_{j}}(\bm{Y}(t)) f_{j}(\bm{Y}(t))
\]
or equivalently
\[
	X(t) =  X(0) +\int_0^t \sum_{j=0}^{d+1}\frac{\partial g}{\partial x_{j}}(\bm{Y}(s)) f_{j}(\bm{Y}(s))\, ds
\]

The analogue of the chain rule in the fully stochastic case is the Meyer-It\^o's formula (see \eg \citet{Protter04}).
 \begin{multline}
 \label{ITOP}
	X^{(n)}(t)= X^{(n)}(0)+\int_0^t\sum_{j = 0}^{d+1} 
	\frac{\partial g}{\partial x_{j}}(\bm{E}^{(n)}(s)) F^{(n)}_{j}(\bar{\bm{E}}^{(n)}(s))\\
	+ \frac{1}{2} \sum_{j,k = 0}^{d+1} a^{(n)}_{jk}(\bar{\bm{E}}^{(n)}(s))
		\frac{\partial^2 g}{\partial x_{j}\partial x_{k}} (\bar{\bm{E}}^{(n)}(s))\, ds
	+  \frac{1}{n}\int_0^t \sum_{j = 0}^{d+1} \frac{\partial g}{\partial x_{j}}(\bm{E}^{(n)}(s))\, 
		dM^{(n)}_{j}(s) + \varepsilon^{(n)}(t), 
\end{multline}
where $\bm{a}^{(n)}(\bm{x})$ is the infinitesimal variance-covariance matrix of $\bar{\bm{E}}^{(n)}(t)$ defined in \eqref{AA} and
\begin{multline}\label{ERROR}
	\varepsilon^{(n)}(t) =  \sum_{s < t} g(\bar{\bm{E}}^{(n)}(s)) - g(\bar{\bm{E}}^{(n)}(s-)) 
	- \sum_{j = 0}^{d+1} \frac{\partial g}{\partial x_{j}} (\bar{\bm{E}}^{(n)}(s-)) \Delta \bar{E}^{(n)}_{j}(s) \\
	- \frac{1}{2} \sum_{j,k = 0}^{d+1}  \frac{\partial^2 g}{\partial x_{j}\partial x_{k}}(\bar{\bm{E}}^{(n)}(s-))
	\Delta \bar{E}^{(n)}_{j}(s) \Delta \bar{E}^{(n)}_{k}(s),
\end{multline}
where the sum is over the times $s$ of discontinuity of $\bar{\bm{E}}^{(n)}$. At a time $t$ of discontinuity,
\[
	\Delta \bar{E}^{(n)}_{j}(t) := \bar{E}^{(n)}_{j}(t) - \bar{E}^{(n)}_{j}(t-) 
\]
denotes the magnitude of the jump in $\bar{E}^{(n)}_{j}$ at time $t$. The term $\varepsilon^{(n)}(t)$ correcting for discontinuities distinguishes the more general Meyer-It\^o formula from the familiar It\^o's formula for diffusions. In Section \ref{BOUNDING}, we show that $\varepsilon^{(n)}(t)=\BigO{1/n^2}$. 

Using this, we can derive Equation (4) from the main text.  Let 
\[
	\Pi_{i}(\bm{x}) = \frac{x_{i}}{\sum_{l=1}^{d} x_{l}}
\]
so that
\[
	P^{(n)}_{i}(t) = \Pi_{i}(\bar{\bm{I}}^{(n)}(t)) = \Pi_{i}(\bar{\bm{E}}^{(n)}(t)).
\]
is the proportion of the population infected with strain $i$. Since $P^{(n)}_{i}(t)$ is a deterministic function of $\bar{\bm{E}}^{(n)}(t)$, we can use It\^o's formula \eqref{ITOP} for jump processes. The following statement will be proved rigorously in Section \ref{BOUNDING}.
\begin{prop}
\label{prop:Pi_i}
The fraction of the population infected by strain $i$ satisfies
\begin{equation}\label{PIF}
	P^{(n)}_{i}(t) = \begin{multlined}[t] 
	P^{(n)}_{i}(0) 
	+ \int_{0}^{t} P^{(n)}_{i}(s) \left(r^{(n)}_{i}(\bar{S}^{(n)}(s),\bar{N}^{(n)}(s)) 
	- \sum_{j=1}^{d} r^{(n)}_{j}(\bar{S}^{(n)}(s),\bar{N}^{(n)}(s)) P^{(n)}_{j}(s)\right)\\
	+ \frac{1}{n} \frac{1}{\sum_{l=1}^{d} \bar{I}^{(n)}_{l}(s)} P^{(n)}_{i}(s)
	\left(v^{(n)}_{i}(\bar{S}^{(n)}(s),\bar{N}^{(n)}(s)) 
	- \sum_{j=1}^{d} v^{(n)}_{j}(\bar{S}^{(n)}(s),\bar{N}^{(n)}(s)) P^{(n)}_{j}(s)\right)\, ds\\
	+ \frac{1}{n} \int_{0}^{t} \frac{1}{\sum_{l=1}^{d} \bar{I}^{(n)}_{l}(s)} \sum_{j = 1}^{d} 
	\left(\mathbbm{1}_{\{i=j\}} - P^{(n)}_{i}(s)\right)\, 
		dM^{(n)}_{j}(s) + \varepsilon^{(n)}_{i}(t), 
	\end{multlined}
\end{equation}
where
\[
	r^{(n)}_{i}(x_{0},x_{d+1}) \defn \beta^{(n)}_{i}\frac{x_{0}}{x_{d+1}} - (\delta^{(n)}+\alpha^{(n)}_{i}+\gamma^{(n)}_{i})
\]
gives the Malthusian growth rate of strain $i$, whereas
\[
	v^{(n)}_{i}(x_{0},x_{d+1}) \defn \beta^{(n)}_{i}\frac{x_{0}}{x_{d+1}} 
	+ (\delta^{(n)}+\alpha^{(n)}_{i}+\gamma^{(n)}_{i})
\]
is the infinitesimal variance associated with the growth of strain $i$. In addition, for any $T>0$, there is a constant $C$ such that for all $t\in[0,T]$, $\mathbb{P}\{\vert n^2 \varepsilon^{(n)}_{i}(t)\vert\ge C\}$ vanishes as $n\to\infty$.
\end{prop}

To obtain Equation (4) in the main text, we omit the lower order error term $\varepsilon^{(n)}_{i}(t)=\BigO{1/n^2}$, recall that 
\[
	I^{(n)}(t) = \sum_{l=1}^{d} I^{(n)}_{l}(t) = n \sum_{l=1}^{d} \bar{I}^{(n)}_{l}(t) 
\]
gives the total number of infectives, and observe that, similarly to the previous section,
\[
	\frac{1}{n} dM^{(n)}_{i} \approx 
	\frac{1}{\sqrt{n}} \sqrt{v^{(n)}_{i}(\bar{S}^{(n)}(s),\bar{N}^{(n)}(s)) \bar{I}^{(n)}_{i}(t)}\, dB_{i}(t)
\]
for independent Brownian motions $B_{1},\ldots,B_{d}$, so that 
\[
	\frac{1}{n} \frac{1}{\sum_{l=1}^{d} \bar{I}^{(n)}_{l}(s)}\, dM^{(n)}_{j}
	\approx \frac{1}{\sqrt{I^{(n)}(t)}} \sqrt{v^{(n)}_{i}(\bar{S}^{(n)}(s),\bar{N}^{(n)}(s))P^{(n)}_{i}(t)}\, \, 	
	dB_{i}(t).
\]
%
%\begin{equation}
%	dP^{(n)}_{i}(t) \approx \begin{multlined}[t] 
%	\left(P^{(n)}_{i}(s) \left(r^{(n)}_{i}(\bar{S}^{(n)}(s),\bar{N}^{(n)}(s)) 
%	- \sum_{j=1}^{d} r^{(n)}_{j}(\bar{S}^{(n)}(s),\bar{N}^{(n)}(s)) P^{(n)}_{j}(s)\right)\right.\\
%	\left. + \frac{1}{n} \frac{1}{\sum_{l=1}^{d} \bar{I}^{(n)}_{l}(s)} P^{(n)}_{i}(s)
%	\left(v^{(n)}_{i}(\bar{S}^{(n)}(s),\bar{N}^{(n)}(s)) 
%	- \sum_{j=1}^{d} v^{(n)}_{j}(\bar{S}^{(n)}(s),\bar{N}^{(n)}(s)) P^{(n)}_{j}(s)\right)\right)\, dt\\
%	+ \frac{1}{n} \frac{1}{\sum_{l=1}^{d} \bar{I}^{(n)}_{l}(t)} \sum_{j = 1}^{d} 
%	\left(\mathbbm{1}_{\{i=j\}} - P^{(n)}_{i}(s)\right)\, 
%	\left(	dB_{S,j}(t) + dB_{j,-}(t) + dB_{j,R}(t) \right)		 
%	\end{multlined}
%\end{equation}

%\textcolor{blue}{We will need to rewrite the above equation \eqref{PIF} in the main text (for equation (4)). I change the notation and used $r$ and $var$ above. Let me know if you are ok or not with this notation.}

\section{Probability of Fixation of a Mutant Pathogen}

In this section, we will be interested in the long time behaviour of our multi-strain stochastic epidemics. In particular, we tackle the problem of predicting which strains will be outcompeted and which strains will fix.

\subsection{A Deterministic Limit and its Asymptotic Analysis}\label{ASYMPT}

We begin this section with a result stating the convergence to a deterministic dynamical system as $n\to\infty$.

\begin{prop}[Theorem 2.2, \citet{Kurtz1978}]\label{KURTZ}
If
\[
	\bar{S}^{(n)}(0) \to S(0), \quad 
	\bar{I}^{(n)}_{i}(0) \to I_{i}(0), 
	\quad \text{and} \quad \bar{N}^{(n)}(0) \to N(0)
\]
as $n \to \infty$, then for any fixed $T > 0$, with probability 1,
\begin{equation}\label{KURTZEQ}
	\sup_{t \leq T} \|\bar{\bm{E}}^{(n)}(t) - \bm{E}(t)\| \to 0,
\end{equation}
where 
%$\bar{S}^{(n)}(t)$, $\bar{I}^{(n)}_{i}(t)$, and $\bar{N}^{(n)}(t)$ converge almost surely in in $C([0,T],\mathbb{R}^{d+2})$ (the set of continuous functions from  $[0,T]$ to $\mathbb{R}^{d+2}$ with the supremum norm) to deterministic functions 
$ \bm{E}(t):=(S(t), I_{1}(t),\ldots, I_{d}(t), N(t))$ is the solution to the following system of ordinary differential equations: 
\begin{subequations}\label{SYSTEM}
 \begin{align}
 	\dot{S}(t) &= \lambda - \left(\sum_{i=1}^{d} \beta_{i}\frac{I_{i}(t)}{N(t)} + \delta\right) S(t), \label{S}\\ 
	\dot{I}_{i}(t) &= \left(\beta_{i}\frac{S(t)}{N(t)} - (\delta+\alpha_{i}+\gamma_{i})\right)I_{i}(t), \label{I}\\
	\dot{N}(t) &= \lambda - \delta N(t) - \sum_{i=1}^{d} \alpha_{i} I_{i}(t), \label{N}
\end{align}
\end{subequations}
with initial conditions $S(0)$, $\bm{I}(0)$, and $N(0)$.
\end{prop}
Note that the result in the previous proposition can be written more compactly as
\[
	\dot{\bm{E}} = \bm{F}(\bm{E}),
\] 
where 
\[
	\bm{F}(\bm{x}) = \lim_{n \to \infty} \bm{F}^{(n)}(\bm{x}).
\]
%In the following, we will stick to the finite $n$ fully stochastic process, while keeping in mind that this process converges to the deterministic solution of the previous ODE if we were to let $n\to\infty$.

While we continue to work with the finite $n$ fully stochastic process, the bifurcation structure of the deterministic system \eqref{SYSTEM} will guide our analysis of the stochastic model.  In particular, the steady states of this model, together with the degenerate case that arises when stability is exchanged between fixed points, give rise to two regimes that correspond to strong and weak selection in classical population genetics.  To be explicit, let 
\[
	R_{0,i} \defn \frac{\beta_{i}}{\delta+\alpha_{i}+\gamma_{i}}.
\]
be the basic reproduction number of strain $i$.  $R_{0,i}$ is the expected total number of new infections caused by a single infected individual, assuming an unlimited supply of susceptibles.

If $R_{0,i} \neq R_{0,j}$ for all $1 \leq i \neq j \leq d$, the equations \eqref{SYSTEM} have $d+1$ fixed points, one at $\bm{0}$ and one at the $d$ equilibria where the population is infected by a single strain 
\[
	\bm{E}^{\star,i} 
		\defn (S^{\star,i},I^{\star,i}_{i},\ldots,I^{\star,i}_{d},N^{\star,i}),
\]
where 
\begin{equation}\label{EQ}	
	S^{\star,i} \defn \frac{\lambda}{\delta R_{0,i}}
		\left(1-\frac{\alpha_{1}(R_{0,i}-1)}{\beta_{1}-\alpha_{1}}\right), \quad
	I^{\star,i}_{j} \defn \begin{cases} 
		\frac{\lambda(R_{0,i}-1)}{\beta_{i}-\alpha_{i}} & \text{if $i = j$, and}\\ 
		0 &  \text{otherwise,}
	\end{cases} \quad \text{and} \quad
	N^{\star,i} \defn R_{0,i}S^{\star}.
\end{equation}

When $d=1$, it is shown in \citet{VargasDeLeon2011} when $\delta > \alpha_{1}$ that: if $R_{0,1} > 1$ then unique endemic equilibrium of the strain, $\bar{\bm{E}}^{\star,1}$ is globally asymptotically stable, whereas if $R_{0,i} \leq 1$, the disease-free equilibrium $\bm{0}$ is globally asymptotically stable.  The stability of fixed points is slightly more subtle when there is more than one strain.

\begin{mydef}
\label{dfn:selection} We distinguish between two regimes of selection.
\begin{itemize}
\item[(i)] The \emph{strong selection} case, when $R_{0,1} > R_{0,i}$ for all $i > 1$, $\delta > \alpha_{1}$ and $R_{0,1} > 1$;
\item[(ii)] The \emph{weak selection} case, $R_{0,1} = R_{0,i} = R_{0}^{\star}$ for $i \leq m$, whilst $R_{0}^{\star} > R_{0,j}$ for $j> m$.
\end{itemize}
\end{mydef}

\begin{prop}
The long term behavior of the deterministic system \eqref{SYSTEM} differs according to the selection regime.
\begin{itemize}
\item[(i)] In the \emph{strong selection} case, the equilibrium state $\bar{\bm{E}}^{\star,1}$ with strain 1 endemic and all other strains extinct is globally asymptotically stable from any initial condition for which $I_{1}(0) > 0$.
\item[(ii)] In the \emph{weak selection} case, we arrive at a degenerate situation in which deterministic coexistence of strains $1,\ldots,d$ is possible. Strains $m+1,\ldots,d$ will eventually disappear, whereas all points $\bm{x} \in \mathbb{R}^{d+2}_{+}$ such that 
\begin{equation}\label{OMEGA}
\begin{gathered}
	\sum_{i=1}^{m} (\beta_{i}-\alpha_{i}) x_{i} = \lambda (R_{0}^{\star} - 1),\\ 
	x_{m+1} = \cdots = x_{d} = 0,\\
	x_{d+1} = \frac{1}{\delta}\left(\lambda - \sum_{i=1}^{m} \alpha_{i} x_{i}\right),\\
	 x_{d+1} = R_{0}^{\star} x_{0}
\end{gathered}
\end{equation}
are fixed points for the system \eqref{SYSTEM}.  The set $\Omega$ of such points is globally attracting, but no point in $\Omega$ is an attracting fixed point. 
\end{itemize}
\end{prop}
\begin{proof}
Point (i) follows by a direct adaptation of the result in \citet{Bremermann1989}: rearranging \eqref{I}, we see that
\[
	\frac{1}{\beta_{i}} \frac{\dot{I}_{i}(t)}{I_{i}(t)} + \frac{1}{R_{0,i}} = \frac{S(t)}{N(t)} = \frac{1}{\beta_{1}} \frac{\dot{I}_{1}(t)}{I_{1}(t)} + \frac{1}{R_{0,1}},
\]
so that
\[
	\left(\frac{I_{i}(t)}{I_{i}(0)}\right)^{\frac{1}{\beta_{i}}} e^{\frac{1}{R_{0,i}}t} = \left(\frac{I_{1}(t)}{I_{1}(0)}\right)^{\frac{1}{\beta_{1}}} e^{\frac{1}{R_{0,1}}t},
\]
and, recalling that $I_{1}(t)$ is bounded for all $t > 0$, we see that for all $i > 1$
\[
	I_{i}(t) = I_{i}(0)\left(\frac{I_{1}(t)}{I_{1}(0)}\right)^{\frac{\beta_{i}}{\beta_{1}}} e^{-\beta_{i}\frac{R_{0,1} - R_{0,i}}{R_{0,1}R_{0,i}}t} \to 0 
\]
as $t \to \infty$. 

The same argument shows that in the weak selection case, strains $m+1,\ldots,d$ will eventually disappear, whereas all points in $\Omega$ are fixed points. Moreover, all vectors $\bm{u}$ tangent to $\Omega$, \ie such that 
\[
	\sum_{i=1}^{m} (\beta_{i}-\alpha_{i}) u_{i} = 0, \quad
	u_{m+1} = \cdots = u_{d} = 0, 
	\quad \text{and} \quad u_{d+1} = R_{0}^{\star} u_{0},
\]
are all eigenvectors to the Jacobian of $\bm{F}$ -- evaluated at any $\bm{x} \in \Omega$ -- 
corresponding to the eigenvalue 0.  %The remaining eigenvalues at $\bm{x} \in \Omega$ are$-\delta$, 
%\[
%	\beta_{i}\left(\frac{1}{R_{0}^{\star}} - \frac{1}{R_{0,i}}\right)
%\]
%for $m < i \leq d$, and a pair of complex-conjugate eigenvalues with real part
%\[
%	\frac{1}{2}
%\]
Thus, while $\Omega$ is a globally attracting set, no point in $\Omega$ is an attracting fixed point. 
\end{proof}
We now turn to the computation of the fixation probability of a novel strain in the fully stochastic system. Informed by the previous statement, we will consider two cases, strong and weak selection, where the dynamics of the process -- and thus our approach to the fixation probabilities -- are qualitatively different.  We will then show, despite the difference in the approaches, and in the expressions for the fixation probability thereby obtained, that our two results for the fixation probability agree on all intermediate scalings, and may thus be combined (heuristically) via the method of matched asymptotic expansions, to obtain a single expression valid across all scales.

\subsection{The Strong Selection Case}\label{STRONG}

We begin by recalling that for the deterministic approximation \eqref{SYSTEM} to apply, we required that $\bar{S}^{(n)}(0) \to S(0)$, $\bar{I}^{(n)}_{i}(0) \to I_{i}(0)$ and $\bar{N}^{(n)}(0) \to N(0)$ as $n \to \infty$.  Unpacking this assumption, we see that 
\[
	I^{(n)}_{i}(0) = n I_{i}(0) + o(n),
\]
\ie that a non-trivial portion of the population is already infected with strain $i$.  

For any strain with $I^{(n)}_{i}(0) \ll n$, $\bar{I}^{(n)}_{i}(0) \to 0$, and thus $I_{i}(t) \equiv 0$ for all $t \leq T$, for any fixed $T > 0$: until  $\BigO{n}$ individuals are infected, strain $i$  is effectively invisible to the deterministic approximation on any finite time interval.  This is not to say that the strain is absent, but rather, if we sample individuals from the population uniformly at random, the probability of sampling an individual infected with strain $i$ is zero.

%Consider the case when a new strain appears with $I^{(n)}_{i}(0) \ll n$; in light of the results in \ref{AA}, we will assume that there is only a single resident strain, strain 1.   We have already observed that when $I^{(n)}_{2}(0) \ll n$, $I_{2}(t) \equiv 0$.  

We will consider the case when a fixed number $k$ of strain 2 individuals invade an established resident population.  For our purposes, a strain $i$ is established if it is initially present in macroscopic numbers, \ie
\[
	\bar{I}^{(n)}_{i}(0) \to I_{i}(0) > 0.
\]
In light of the results in \ref{ASYMPT}, we will assume that there is only a single resident strain, strain 1.  We will first consider the case when the resident strain is in endemic equilibrium, and then generalise to the case when the resident strain, whilst still present in macroscopic numbers, is initially away from equilibrium.  

Should the invading strain, strain 2, exceed $\varepsilon n$ individuals, for any $\varepsilon > 0$, we arrive again in the domain of applicability of the deterministic approximation ($\bar{I}^{(n)}_{2}(0) \to I_{2}(0) > \varepsilon > 0$).   If the reproductive number of the invader is greater than that of the resident, $R_{0,2} > R_{0,1}$, then for $n$ sufficiently large, the dynamics are essentially deterministic, and with high probability (\ie tending to 1 as $n \to \infty$) the process will in finite time $T_{\varepsilon}$ enter an  $\varepsilon$-neighbourhood of the fixed point 
$\bm{E}^{\star,2}$ for arbitrarily small $\varepsilon > 0$.  Once the process reaches this new equilibrium, we will see the resident strain is no longer viable, and subsequently disappears.

\subsubsection{Invasion at the Resident Endemic Equilibrium}

If we start at the endemic equilibrium of the resident strain 1,  $\bm{E}^{\star,1}$, then until  $I^{(n)}_{2}$ exceeds $\varepsilon n$, the epidemic process $\bm{E}^{(n)}(t)$ will remain close to that point.  We thus have  
\[
	\frac{S^{(n)}(t)}{N^{(n)}(t)} =
	\frac{\bar{S}^{(n)}(t)}{\bar{N}^{(n)}(t)}  \approx \frac{1}{R_{0,1}},
\]
and to first approximation, strain 2 has per-host transmission and clearance/mortality rates of 
\[
	\frac{\beta_{2}}{R_{0,1}} \quad \text{and} \quad \delta+\alpha_{2}+\gamma_{2}.
\]
This latter is a birth and death process (see \eg  \citet{Bartlett1955}) which will go extinct with probability 
\[
	q=\frac{\delta+\alpha_{2}+\gamma_{2}}{\frac{\beta_{2}}{R_{0,1}}} 
	= \frac{R_{0,1}}{R_{0,2}}
\]
if $R_{0,2} > R_{0,1}$, and with probability $q=1$ otherwise. This probability of extinction $q$ is for a single initial individual infected with strain 2, and becomes $q^k$ for $k$ initial individuals.
  Now, a birth and death process either goes extinct or grows arbitrarily large, so with probability $1-\frac{R_{0,2}}{R_{0,1}}$ it will eventually exceed $\varepsilon n$.

Similarly, when we have reached a neighbourhood of $\bm{E}^{\star,2}$ the transmission and clearance/mortality rates of strain 1 are approximately
\[
	\frac{\beta_{1}}{R_{0,2}} \quad \text{and} \quad \delta+\alpha_{1}+\gamma_{1}.
\]
Since $R_{0,2} > R_{0,1}$, this is a subcritical birth-death process which goes extinct with probability 1. Thus, invasion implies replacement, where for our purposes, the process \textit{invades} if it exceeds $\varepsilon n$ individuals for some fixed $\varepsilon > 0$.  

To summarise, we have heuristically derived 

\begin{prop}[Strong Selection]\label{SSPF}
Consider a population infected with $2$ strains such that $R_{0,2} > R_{0,1}$.
\begin{enumerate}[(i)]
\item \label{SSPF1} {\rm [Macroscopic initial frequencies]} If $\bar{I}^{(n)}_{2}(0) \to I_{2}(0) > 0$, then strain $1$ will go extinct with high probability.   
\item \label{SSPF2}  {\rm [Novel strain in small number of copies, resident at endemic equilibrium]} Suppose that $\bar{I}^{(n)}_{1}(0) \to \bar{I}^{\star,1}$, and that $I^{(n)}_{2}(0) = k$ for some fixed positive integer $k$.   Then, for any $\varepsilon > 0$
\begin{equation}\label{STRONGAPPROX}
		\lim_{n \to \infty} \mathbb{P}\left\{I^{(n)}_{1}(t) = 0\; \text{and}\; 
		\|\bar{\bm{E}}^{(n)}(t) - \bar{\bm{E}}^{\star,2}\| < \varepsilon\; 
		\text{for some}\; t < \infty \right\} 
		= 1- \left(\frac{R_{0,1}}{R_{0,2}}\right)^{k}.
\end{equation}
In other words, the event that strain 1 becomes extinct asymptotically coincides with the event that strain 2 invades, which happens with a probability asymptotically equal to the probability of survival $1- \left(\frac{R_{0,1}}{R_{0,2}}\right)^{k}$ of a time-homogeneous birth-death process; on this event, the system reaches in finite time the deterministic equilibrium \eqref{EQ} with only strain 2 endemic. 
\end{enumerate}
\end{prop}
We give a rigorous proof of this result in Section \ref{PROOF}, and compare \eqref{STRONGAPPROX} to fixation probabilities estimated from simulated epidemics in 
Figure \ref{STRONGAPPROXFIGURE:b}.

\subsubsection{Invasion Away From the Resident Endemic Equilibrium}\label{STRONGTRANSIENT}

In this section, we consider exactly the same setting as previously when a novel strain in small number of copies appears when the resident strain is at macroscopic initial frequency (\ie $I^{(n)}_{2}(0) = k$, $\bar{I}^{(n)}_{1}(0) \to \bar{I}_{1}(0) > 0$) but we relax the assumption that the resident strain is at endemic equilibrium. 
As previously, we use a branching process approximation to determine the probability that $I^{(n)}_{2}(t) > \varepsilon n$, for some small $\varepsilon > 0$, for some finite $t > 0$.  However, now, this branching process will no longer be time-homogeneous, since it will evolve within the changing environment imposed by the deterministic dynamics of the resident strain. Indeed, rather than assume that $\frac{\bar{S}^{(n)}(t)}{\bar{N}^{(n)}(t)}  \approx \frac{1}{R_{0,1}}$, we will use the law of large numbers to conclude that 
\[
	\frac{\bar{S}^{(n)}(t)}{\bar{N}^{(n)}(t)}  \approx \frac{S(t)}{N(t)}
\]
where $S(t)$ and $N(t)$ are determined via the reduced (deterministic) system 
\begin{subequations}\label{REDUCED}
 \begin{align}
 	\dot{S}(t) &= \lambda - \left(\sum_{i=1}^{d} \beta_{i}\frac{I_{i}(t)}{N(t)} + \delta\right) S(t),\\ 
	\dot{I}_{1}(t) &= \left(\beta_{1}\frac{S(t)}{N(t)} - (\delta+\alpha_{1}+\gamma_{1})\right)I_{1}(t),\\
	\dot{N}(t) &= \lambda - \delta N(t) - \sum_{i=1}^{d} \alpha_{i} I_{i}(t),
\end{align}
\end{subequations}
(\ie $I_{2}(t) \equiv 0$) with initial conditions 
\[
	S(0) = \lim_{n \to \infty} \bar{S}^{(n)}(0), \quad
	I_{1}(0) = \lim_{n \to \infty} \bar{I}^{(n)}_{1}(0), \quad \text{and} \quad
	N(0) = \lim_{n \to \infty} \bar{N}^{(n)}(0),
\]

We thus approximate the number of individuals infected with strain 2 by replacing the stochastic quantities $\bar{S}^{(n)}(t)$ and $\bar{N}^{(n)}(t)$ by their deterministic approximations, and the number of infectives with the novel strain 2 by a time-inhomogeneous birth and death process with per-host transmission and clearance/mortality rates of 
\[
	\beta_{2}\frac{S(t)}{N(t)} \quad \text{and} \quad \delta+\alpha_{2}+\gamma_{2}.
\]

As before, the probability this branching process reaches $\varepsilon n$ is exactly 1 less the probability of extinction for this time-inhomogeneous branching process, which is a classical result:

\begin{thm}[\citet{Kendall1948b}]\label{KENDALL}
Let $Z(t)$ be a continuous time linear birth-death process with time-varying birth rate $\lambda(t)$ and death rate $\mu(t)$, i.e., such that
\begin{gather*}
	\mathbb{P}\left\{Z(t+h) = k+1\mid Z(t) = k\right\} = \lambda(t)kh + o(h)\\
	\mathbb{P}\left\{Z(t+h) = k-1\mid Z(t) = k\right\} = \mu(t)kh + o(h)
\end{gather*}
Then, $m(t) \defn \mathbb{E}\left[Z(t)\middle\vert Z(0) = 1\right] = e^{\int_{0}^{t} \lambda(u) - \mu(u)\, du}$, 
{\footnotesize \[
	\mathbb{P}\left\{Z(t) > 0\right\} = \frac{1}{1+\int_{0}^{t} e^{\int_{0}^{s} \mu(u)-\lambda(u)\, du} \mu(s)\,ds},
\]}
and, the  probability of extinction in finite time, $q$, is
{\footnotesize \[
	q = \frac{J}{1+J},	
\]}
which is equal to 1 if and only if the integral $J := \int_{0}^{\infty} e^{\int_{0}^{s} \mu(u)-\lambda(u)\, du} \mu(s)\,ds$ diverges.  	
\end{thm}

Using the last statement for our infection process, we have

\begin{cor}\label{FIXNONEQ}
Consider a single individual infected with strain 2 entering a population where strain 1 is endemic but not necessarily at equilibrium.  Then, as $n \to \infty$, the probability strain 2 dies out is 
\begin{equation}\label{Q}
	q = \frac{J}{1+J} ,
\end{equation}
where 
\begin{equation}\label{QINTEGRAL}
	J:=\int_{0}^{\infty} e^{-\int_{0}^{s} \beta_{2} \frac{S(u)}{N(u)}-(\delta+\alpha_{2}+\gamma_{2})\, du} (\delta+\alpha_{2}+\gamma_{2})\, ds
\end{equation}
More generally, if strain 2 is initially in $k$ copies, the probability strain 2 fixes is asymptotic to $1-q^{k}$.
\end{cor}

The argument is, \textit{mutatis mutandis}, that of Section \ref{PROOF}; for details, we refer the reader to \citet{Parsons2012}.

While we can not evaluate $J$ analytically, we can evaluate it numerically.
More generally, we will compute $U(t) = \mathbb{P}\left\{Z(t) > 0\right\}$ numerically, and observe that it rapidly converges to equilibrium.  Rather than evaluate the integral directly, we find it more convenient to use the ordinary differential equations used in \citet{Kendall1948b} to derive \eqref{Q}.  Using the notation of Theorem \ref{KENDALL} (and \citet{Kendall1948b}), this is obtained via a system of two equations, 
\begin{align*}
	\dot{U}(t) &= - \mu(t) U(t)V(t)\\
	\dot{V}(t) &= (\mu(t) - \lambda(t)) V(t) - \mu(t) V(t)^{2},
\end{align*}
where $V(t) = \mathbb{P}\left\{Z(t)=1\mid Z(t) > 0\right\}$.
\begin{rem}
Whilst we will not use it in the sequel, we note that the auxiliary function $V(t)$ allows one to fully characterise the branching process $Z(t)$: the number of individuals alive at time $t$ is given by a modified geometric distribution with parameter $1-V(t)$:
\[
	\mathbb{P}\{Z(t) = k\} = \begin{cases}
		1-U(t) & \text{if $k = 0$, and}\\
		U(t)V(t)(1-V(t))^{k-1} & \text{if $k \geq 1$.}
	\end{cases}
\]
\end{rem}

For our epidemic model, this gives us the system
%\begin{subequations}
\begin{align*}
 	\dot{S}(t) &= \lambda - \left(\sum_{i=1}^{d} \beta_{i}\frac{I_{i}(t)}{N(t)} + \delta\right) S(t),\\ 
	\dot{I}_{1}(t) &= \left(\beta_{1}\frac{S(t)}{N(t)} - (\delta+\alpha_{1}+\gamma_{1})\right)I_{1}(t),\\
	\dot{N}(t) &= \lambda - \delta N(t) - \sum_{i=1}^{d} \alpha_{i} I_{i}(t),\\
	\dot{U}(t) &= - (\delta+\alpha_{2}+\gamma_{2}) U(t)V(t)\\
	\dot{V}(t) &= \left((\delta+\alpha_{2}+\gamma_{2}) - \beta_{2}\frac{S(t)}{N(t)}\right) V(t) 
		- (\delta+\alpha_{2}+\gamma_{2}) V(t)^{2},
\end{align*}
%\end{subequations}
which is easily computed numerically.

In Figure \ref{STRONGNUMERICFIGURE}, we show the consequences of choosing initial conditions away from equilibrium under two scenarios.  We fix the values  $R_{0,1} = 4$ and $R_{0,2} = 6$, where the increase in $R_{0,2}$ is achieved either by increasing the contact rate $\beta_{2}$ while holding the virulence, $\alpha_{2}$, fixed (solid curves) or by reducing the virulence $\alpha_{2}$ while holding the contact rate, $\beta_{2}$, fixed (dashed curves).  In both cases, we see that the change in the fixation probability (the asymptotic value of $U(t)$) is most visible when the number of individuals is infected with the resident strain is varied (Figure \ref{STRONGNUMERICFIGURE:a}), whereas it is less pronounced when the number of susceptibles is varied (Figure \ref{STRONGNUMERICFIGURE:b}), and relatively small when the population size is varied (Figure \ref{STRONGNUMERICFIGURE:c}).  We note that the strain that achieves the higher value $R_{0,2}$ via an elevated virulence has a higher probability of fixation than the strain with a higher contact rate when there is a surplus of resident-strain (green curves, \ref{STRONGNUMERICFIGURE:a}) and a lower fixation probability when there are fewer (blue curves, \ref{STRONGNUMERICFIGURE:a}).  This relation is inverted when the number of susceptibles is varied: when there are more susceptibles, the less virulent strain has a lower fixation probability than the strain with higher contact rate (green curves, \ref{STRONGNUMERICFIGURE:b}) and higher fixation probability when susceptible hosts are limited  (blue curves, \ref{STRONGNUMERICFIGURE:b}).  We will quantify these changes for small perturbations away from the endemic equilibrium in the next section.

\begin{figure}[h] 
  \begin{subfigure}[t]{0.33\linewidth}
    \centering
    \includegraphics[width=0.85\linewidth]{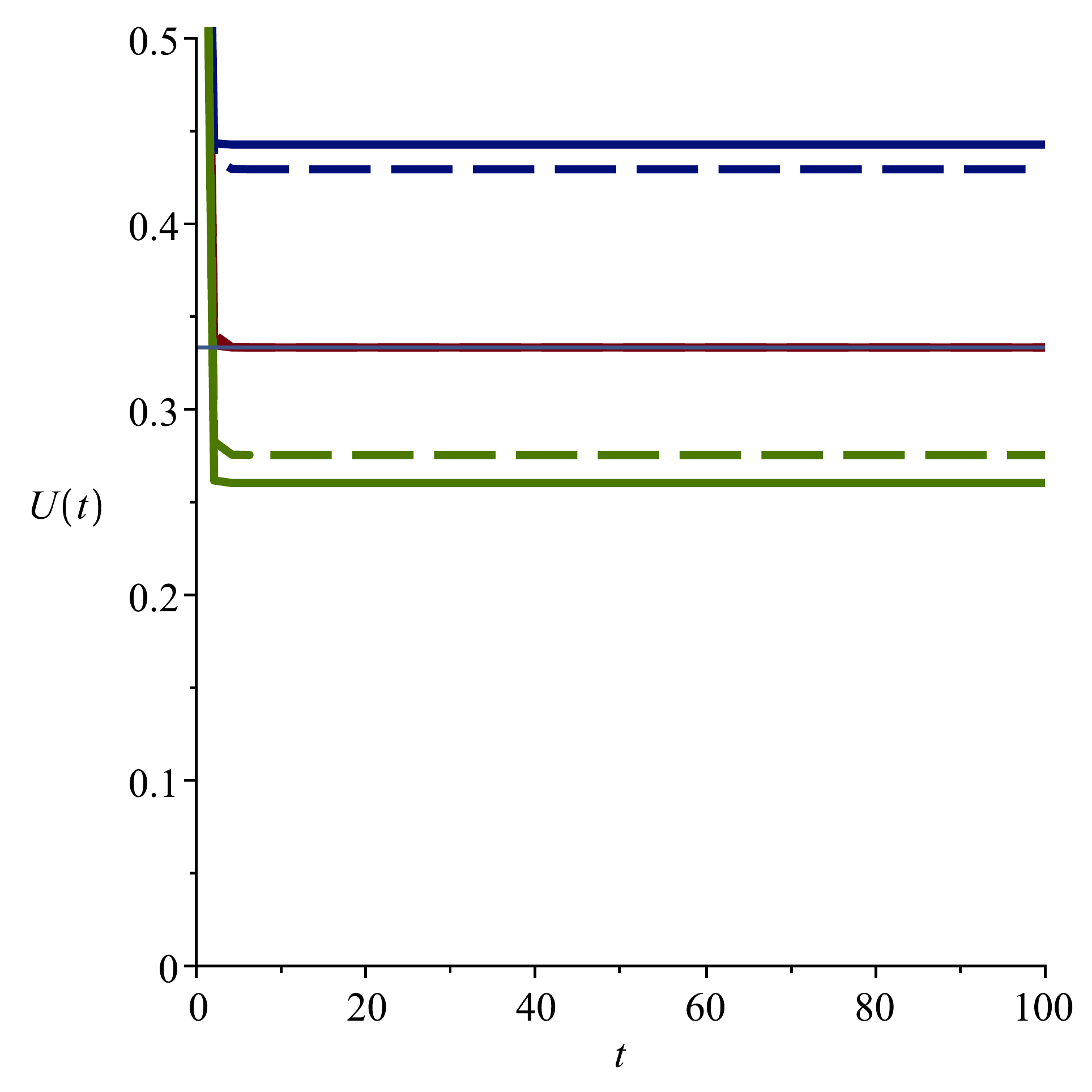} 
    \caption{varying $I_{1}(0)$} 
   \label{STRONGNUMERICFIGURE:a} 
    \vspace{4ex}
  \end{subfigure}%% 
  \begin{subfigure}[t]{0.33\linewidth}
    \centering
    \includegraphics[width=0.85\linewidth]{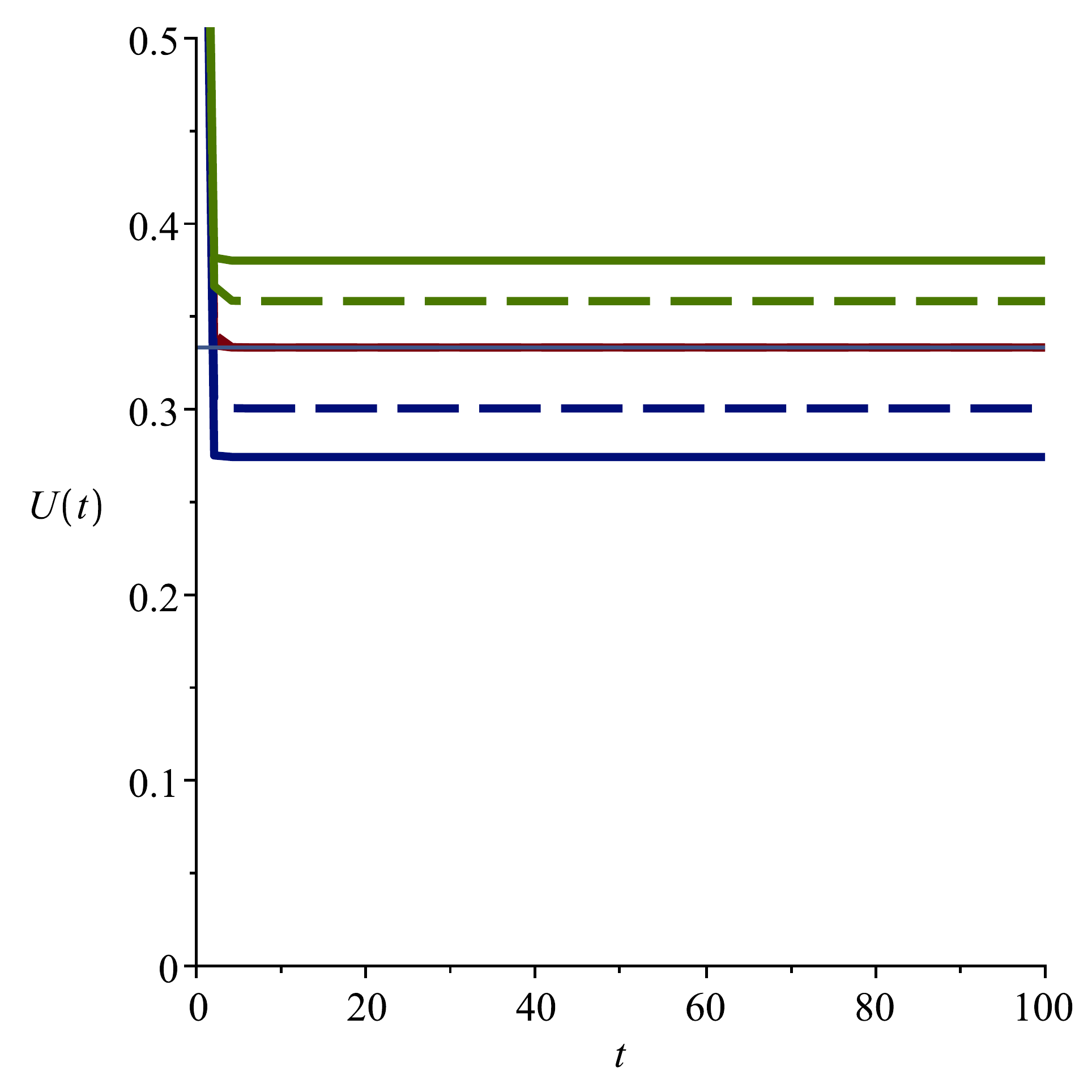} 
    \caption{varying $S(0)$} 
    \label{STRONGNUMERICFIGURE:b} 
    \vspace{4ex}
  \end{subfigure} 
  \begin{subfigure}[t]{0.33\linewidth}
    \centering
    \includegraphics[width=0.85\linewidth]{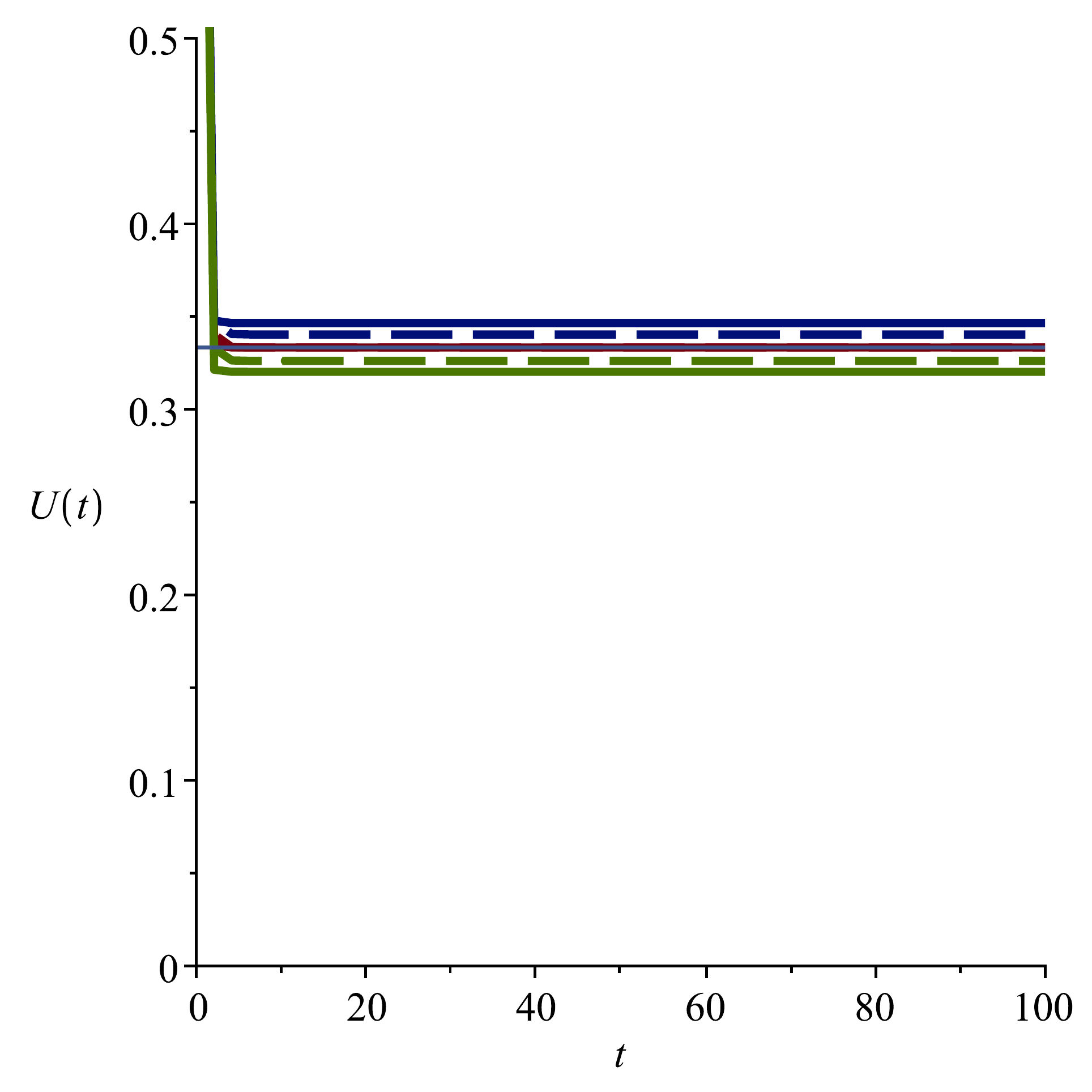} 
    \caption{varying $N(0)$} 
    \label{STRONGNUMERICFIGURE:c} 
  \end{subfigure}%%
    \caption{We assume a single mutant invading a resident population away from equilibrium.  We have set $R^{(n)}_{0,2}= \frac{3}{2} R^{(n)}_{0,1}$, so that the mutant is under strong favourable selection.  We implement this in two ways: first, by setting $\beta_{2} = \frac{3}{2} \beta_{1}$ (solid curves) and secondly, by setting $\alpha_{2} = \frac{\beta_{2}}{R^{(n)}_{0,2}}-\delta-\gamma_{2}$ (dashed curves).  For all curves, $\lambda = 2$, $\delta = 1$, $\beta_{1} = 20$, $\alpha_{1} = 3$, and $\gamma_{1} = \gamma_{2} = 1$.  When $\beta_{2}$ is varied, $\alpha_{2} = \alpha_{1} = 3$; when $\alpha_{2}$ is varied, $\beta_{2} = \beta_{1} = 20$.  For these parameters, the reduced system \eqref{REDUCED} has an equilibrium at $S^{\star,1} = \frac{4}{17}$, $I^{\star,1}_{1} = \frac{6}{17}$, $N^{\star,1} = \frac{16}{17}$.  (a) shows the effect of varying the initial number of individuals infected with the resident strain, $I_{1}(0) = I^{\star,1}_{1} = \frac{6}{17}$ (red), $I_{1}(0) = \frac{3}{17}$ (blue) and $I_{1}(0) = \frac{9}{17}$ (green).   (a) shows the effect of varying the initial number of susceptible individuals, $S(0) = S^{\star,1} = \frac{4}{17}$ (red), $S(0) = \frac{1}{17}$ (blue) and $S(0) = \frac{7}{17}$ (green).    (c) shows the effect of varying the initial population size, $N(0) = N^{\star,1} = \frac{16}{17}$ (red), $N(0) = \frac{13}{17}$ (blue) and $N(0) = \frac{19}{17}$ (green).  The probability of fixation for single mutant at the endemic equilibrium is $\frac{1}{11}$, which is shown in all panels by the azure line.}
  \label{STRONGNUMERICFIGURE} 
\end{figure}

\subsubsection{Invasion Near Endemic Equilibrium}

One way to study the effect of the initial state of the resident population on the invasion of a novel strain is to consider small perturbations near the endemic equilibrium of the resident,
\[	
	S^{\star} = \frac{\lambda}{\delta R_{0,1}}
		\left(1-\frac{\alpha_{1}(R_{0,1}-1)}{\beta_{1}-\alpha_{1}}\right), \quad
	I^{\star}_{1} =
		\frac{\lambda(R_{0,1}-1)}{\beta_{1}-\alpha_{1}} \quad \text{and} \quad
	N^{\star} = R_{0,i}S^{\star}.
\]
More specifically we consider
\[
	S(0) = S^{\star} + \varepsilon s(0), \quad
	I_{1}(0) = I_1^{\star} + \varepsilon i_{1}(0), \quad \text{and} \quad 
	N(0) = N^{\star} + \varepsilon n(0)
\] 
for some small dimensionless constant $\varepsilon > 0$ (independent of $n$ -- we consider perturbations that are a positive fraction of the population), and compute an expansion in $\varepsilon$ of the probability of invasion $1-q$ of strain 2 starting from one single infected with this novel strain. 

The proof of the next statement can be found in Section  .
\begin{prop}
\label{prop:noneq}
Set $\Lambda := \beta_{2} \left(\frac{1}{R_{0,1}} -\frac{1}{R_{0,2}}\right)$. Then the expansion of the invasion probability $1-q$ in the initial deviation $\varepsilon$ to endemic equilibrium of strain 1 is given by
\begin{multline}\label{NONEQ}
	1- q  = 1 - \frac{R_{0,1}}{R_{0,2}} 
	- \varepsilon \left(1 - \frac{R_{0,1}}{R_{0,2}}\right)^{2}	
	K
	\left(R_{0,1} \Lambda s(0) - \Lambda n(0)-(\beta_{1}-\alpha_{1}) i_1(0)\right)
	+ \BigO{\varepsilon^{2}}
\end{multline}
where
\begin{multline*}
K=\frac{(\beta_{1}-\alpha_{1})\beta_{2}^{2}\delta}
	{\left(\left(\Lambda+\delta\right)R_{0,1}\alpha_{1}
	-\left(\delta R_{0,1}-\Lambda\right)\beta_{1})\Lambda R_{0,1}
	-\beta_{1}\delta(\beta_{1}-\alpha_{1})(R_{0,1}-1)\right)}\\
\end{multline*}
\end{prop}

To first order, as expected, we obtain the probability $1 - \frac{R_{0,1}}{R_{0,2}}$ of fixation of strain 2 invading the strain 1 initially at equilibrium. The higher order terms allow us to study the effect of the change in population size. When the novel strain carries beneficial mutations (i.e. $R_{0,2} > R_{0,1}$) demographic perturbations that result in initial population growth of the pathogen population (i.e. $s(0) > 0$, $i_{1}(0) < 0$, $n(0) < 0$) increase  the probability of fixation of the novel strain. This effect is related to classical population genetics results (see e.g. \citet{Ewens1967, OttoWhitlock1997}) on probability of fixation in populations of changing size.

Note, however, that away from equilibrium, for given values of $R_{0,1}$ and $R_{0,2}$, the fixation probability depends also on the transmission $\beta_{2}$ of the novel strain. In other words, the life history traits of the novel strain can also affect the evolutionary outcome. If the novel strain is more transmissible than the resident (i.e. $\beta_{2} > \beta_{1}$)  its probability of fixation is favoured when the perturbation leads to an initial increase of the pathogen population. For instance, a high number of susceptible hosts (i.e. $s(0) > 0$) increases the probability of fixation of a transmissible strain. This effect results from a form of $r$ \vs $K$ selection (\citep{Pianka1970, Pianka1972, Reznick2002}), where fast-growing (high transmission and high virulence) strains are favoured in undersaturated environments, whereas long-lived (low virulence and low transmission) strains are favoured in over-saturated environments. The effect of epidemiology on evolution is very apparent in the analysis of the deterministic dynamics (see equation (**) in the main text). The above \eqref{NONEQ} provides a stochastic treatment of the influence of epidemiology on the ultimate evolutionary outcome.

Furthermore, we note that 
\[ 	
	K = -	\frac{\beta_{2}^2}{\beta_{1}(R_{0,1}-1)}
	+ \BigO{1 - \frac{R_{0,1}}{R_{0,2}}}
\]
and
\[ 
	\Lambda = \frac{\beta_{2}}{R_{0,1}}\left(1 - \frac{R_{0,1}}{R_{0,2}}\right),
\]
so that
\begin{equation}\label{NONEQRED}
	1-q  = 1 - \frac{R_{0,1}}{R_{0,2}} 
	- \varepsilon \left(1 - \frac{R_{0,1}}{R_{0,2}}\right)^{2} 
	\frac{\beta_{2}^2(\beta_{1}-\alpha_{1}) i_1(0)}{\beta_{1}(R_{0,1}-1)}	
	+ \BigO{\varepsilon^{2},\left(1 - \frac{R_{0,1}}{R_{0,2}}\right)^{3}}.
\end{equation}
Thus, if $R_{0,1}$ and $R_{0,2}$ are close, we see that the difference in the fixation probability from that at equilibrium is, to lowest order, proportional to $ i_1(0)$, and is positive if $i_{1}(0) < 0$, and negative if $i_{1}(0) > 0$.  Moreover, the difference is proportional to $\beta_{2}^{2}$, so that if $i_{1}(0) < 0$, there is an advantage to larger values of $\beta_{2}$, whereas if $i_{1}(0) > 0$, the fixation probability is maximized by minimizing $\beta_{2}$.

\begin{exmp}[Vaccination]
Suppose that starting at time $t  = 0$, a fraction $f$ of the incoming susceptibles to a population at the endemic equilibrium are immunized.  Then, the rate of incoming susceptibles is reduced from $\lambda$ to $\lambda(1-f)$.  Recalling the values of $S^{\star}$, $I^{\star}_{1}$, and $R^{\star}$
above, we see that post-vaccination, the population has a new equilibrium, 
\[	
	S^{\star}_{f} = S^{\star} + \frac{\alpha_{1} \lambda f}{\delta (\beta_{1}-\alpha_{1})}, \quad
	I^{\star}_{1,f} = I^{\star}_{1} - \frac{\lambda R_{0,1} f}{\beta_{1}-\alpha_{1}}
		\quad \text{and} \quad
	N^{\star}_{f} = N^{\star} + \frac{\alpha_{1} \lambda R_{0,1} f}{\delta (\beta_{1}-\alpha_{1})}
\]
In particular, the initial condition differs from the new equilibrium by a perturbation of magnitude $f$ with 
\[ 
	s(0) = - \frac{\alpha_{1} \lambda}{\delta (\beta_{1}-\alpha_{1})}, \quad 
	i_{1}(0) = \frac{\lambda R_{0,1}}{\beta_{1}-\alpha_{1}}, \quad \text{and} \quad 
	n(0) = - \frac{\alpha_{1} \lambda R_{0,1}}{\delta (\beta_{1}-\alpha_{1})}.
\]
Moreover, \eqref{NONEQ} is otherwise independent of $\lambda$, and thus $f$, so that when $f$ is not too large, we may use it to determine the fixation probability of a new strain that emerges at approximately the time when vaccination commences.  In particular, if we assume if $R_{0,1}$ and $R_{0,2}$ are close, we can use \eqref{NONEQRED} to see that the fixation probability of an invading strain is
\[
	1 - \frac{R_{0,1}}{R_{0,2}} 
	- f \left(1 - \frac{R_{0,1}}{R_{0,2}}\right)^{2} 
	\frac{\beta_{2}^2\lambda R_{0,1}}{\beta_{1}(R_{0,1}-1)}	
	+ \BigO{f^{2},\left(1 - \frac{R_{0,1}}{R_{0,2}}\right)^{3}},
\]
which is lower than the fixation probability at equilibrium, and decreasing with increasing contact rate.  Thus, vaccination reduces the probability of new strains arising, in particular more virulent strains.
\end{exmp} 
 
\subsection{The Weak Selection Case}\label{WEAK}

Recall that weak selection corresponds to the case when $R_{0,i} = R_{0,j} = R_{0}^{\star}$ for $1 \leq i,j \leq m$.  We hasten to clarify, however, that 
\[
	R^{(n)}_{0,i}= \frac{\beta^{(n)}_{i}}{\delta^{(n)}+\alpha^{(n)}_{i}+\gamma^{(n)}_{i}},
\]
so our assumptions \eqref{ASSUMPTIONS} only impose that 
\[
	R^{(n)}_{0,i}= R_{0}^{\star}\left(1+\frac{r_{i}}{n}\right) 
		+ {\textstyle o\left(\frac{1}{n}\right)}
\]
\ie $R^{(n)}_{0,i}$ and $R^{(n)}_{0,j}$ are allowed to differ by $\BigO{\frac{1}{n}}$ terms; as we shall see below, this is analogous to the weak selection limit of classical population genetics, and the values $r_{i}$ will appear as selection coefficients in a diffusion approximation.

In this case, we have a separation of timescales: there is a fast time-scale, in which \eqref{KURTZEQ} tells us that the stochastic process approximately follows the trajectories of \eqref{SYSTEM} arbitrarily closely to an arbitrarily small neighbourhood of $\Omega$.  Then, as we discuss below, there is a slow-time scale, in which, having arrived at $\Omega$, the stochastic process remains near this critical manifold.  

Let
\begin{gather*}
	\hat{S}^{(n)}(t) \defn \bar{S}^{(n)}(nt) = \frac{1}{n} S^{(n)}(nt), \\
	\quad \hat{I}^{(n)}_{i}(t) \defn \bar{I}^{(n)}_{i}(nt) = \frac{1}{n} I^{(n)}_{i}(nt), \\
	\hat{N}^{(n)}(t) \defn \bar{N}^{(n)}(nt) = \frac{1}{n} N^{(n)}(nt), \\
\end{gather*}
and let 
\[
	\hat{\bm{E}}^{(n)}(t) \defn (\hat{S}^{(n)}(t),\hat{\bm{I}}^{(n)}_{i}(t),\hat{N}^{(n)}(t)),
\]
where, as before, we let $\hat{\bm{I}}^{(n)}(t) = \left(\hat{I}^{(n)}_{1}(t),\ldots,\hat{I}^{(n)}_{m}(t)\right)$, and note that 
\[
	(\hat{S}^{(n)}(0),\hat{\bm{I}}^{(n)}(0),\hat{N}^{(n)}(0))
	= (\bar{S}^{(n)}(0),\bar{\bm{I}}^{(n)}(0),\bar{N}^{(n)}(0)).
\]
Here, rescaling time by $n$ is analogous to the passage to so-called ``coalescent time'' or ``generation time'', which is used to derive the diffusion limit of the Wright-Fisher model in classical population genetics.  

Recalling \eqref{SDEE}, the SDE for $\hat{\bm{E}}^{(n)}(t)$ is then 
\begin{align*}
	\hat{\bm{E}}^{(n)}(t) &= \hat{\bm{E}}^{(n)}(0)
	+ \int_{0}^{nt} \bm{F}^{(n)}\left(\hat{\bm{E}}^{(n)}(s)\right)\, ds + \frac{1}{n} \bm{M}^{(n)}(nt)\\
	&= \hat{\bm{E}}^{(n)}(0)
	+ \int_{0}^{t} n \bm{F}^{(n)}\left(\hat{\bm{E}}^{(n)}(s)\right)\, ds + \frac{1}{n} \bm{M}^{(n)}(nt).
\end{align*}
Thus, in the slow time scale, the drift is accelerated by a factor of $n$, causing the process to move rapidly to the critical manifold $\Omega$; as $n \to \infty$, this movement becomes instantaneous, and the process immediately jumps to $\Omega$ at time $t = 0$.  Moreover, stochastic fluctuations away from $\Omega$ are restored instantaneously, so the process becomes ``trapped'' on $\Omega$ as $n \to \infty$.
The following statement formalizes this idea using the projection $\bm{\pi}$ defined as follows.
Let $\bm{E}(t,\bm{x}) = (S(t,\bm{x}),\bm{I}(t,\bm{x}),N(t,\bm{x}))$ be the solution to \eqref{SYSTEM} with initial conditions $S(0) = x_{0}$, $I_{i}(0) = x_{i}$, and $N(0) = x_{d+1}$ and let 
% (\ie $\bm{I}(t,\bm{x}) = (I_{1}(t,\bm{x}),\ldots,I_{m}(t,\bm{x}))$) 
\[
	\bm{\pi}(\bm{x}) := \lim_{t \to \infty} \bm{E}(t,\bm{x}),
\]
\ie $\bm{\pi}(\bm{x})$ is the point on $\Omega$ at which the trajectory of \eqref{SYSTEM} starting from $\bm{x}$ meets $\Omega$.  

%Recall  the independent Brownian motions (taking $d=m$) defined in Eq \eqref{eqn:indepMBs}, the infinitesimal variance-covariance matrix $a_{ij}^{(n)}$ defined in Eq \eqref{A} and its expression given by Eq \eqref{AA} in terms of the maps $\sigma_{ij}^{(n)}$ specified in Eq \eqref{sigma}.
\begin{prop}
\label{prop:weaklimit}
As $n \to \infty$, $(\hat{\bm{E}}^{(n)}(t))\Rightarrow (\hat{\bm{E}}(t))\footnote{A family of random variables 
	$\{X^{(n)}\}$ taking values in a space $S$ is said to \textit{converge weakly} to $X$ if 
	\[
		\lim_{n \to \infty} \mathbb{E}[f(X^{(n)})] = \mathbb{E}[f(X)]
	\]
	for all $f \in C(S)$; the values $\mathbb{E}[f(X)]$ completely characterise the distribution of 
	$X$.  Weak convergence is denoted by
	\[
		X^{(n)} \Rightarrow X.
	\]
	Here, $S$ is the Skorokhod space $\mathbb{D}_{\mathbb{R}^{d+2}}[0,\infty)$ of right-
	continuous functions from $[0,\infty)$ to $\mathbb{R}^{d+2}$ with left limits; the interested 
	reader is referred to \citet{Billingsley1968} for a very readable account of weak convergence on 
	$\mathbb{D}$.} $ 
where the latter is a diffusion on the manifold $\Omega$, solution to the system of stochastic differential equations
\begin{equation}\label{WEAKLIMIT}
	d\hat{E}_{i}  
	%\sum_{j = 0}^{d+1} \frac{\partial\pi_{i}}{\partial x_{j}} d\hat{E}_{j}(t)
	%+ \frac{1}{2} \sum_{j,k} a_{jk}(\hat{\bm{E}})\frac{\partial^{2}\pi_{i}}{\partial x_{j}\partial x_{k}}\\
	= \sum_{j = 0}^{d+1} \frac{\partial\pi_{i}}{\partial x_{j}} (\hat{\bm{E}}) f_{j}(\hat{\bm{E}})\, dt
	+ \frac{1}{2} \sum_{j = 0}^{d+1} \sum_{k=0}^{d+1} \frac{\partial^{2}\pi_{i}}{\partial x_{j}\partial x_{k}} (\hat{\bm{E}}) a_{jk}(\hat{\bm{E}})\, dt
	+ \sum_{j = 0}^{d+1} \sum_{k=1}^{D} \frac{\partial\pi_{i}}{\partial x_{j}}(\hat{\bm{E}}) \sigma_{jk}(\hat{\bm{E}})
	\, dB_{k}(t)
\end{equation}
where the $(B_k)$ denote  $D=3(m+1)$ independent standard Brownian motions,
\[
	\bm{f}(\bm{x})  := \lim_{n \to \infty} n\left(\bm{F}^{(n)}(\bm{x}) - \bm{F}(\bm{x})\right),
\]
and $\bm{\sigma}(\bm{x})$ %:=\lim_{n\to\infty}\bm{\sigma}^{(n)}(\bm{x})$ 
is the $(d+1) \times D$ matrix\\

\resizebox{\linewidth}{!}{
%\hspace{-2.65cm}
%\bm{\sigma}(\bm{x}):=	 
$%\begin{medsize} 
\begin{bmatrix} 
	\sqrt{\lambda} & -\sqrt{\delta x_{0}} & -\sqrt{\frac{\beta_{1} x_{0} x_{1}}{x_{d+1}}} & 0 & 0 & -\sqrt{\frac{\beta_{2} x_{0}x_{2}}{x_{d+1}}}  &\cdots & 0 & -\sqrt{\frac{\beta_{d} x_{0}x_{d}}{x_{d+1}}}&0&0&0 \\
	0 & 0 & \sqrt{\frac{\beta_{1} x_{0} x_{1}}{x_{d+1}}} & - \sqrt{(\delta\!+\!\alpha_{1})x_{1}}
		& -\sqrt{\gamma_{1} x_{1}} & 0  & \cdots& \cdots&\cdots&\cdots&0&0\\
	\vdots&\vdots &\ddots &\ddots & \ddots  & \ddots& \ddots& \ddots& \ddots& \ddots&\ddots& \vdots \\
	0&0 &0 &\cdots&\cdots&\cdots & \cdots& 0& \sqrt{\frac{\beta_{d} x_{0} x_{d}}{x_{d+1}}} & - \sqrt{(\delta\!+\!\alpha_{d})x_{d}} 
		& -\sqrt{\gamma_{d} x_{d}}& 0 \\
	\sqrt{\lambda} & -\sqrt{\delta x_{0}} & 0 & - \sqrt{(\delta\!+\!\alpha_{1}) x_{1}} & 0 & 0 & \cdots&\cdots&0& - \sqrt{(\delta\!+\!\alpha_{d})x_{d}} &0&\scriptscriptstyle-\sqrt{\delta \left(x_{d+1}- \sum_{i=0}^{d} x_i\right)}
	\end{bmatrix} 
%\end{medsize}
$,}\\

\noindent and
\[
	\bm{a}(\bm{x}) = \lim_{n \to \infty}n \bm{a}^{(n)}(\bm{x}) =\bm{\sigma}(\bm{x})\bm{\sigma}(\bm{x})^{\top} .
\]
%and we note that $\bm{\sigma}(\bm{x})\bm{\sigma}(\bm{x})^{\top} = \bm{a}(\bm{x})$. 
In other words, for each $i$
\begin{multline}\label{PROJECTED}
d\hat{\bm{E}}_i(t) = 	\bm{D}\pi_i(\hat{\bm{E}}(t)) \bm{f}(\hat{\bm{E}}(t)) \, dt 
+	\bm{D}\pi_i(\hat{\bm{E}}(t)) \,\bm{\sigma}(\hat{\bm{E}}(t))\, d\bm{B}(t)\\
+ \frac12 {\rm Tr}\left[\bm{\sigma}^{\top}(\hat{\bm{E}}(t)) \bm{H}\pi_i(\hat{\bm{E}}(t))\,\bm{\sigma}(\hat{\bm{E}}(t))\right]\, dt,
\end{multline}
where $D\pi_i$ denotes the gradient vector of $\pi_i$, $H\pi_i$ its Hessian matrix, and $\rm Tr$ denotes the trace operator.
\end{prop}
%As we will show later, $\bm{D}\bm{\pi}(\hat{\bm{E}}(t))$ is the projection onto the tangent space to $\Omega$ along the tangents to the trajectories of the deterministic flow.
%We remark that the drift terms in these equations correspond to the vector fields in \eqref{SYSTEM}, with the original processes replaced by their time rescaled equivalents, and  multiplied by $n$; thus, the drift, which is in the direction of $\Omega$, becomes arbitrarily strong as $n \to \infty$, except when the process already lies in $\Omega$, in which case the drift vanishes.  This explains both why the limiting process is confined to $\Omega$, and also why the limit is non-trivial - the noise terms remain.  They also prevent us from directly taking the limit in the equations above.

This diffusion can be understood as the result of stochastic fluctuations around $\Omega$ immediately followed by a strong deterministic drift towards $\Omega$. 
%To simplify our explanation, let us assume that all coefficients of the dynamics do not depend on $n$, that is, $\alpha^{(n)}_{i} = \alpha_{i}$, $\beta^{(n)}_{i} = \beta_{i}$, etc. In particular, $F^{(n)}_{i} = F_{i}$, so that $\bm{f}\equiv \bm 0$, and the last displayed equation becomes
%\begin{equation}
%\label{f-equals-zero}
%d\hat{\bm{E}}_i(t) = 	\bm{D}\pi_i(\hat{\bm{E}}(t))\, \bm{\sigma}(\hat{\bm{E}}(t))\, d\bm{B}(t)
%+ \frac12 {\rm Tr}\left[\bm{\sigma}^{\top}(\hat{\bm{E}}(t)) \bm{H}\pi_i(\hat{\bm{E}}(t))\,\bm{\sigma}(\hat{\bm{E}}(t))\right]\, dt.
%\end{equation}
%\textcolor{red}{To better understand the above equation recall the diffusion approximation \eqref{Diff-approx}}. After acceleration by $n$, it reads
%$$
%d\hat{\bm{E}}^{(n)}(t) = 	n \bm{F}\left(\hat{\bm{E}}^{(n)}(t)\right)\, dt +\bm{\sigma}\left(\hat{\bm{E}}^{(n)}(t)\right) d\bm{B}(t),
%$$

As can be seen from Proposition \ref{KURTZ}, the drift pushes the process very rapidly onto $\Omega$, so that in the limit, the process lives permanently in $\Omega$. Now to understand the interplay between the deterministic dynamics towards $\Omega$ and the stochastic fluctuations around $\Omega$, it is useful to think of the dynamics in two steps. Suppose that starting from a point $\hat{\bm{E}}(t-) \in \Omega$, the process $\bm{E}$ has a jump $\bm{l}$.  Then, the rescaled process $(\hat{\bm{E}}^{(n)}(t))$ has a jump $\frac{1}{n} \bm{l}$.

In a second step, it is immediately projected back to the manifold by the drift \textit{at the new location}, so:
\[
	d\hat{\bm{E}}(t) = \bm{\pi}\left(\hat{\bm{E}}(t-) + \frac{1}{n} \bm{l}\right) - \hat{\bm{E}}(t-)
\]
Thus, expanding the $i$-th component of the r.h.s. of the last equation and recalling that $\bm{\pi}\left(\hat{\bm{E}}(t-)\right) = \hat{\bm{E}}(t-)$ yields
\begin{align*}
	d\hat{E}^{(n)}_{i}(t)&= \pi_{i}\left(\hat{\bm{E}}^{(n)}(t-) + \frac{1}{n} \bm{l}\right)
	 - E^{(n)}_{i}(t-)\\
	&= \frac{1}{n} \sum_{j} \frac{\partial \pi_{i}}{\partial x_{j}}(\hat{\bm{E}}^{(n)}(t-)) l_{i}
	+  \frac{1}{n^{2}} \sum_{j} \sum_{k} \frac{1}{2}
		\frac{\partial^{2} \pi_{i}}{\partial x_{j} \partial x_{k}}(\hat{\bm{E}}^{(n)}(t-)) l_{j}l_{k} 
		+ {\textstyle o\left(\frac{1}{n^{2}}\right)}.
\end{align*}

To determine $\mathbb{E}[d\hat{E}^{(n)}_{i}(t)]$ (\ie the $i$\textsuperscript{th} component of the drift in the diffusion approximation) we need only sum this over all possible jumps $\bm{l}$, weighted by their probabilities:
\begin{multline*}
\mathbb{E}[d\hat{E}^{(n)}_{i}(t)] \\
= \sum_{\bm{l}} \left(\frac{1}{n} \sum_{j} \frac{\partial \pi_{i}}{\partial x_{j}}(\bm{E}^{(n)}(t-)) l_{j}
+  \frac{1}{n^{2}} \sum_{j} \sum_{k} \frac{1}{2}
	\frac{\partial^{2}  \pi_{i}}{\partial x_{j} \partial x_{k}}(\bm{E}^{(n)}(t-)) l_{j}l_{k} 
	+ {\textstyle o\left(\frac{1}{n}\right)}\right) n \rho_{\bm{l}}(\bm{E}^{(n)}(t-))\, d(nt)\\
= n \sum_{j} \frac{\partial \pi_{i}}{\partial x_{j}}(\hat{\bm{E}}^{(n)}(t-)) F^{(n)}_{j}(\hat{\bm{E}}(t-))\, dt 
		+ \frac{1}{2} \sum_{j} \sum_{k} \frac{\partial^{2} \pi_{i}}{\partial x_{j} \partial x_{k}}
		(\hat{\bm{E}}^{(n)}(t-)) a_{jk}(\hat{\bm{E}}^{(n)}(t-))\, dt,
\end{multline*}
where, because we have rescaled time, $d(nt)$ replaces $dt$ in probability of a jump at $t$.  Recalling $\bm{F}(\hat{\bm{E}}(t)) = 0$ and \eqref{AAA}, this yields the first term of the previous equation yields the first term of Eq \eqref{WEAKLIMIT}.

A picture (Figure \ref{GEOMETRYFIGURE}) more immediately explains the emergence of the variance induced drift: unless the flow lines are parallel, jumps of identical magnitude and direction will be returned to the manifold $\Omega$ at different distances from the initial point, as one moves along the manifold:

\begin{figure}[h] 
    \centering
    \includegraphics[width=\linewidth]{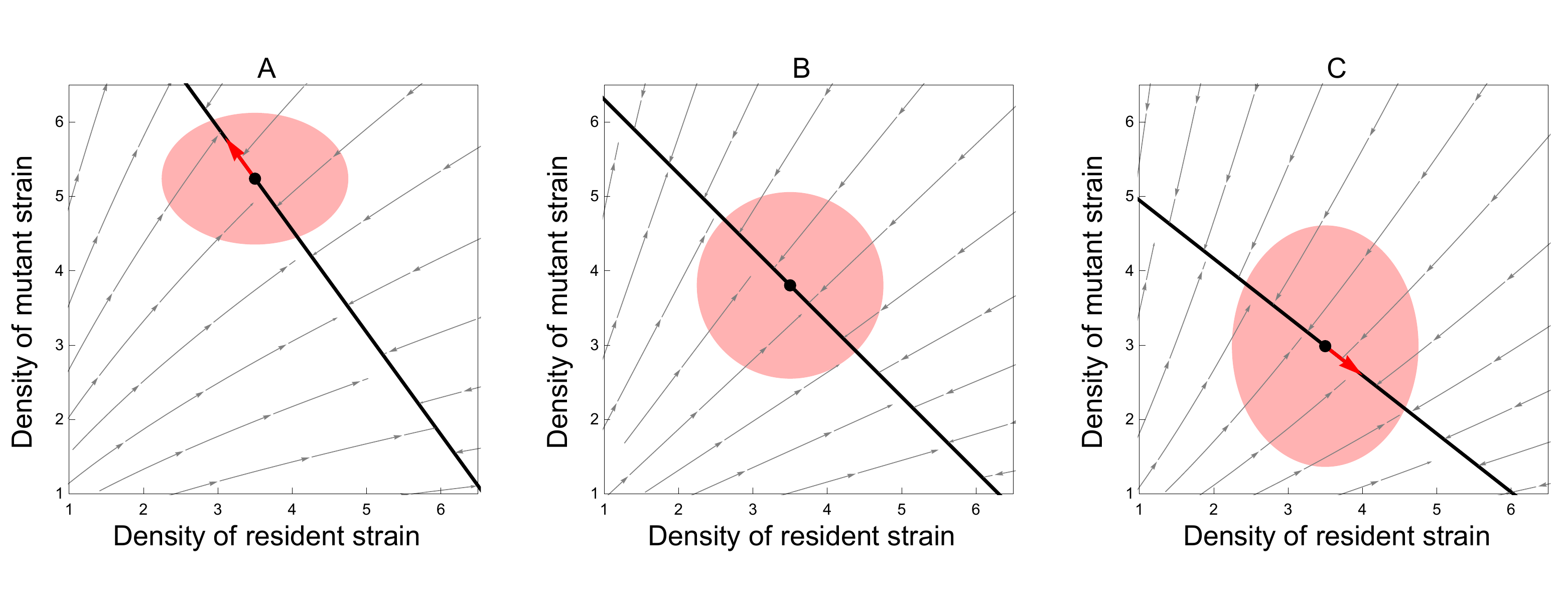} 
    \caption{Dynamics of the densities of the resident and the mutant strain in the phase plane when the two strains have the same basic reproduction number $R_{0}$, but under three different scenarios. (A) The mutant strain has a lower virulence than the resident. (B) The two strains have the same virulence. (C) The mutant strain has a higher virulence than the resident. The deterministic trajectories are shown as grey arrows that point towards the manifold $\Omega$ (the black line). The light red ellipsoid has axes proportional to the infinitesimal variance of the jumps that displace each strain from a given point on the manifold (black dot).  The combination of the effect of stochasticity and the fast deterministic return to the manifold generates a drift (red arrow) that favours the strain with the lower virulence. Parameter values of the resident:  $\beta_{1} = 10$,  $\alpha_{1} = 2$,  $\delta = 0.05$,  $\gamma = 0.5$. Virulence of the mutant:  $\alpha_{2} = 1.25$, 2, and 2.75 in A, B and C, respectively.} 
   \label{GEOMETRYFIGURE} 
\end{figure}

Of course the rigorous way of obtaining the result is to use It\^o's formula as done in the proof.
\begin{proof}
The weak convergence $\hat{\bm{E}}^{(n)}\Rightarrow \hat{\bm{E}}$ is proven in \citet{Katzenberger1991} (see \citet{Parsons2017} for an informal, applications-oriented discussion).

To characterise the limit $\hat{\bm{E}}$, we shall make use of $\bm{\pi}$. Unfortunately, $\bm{\pi}(\bm{x})$ is impossible to compute analytically, but we can still use it to obtain an SDE for $\hat{\bm{E}}(t)$.  We first observe that if $\bm{F}$ is twice-continuously differentiable, then $\bm{\pi}$ is as well \citet{Hirsch+Smale74}.  The continuity of $\bm{\pi}$ then tells us that $\bm{\pi}(\hat{\bm{E}}^{(n)}(t)) \Rightarrow \bm{\pi}(\hat{\bm{E}}(t))$ as well.   Since $\bm{\pi}$ has first and second derivatives, we may apply It\^o's formula (see Section \ref{DERIVATION}) to $\bm{\pi}(\hat{\bm{E}}^{(n)}(t))$:
 \begin{multline}\label{ITOPI}
	\pi_{i}(\hat{\bm{E}}^{(n)}(t)) = \pi_{i}(\hat{\bm{E}}^{(n)}(0))\\
	 + \int_{0}^{t} \sum_{j = 0}^{d+1} n
	\frac{\partial \pi_{i}}{\partial x_{j}}(\bm{E}^{(n)}(s)) F^{(n)}_{j}(\hat{\bm{E}}^{(n)}(s))
	+ \frac{1}{2} \sum_{j,k = 0}^{d+1} na^{(n)}_{jk}(\hat{\bm{E}}^{(n)}(s))
		\frac{\partial \pi_{i}}{\partial x_{j}\partial x_{k}} (\bm{E}^{(n)}(s))\, ds\\
	+ \int_{0}^{t}  \frac{1}{n} \sum_{j = 0}^{d+1} \frac{\partial \pi_{i}}{\partial x_{j}}(\bm{E}^{(n)}(s))\, 
		dM^{(n)}_{j}(ns) + \varepsilon^{(n)}_{i}(nt) 
\end{multline}
where, as before, $a^{(n)}_{jk}(\bm{x})$ is given by \eqref{AA} and $\varepsilon^{(n)}_{i}(nt)$ is a smaller order error term.
 
On first inspection, it might appear that the drift term, which is multiplied by $n$, explodes as $n \to \infty$; however, from the definition of $\bm{\pi}$, we see that $\bm{\pi}(\bm{E}(t,\bm{x})) = \bm{\pi}(\bm{x})$, and thus, 
\[
	0 = \frac{d}{dt}\bigg\vert_{t = 0} \bm{\pi}(\bm{E}(t,\bm{x}))
	 =  \sum_{j = 0}^{d+1} 
	\frac{\partial \pi_{i}}{\partial x_{j}}(\bm{x}) F_{j}(\bm{x}),
\]
and the terms of order $\BigO{n}$ vanish identically.

We can thus replace $\bm{F}^{(n)}$ by $\bm{F}^{(n)} - \bm{F}$ in \eqref{ITOPI}, leaving 
\[
	 \sum_{j = 0}^{d+1} \frac{\partial \pi_{i}}{\partial x_{j}}(\bm{E}^{(n)}(s)) 
		n \left(F^{(n)}_{j}(\hat{\bm{E}}^{(n)}(s)) - F_{j}(\hat{\bm{E}}^{(n)}(s))\right),
\]	
which remains bounded, as our assumptions \eqref{ASSUMPTIONS} guarantee that 
$\bm{F}^{(n)} - \bm{F}$ is $\BigO{\frac{1}{n}}$.
 
 Next, we recall \eqref{eqn:Mn}
 \begin{multline*}
 	\bm{M}^{(n)}(t) =  (\bm{e}_{0} + \bm{e}_{d+1})M^{(n)}_{\bm{e}_{0} + \bm{e}_{d+1}}(t)
	-  (\bm{e}_{0} + \bm{e}_{d+1}) M^{(n)}_{-\bm{e}_{0} - \bm{e}_{d+1}}(t)\\
	+ \sum_{i=1}^{d} (\bm{e}_{i} - \bm{e}_{0}) M^{(n)}_{-\bm{e}_{0} + \bm{e}_{i}}(t)
	- \sum_{i=1}^{d} (\bm{e}_{i} + \bm{e}_{d+1}) M^{(n)}_{-\bm{e}_{i} - \bm{e}_{d+1}}(t)
	- \sum_{i=1}^{d} \bm{e}_{i} M^{(n)}_{-\bm{e}_{i}}(t)
	- \bm{e}_{d+1} M^{(n)}_{-\bm{e}_{d+1}}(t),
 \end{multline*}
 where, for example, 
 \[
	M^{(n)}_{-\bm{e}_{0} - \bm{e}_{d+1}}(t) = \tilde{P}_{-\bm{e}_{0} - \bm{e}_{d+1}}\left(\int_{0}^{t} \delta^{(n)} S^{(n)}(s)\, ds\right).
\]
Thus, 
 \begin{align*}
	\frac{1}{n} M^{(n)}_{-\bm{e}_{0} - \bm{e}_{d+1}}(nt) &= \tilde{P}_{-\bm{e}_{0} - \bm{e}_{d+1}}\left(\int_{0}^{nt} \delta^{(n)} S^{(n)}(s)\, ds\right)\\
	&= \frac{1}{n} \tilde{P}_{-\bm{e}_{0} - \bm{e}_{d+1}}\left(n^{2} \int_{0}^{t} \delta^{(n)} \bar{S}^{(n)}(n s)\, ds\right).
\end{align*}
The latter is a stochastic process with jumps of order $\frac{1}{n}$ and variance 
\[
	\int_{0}^{t} \delta^{(n)} \bar{S}^{(n)}(n s)\, ds = 
		\int_{0}^{t} \delta^{(n)} \hat{E}^{(n)}_{0}(s)\, ds.
\]
Thus, as $n \to \infty$, $\frac{1}{n} M^{(n)}_{-\bm{e}_{0} - \bm{e}_{d+1}}(nt)$ approaches a continuous stochastic process with variance 
\[
	\int_{0}^{t} \delta \hat{E}_{0}(s)\, ds.
\]
The martingale central limit theorem (see \eg \citet{Ethier+Kurtz86}) tells us that the only stochastic process with these properties is a Brownian motion with the same variance, 
\[
	\int_{0}^{t} \sqrt{\delta \hat{E}_{0}(s)}\, dB_{-\bm{e}_{0} - \bm{e}_{d+1}}(s).
\]
(\ie $B_{-\bm{e}_{0} - \bm{e}_{d+1}}(t)$ is a standard Brownian motion with mean 0 and variance $t$).

Proceeding similarly, in the limit, we may replace all the terms $M^{(n)}_{\bm{l}}$ with integrals of  independent Brownian motions, so that as $n \to \infty$,  $\frac{1}{n} \bm{M}^{(n)}(nt)$ aproaches
 \[
 	\int_{0}^{t} \sigma(\hat{\bm{E}}(s))\, d\bm{B}(s)
\]
where
\[
	\bm{B}(t) = (B_{\bm{e}_{0} \!+\! \bm{e}_{d\!+\!1}}(t),B_{\!-\!\bm{e}_{0} \!-\! \bm{e}_{d\!+\!1}}(t),
	B_{S,1}(t),B_{1,\!-\!}(t),B_{1,R},(t),\ldots,B_{S,m}(t),B_{m,\!-\!}(t),B_{m,R}(t),B_{\!-\!\bm{e}_{d\!+\!1}}(t))
\]
is an ordered list of the $D$ Brownian motions corresponding to the $D$ noises $M^{(n)}_{\bm{l}}(t)$ and $\bm{\sigma}(\bm{x})$ is as in the statement. Taking the limit as $n \to \infty$ on both sides of \eqref{ITOPI} and recalling that $\bm{\pi}(\hat{\bm{E}}(t)) = \hat{\bm{E}}(t),$  we obtain \eqref{WEAKLIMIT}.
\end{proof}

While the drift terms seem rather mysterious, they may be interpreted geometrically.  We first observe that 

\begin{prop}
$(\bm{D}\bm{\pi})(\bm{x})$ is the projection onto the tangent space to $\Omega$ at $\bm{x}$,  $T_{\bm{x}} \Omega$.
\end{prop}

\begin{proof}
We first observe that, since $\bm{\pi}(\bm{x}) \in \Omega$ for all $\bm{x}$, we must have 
\[
	\bm{\pi}(\bm{\pi}(\bm{x})) = \bm{\pi}(\bm{x}).
\]
If, moreover, $\bm{x} \in \Omega$, we also have $\bm{\pi}(\bm{\pi}(\bm{x})) = \bm{x}$, so taking derivatives on left and right, using the chain rule, we have that
\[
	(\bm{D}\bm{\pi})(\bm{\pi}(\bm{x}))(\bm{D}\bm{\pi})(\bm{x}) = \mathbb{I},
\]
where $\mathbb{I}$ denotes the identity matrix.  Now, since $\bm{x} \in \Omega$, the right hand side is equal to 
\[
	(\bm{D}\bm{\pi})(\bm{x})(\bm{D}\bm{\pi})(\bm{x}),
\]
so we have that $(\bm{D}\bm{\pi})(\bm{x})$ is a projection.  It remains to see that it is a projection onto the tangent space.  We will do so by showing it's image contains, and is contained by, the tangent space.  

For the former, we recall that a vector $\bm{X}$ is in the tangent space to $\Omega$ if and only if there exists a parametric curve $\sigma_{\bm{x},\bm{X}}(t)$ such that 
\begin{enumerate}[(i)]
\item $\sigma_{\bm{x},\bm{X}}(0) = \bm{x}$,
\item $\dot{\sigma}_{\bm{x},\bm{X}}(0) = \bm{X}$, and,
\item  $\sigma_{\bm{x},\bm{X}}(t) \in \Omega$ for all $t \in \mathbb{R}$.
\end{enumerate}
We then have $\pi(\sigma_{\bm{x},\bm{X}}(t)) = \sigma_{\bm{x},\bm{X}}$, and thus
\[
	(\bm{D}\bm{\pi})(\bm{x}) \bm{X} = \frac{d}{dt}\bigg\vert_{t=0} \pi(\sigma_{\bm{x},\bm{X}}(t))
	 = \dot{\sigma}_{\bm{x},\bm{X}}(0) = \bm{X},
\]
and thus $T_{\bm{x}} \Omega$ is in the image of $(\bm{D}\bm{\pi})(\bm{x})$.  

On the other hand, since $\bm{\pi}(\bm{x}) \in \Omega$, we have $\bm{F}(\bm{\pi}(\bm{x})) = \bm{0}$, and again, taking derivatives using the chain rule, we have 
\[
	(\bm{D}\bm{F})(\bm{\pi}(\bm{x})) (\bm{D}\bm{\pi})(\bm{x}) = \bm{0},
\]
so that if $\bm{x} \in \Omega$, we have $(\bm{D}\bm{F})(\bm{x}) (\bm{D}\bm{\pi})(\bm{x}) = \bm{0}$ and thus 
\[
	(\bm{D}\bm{F})(\bm{x}) (\bm{D}\bm{\pi})(\bm{x})\bm{X} = \bm{0}	
\]
\ie the image of $(\bm{D}\bm{\pi})(\bm{x})$ is contained in the kernel of $(\bm{D}\bm{F})(\bm{x})$, which we have already observed is $T_{\bm{x}} \Omega$.  Thus, $\text{Im}\left((\bm{D}\bm{\pi})(\bm{x})\right) = T_{\bm{x}} \Omega$.
\end{proof}

Thus, the drift vector $(\bm{D}\bm{\pi})(\bm{x})\bm{f}(\bm{x})$ from \eqref{PROJECTED} %\
%\[
%	\bm{b}(\bm{x}) \defn 
%	\begin{bmatrix} s_{1}(\bm{x}) x_{1} \\ \vdots \\ s_{d}(\bm{x}) x_{d} \end{bmatrix}
%\]
is the projection of the vector $\bm{f}(x)$
%\[
%	\hat{\bm{b}}(\bm{x}) \defn
%	\begin{bmatrix} \hat{s}_{1}(\bm{x}) x_{1} \\ \vdots \\ \hat{s}_{d}(\bm{x}) x_{d} \end{bmatrix}
%\]
onto the tangent space to $\Omega$.  %looking at It\^o's formula, \eqref{ITOPI}, we see that
%\[
%	  \bm{b}(\bm{x}) = (\bm{D}\bm{\pi})(\bm{x})\hat{\bm{b}}(\bm{x}).
%\]
This is an immediate consequence of the strong drift: in the absence of constraints, the process would
move (on average) in the direction of this vector, whose components are the relative fitness of each strain, multiplied by the density of that strain.  However, density limitation prevents unlimited growth, confining the process to the manifold $\Omega$, and thus the direction of motion to the tangent plane, and the strains experience a drift that is the best approximating vector to their unconstrained growth rates.
\subsubsection{Computing the derivatives of $\bm{\pi}$}

To complete our derivation of the equations for the limiting process $\hat{\bm{E}}(t)$, we must compute the derivatives of the $\pi_{i}$.  
\begin{prop}
Let $\bm{x} \in \Omega$ and $0\le i\le d+1$. The first partial derivatives of $\bm{\pi_i}$ at $\bm{x}$ are given by $\frac{\partial \pi_{i}}{\partial x_{k}}=0$ if $k=0$ or $d+1$, otherwise by
\[
	\frac{\partial \pi_{i}}{\partial x_{k}} = \mathbbm{1}_{\{i=k\}} -  \frac{(\beta_{k}-\alpha_{k})\beta_{i} x_{i}}{\sum_{j=1}^{m} (\beta_{j}-\alpha_{j})\beta_{j} x_{j}}.
\]
The second partial derivatives $\bm{\pi_i}$ at $\bm{x}$ are given for any $k,n$ both different from $0$ and $d+1$, by:
\begin{multline*}
	\frac{\partial^{2} \pi_{i}}{\partial x_{k}\partial x_{n}} 
	= \frac{\beta_{i}}{\sum_{j=1}^{m} (\beta_{j}-\alpha_{j}) \beta_{j} x_{j}}\left(-(\beta_{k}-\alpha_{k}) \mathbbm{1}_{\{n=i\}}-(\beta_{n}-\alpha_{n}) \mathbbm{1}_{\{k=i\}}\right.\\
	+\left.\frac{(\beta_{k}-\alpha_{k})(\beta_{n}-\alpha_{n}) x_{i}}{\sum_{j=1}^{m}  (\beta_{j}-\alpha_{j}) \beta_{j} x_{j}} \left(\beta_{k}+\beta_{n}+\beta_{i}
	-\frac{\sum_{j=1}^{m} (\beta_{j}-\alpha_{j}) \beta_{j}^{2} x_{j}}{\sum_{j=1}^{m} (\beta_{j}-\alpha_{j}) \beta_{j} x_{j}}\right)\right).
\end{multline*}
\end{prop}
\begin{proof}
We recall that under the weak selection hypothesis, 
\[
	\dot{I}_{i}(t) = \beta_{i} \left(\frac{S(t)}{N(t)} -\frac{1}{R^{\star}_{0}}\right) I_{i}(t),
\]
so that
\[
	\frac{dI_{i}}{dI_{j}} = \frac{\beta_{i} I_{i}}{\beta_{j} I_{j}}.
\]
We can solve this to obtain 
\[
	\frac{1}{\beta_{i}} \ln{\left(\frac{I_{i}(t)}{I_{i}(0)}\right)} = \frac{1}{\beta_{j}} \ln{\left(\frac{I_{j}(t)}{I_{j}(0)}\right)},
\]
for all $i, j$, \ie
\[
	\frac{1}{\beta_{i}} \ln{\left(\frac{I_{i}(t,\bm{x})}{x_{i}}\right)} = \frac{1}{\beta_{j}} \ln{\left(\frac{I_{j}(t,\bm{x})}{x_{j}}\right)},
\]
and, taking the limit as $t \to \infty$,
\begin{equation}\label{FLOWS}
		\frac{1}{\beta_{i}} \ln{\left(\frac{\pi_{i}(\bm{x})}{x_{i}}\right)} = \frac{1}{\beta_{j}} \ln{\left(\frac{\pi_{j}(\bm{x})}{x_{j}}\right)}.
\end{equation}
Taking derivatives, we then have 
\begin{gather*}
	\frac{1}{\beta_{i}} \left(\frac{1}{\pi_{i}} \frac{\partial \pi_{i}}{\partial x_{k}} - \frac{1}{x_{i}} \mathbbm{1}_{\{k=i\}} \right) = 
		\frac{1}{\beta_{j}} \left(\frac{1}{\pi_{j}} \frac{\partial \pi_{j}}{\partial x_{k}} - \frac{1}{x_{j}} \mathbbm{1}_{\{k=j\}}\right),\\
	\frac{1}{\beta_{i} \pi_{i}} \frac{\partial \pi_{i}}{\partial x_{0}} = \frac{1}{\beta_{j} \pi_{j}} \frac{\partial \pi_{j}}{\partial x_{0}}\\
	\frac{1}{\beta_{i} \pi_{i}} \frac{\partial \pi_{i}}{\partial x_{d+1}} = \frac{1}{\beta_{j} \pi_{j}} \frac{\partial \pi_{j}}{\partial x_{d+1}},	
\intertext{and}
\begin{multlined}
	\frac{1}{\beta_{i}} \left(-\frac{1}{\pi_{i}^{2}} \frac{\partial \pi_{i}}{\partial x_{k}} \frac{\partial \pi_{i}}{\partial x_{n}} + \frac{1}{\pi_{i}} \frac{\partial^{2} \pi_{i}}{\partial x_{k}\partial x_{n}}
		+ \frac{1}{x_{i}^{2}} \mathbbm{1}_{\{k=i\}}\mathbbm{1}_{\{n=i\}} \right)\\
	= \frac{1}{\beta_{j}} \left(-\frac{1}{\pi_{j}^{2}} \frac{\partial \pi_{j}}{\partial x_{k}} \frac{\partial \pi_{j}}{\partial x_{n}} + \frac{1}{\pi_{j}} \frac{\partial^{2} \pi_{j}}{\partial x_{k}\partial x_{n}}
		+ \frac{1}{x_{j}^{2}} \mathbbm{1}_{\{j=k\}}\mathbbm{1}_{\{j=n\}} \right).
	\end{multlined}
\end{gather*}

Moreover, using \eqref{OMEGA} we have that 
\[
	\sum_{i=1}^{m} (\beta_{i}-\alpha_{i}) \pi_{i}(\bm{x}) = \lambda (R_{0}^{\star} - 1), 
\]
so that 
\[
	\sum_{i=1}^{m} (\beta_{i}-\alpha_{i}) \frac{\partial \pi_{i}}{\partial x_{k}} 
	= \sum_{i=1}^{m} (\beta_{i}-\alpha_{i}) \frac{\partial \pi_{i}}{\partial x_{0}} 
	= \sum_{i=1}^{m} (\beta_{i}-\alpha_{i}) \frac{\partial \pi_{i}}{\partial x_{d+1}} = 0 
\]
and
\[
	\sum_{i=1}^{m} (\beta_{i}-\alpha_{i}) \frac{\partial^{2} \pi_{i}}{\partial x_{k}\partial x_{n}} = 0. 
\]
Together, these equations give us systems of linear equations that may be solved for the various derivatives of $\pi_{i}(\bm{x})$.  To illustrate, consider $\frac{\partial \pi_{i}}{\partial x_{0}}$; from the above, we have that 
\[
	0 = \sum_{i=1}^{m} (\beta_{i}-\alpha_{i}) \frac{\partial \pi_{i}}{\partial x_{0}} 
	= (\beta_{1}-\alpha_{1}) \frac{\partial \pi_{1}}{\partial x_{0}} + \sum_{i=2}^{m} (\beta_{i}-\alpha_{i}) \frac{\beta_{i} \pi_{i}}{\beta_{1} \pi_{1}} \frac{\partial \pi_{1}}{\partial x_{0}}
	= \left(\sum_{i=1}^{m} (\beta_{i}-\alpha_{i}) \frac{\beta_{i} \pi_{i}}{\beta_{1} \pi_{1}}\right) \frac{\partial \pi_{1}}{\partial x_{0}},
\]
whence $\frac{\partial \pi_{1}}{\partial x_{0}} = 0$, and thus $\frac{\partial \pi_{i}}{\partial x_{0}} = 0$ for all $i$.  Proceeding in the same manner, we find $\frac{\partial \pi_{i}}{\partial x_{d+1}} = 0$ as well, and thus that all second derivatives of $\pi_{i}(\bm{x})$ involving $x_{0}$ or $x_{d+1}$ vanish identically, whilst
\[
	\frac{\partial \pi_{i}}{\partial x_{k}} = \frac{\pi_{i}}{x_{i}}  \mathbbm{1}_{\{i=k\}} 
		- \frac{ \pi_{k}}{ x_{k}} \frac{(\beta_{k}-\alpha_{k})\beta_{i}\pi_{i}}{\sum_{j=1}^{m} (\beta_{j}-\alpha_{j})\beta_{j} \pi_{j}}.
\]
We shall only need to evaluate these for $\bm{x} \in \Omega$, where $\hat{\bm{E}}(t)$ is trapped.  For such $\bm{x}$, the first derivatives simplify to the expression given in the statement, since $\bm{\pi}(\bm{x}) = \bm{x}$ for $\bm{x} \in \Omega$.
Similar calculations lead to the second partial derivatives.\end{proof}

\subsubsection{Reduced Diffusion}

We can use the results of the previous section to provide semi-explicit expressions for the SDE satisfied by $\hat{\bm{E}}$ and displayed in Proposition \ref{prop:weaklimit}.
\begin{prop}
Unlike the full stochastic SIR model, the weak selection limit $\hat{\bm{E}}$ can be completely characterised by a system of equations that depend only on the variables $\hat{I}_{1}, \ldots, \hat{I}_{d}$:
\begin{equation}\label{WEAKLIMIT2}
	d\hat{I}_{i}  = s_{i}(\bm{I}(t))\hat{I}_{i}(t)\, dt 
	+ \frac{1}{\sqrt{R_{0}^{\star}}}
		\sum_{k=1}^{m} \left(\mathbbm{1}_{\{i=k\}} - \frac{(\beta_{k}-\alpha_{k})\beta_{i}\hat{I}_{i}(t)}
		{\sum_{j=1}^{m} (\beta_{j}-\alpha_{j})\beta_{j}\hat{I}_{j}(t)}\right)\sqrt{2 \beta_{k}\hat{I}_{k}(t)} \, dB_{k}(t).
\end{equation}
where
\begin{equation}\label{SELECTION}
	s_{i}(\bm{x}) =   
	\hat{s}_{i}(\bm{x})
	- \frac{\beta_{i} x_{i}}{\sum_{j=1}^{d} (\beta_{j}-\alpha_{j}) \beta_{j} x_{j}}
	 \sum_{j=1}^{d} (\beta_{j}-\alpha_{j}) \hat{s}_{j}(\bm{x})
%	 \right.\\
%	\left.- \frac{2}{I_{e}(\bm{p})}
%	\frac{\sum_{m=1}^{d} \beta_{m}p_{m}}{\sum_{m=1}^{d}(\beta_{m}-\alpha_{m}) \beta_{m} p_{m}}
%	\left((\beta_{i}-\alpha_{i})-
%	\frac{ \sum_{m=1}^{d} (\beta_{m}-\alpha_{m})^{2} \beta_{m} p_{m}}
%	{\sum_{m=1}^{d} (\beta_{m}-\alpha_{m}) \beta_{m} p_{m}}\right)
%	\right)	
\end{equation}
for
\begin{equation}\label{SELECTION2}
	\hat{s}_{i}(\bm{x}) \defn \frac{\beta_{i}}{R_{0}^{\star}} \left(r_{i} 
	- \frac{1}{\sum_{j=1}^{d} (\beta_{j}-\alpha_{j}) \beta_{j} x_{j}}
	\left(2(\beta_{i} - \alpha_{i})
	- \frac{ \sum_{j=1}^{d} (\beta_{j}-\alpha_{j})^{2} \beta_{j} x_{j}}
	{\sum_{j=1}^{d} (\beta_{j}-\alpha_{j}) \beta_{j} x_{j}}\right)\right),
\end{equation}
and $dB_{1}(t),\ldots,dB_{m}(t)$ are independent Brownian motions.
\end{prop}
\begin{proof}
In the previous section, we observed that $\frac{\partial \pi_{i}}{\partial x_{0}} = \frac{\partial \pi_{i}}{\partial x_{d+1}} = 0$, and thus any second partial derivative with respect to $x_{0}$ or $x_{d+1}$ vanishes as well.   Moreover, for $i = 1,\ldots,d$, 
\begin{align*}
	n\left(F^{(n)}_{i}(\bm{x}) - F_{i}(\bm{x})\right) 
	&= n\left(\beta^{(n)}_{i} \left(\frac{x_{0}}{x_{d+1}} - \frac{1}{R^{(n)}_{0,i}}\right) x_{i}
	- \beta_{i} \left(\frac{x_{0}}{x_{d+1}} - \frac{1}{R_{0}^{\star}}\right)x_{i}\right)\\
	&\to - \frac{\beta_{i} r_{i} x_{i}}{R_{0}^{\star}},
\end{align*}
Moreover, for $\bm{x} \in \Omega$, we have 
\[
	\frac{x_{0}}{x_{d+1}} = \frac{1}{R_{0}^{\star}} 
	\left(= \frac{\delta + \alpha_{j} + \gamma_{j}}{\beta_{j}}\right), 
\]	
so that, from \eqref{AA}, we obtain in the limit $n \to \infty$ 
\[
	a_{jk}(\bm{x}) = \begin{cases}
		\frac{2 \beta_{i} x_{j}}{R_{0}^{\star}} & \text{if $j = k$, and}\\
		0 & \text{otherwise.}
	\end{cases}
\]
Similarly, for $1 \leq j,k \leq d$,  $\sigma_{jk}(\bm{x})$ depends on $x_{0}$ or $x_{d+1}$ only via 
the ratio $\frac{x_{0}}{x_{d+1}}$ which is identically equal to $\frac{1}{R_{0}^{\star}}$ on $\bm{x} \in \Omega$.  

Substituting these and the first derivatives into \eqref{WEAKLIMIT} allows us to complete the description of the weak limit $\hat{\bm{E}}(t)$, exploiting the fact that the triples of Brownian motions $B_{-\bm{e}_{0} + \bm{e}_{i}}$, $B_{-\bm{e}_{i} - \bm{e}_{d+1}}$ and $B_{-\bm{e}_{i}}$ and $B_{S,j}$, $B_{j,-}$ and $B_{j,R}$ are independent for $i \neq j$ to combine each triple into a single Brownian motion.
\end{proof}

\subsubsection{Frequency Process}

Repeating the argument of Section \ref{DERIVATION}, we can use the functions $\Pi_{i}$ to finding an equation for the frequency of strain $i$,
\[
	P_{i}(t) = \frac{\hat{I}_{i}(t)}{\sum_{j = 1}^{d} \hat{I}_{j}(t)} 
	%\left(= \lim_{n \to \infty} \frac{I^{(n)}_{i}(t)}{\sum_{j = 1}^{d} I^{(n)}_{j}(t)}.\right),
\]
where, because the limiting process is a diffusion, the standard It\^o formula applies.  We omit the lengthy calculations this entails, and present simply the result.  

For our process $\bm{P}(t)$, we find that
\begin{multline}\label{REDUCEDFREQUENCY}
	dP_{i}(t) = b_{i}(\bm{P}(t))\, dt 
	+  \frac{1}{\sqrt{R_{0}^{\star}}} \frac{1}{\sqrt{I_{e}(\bm{P}(t))}}  
		\frac{1}{\sum_{j=1}^{d} (\beta_{j}-\alpha_{j}) \beta_{j} P_{j}(t)}\\
		\times
	\sum_{j=1}^{m} (\mathbbm{1}_{\{i=j\}} - P_{i}(t)) \sum_{k=1}^{m} (\beta_{k}-\alpha_{k})
	\left(\beta_{k}P_{k}(t)\sqrt{2\beta_{j}P_{j}(t)}\, dB_{j}(t) 
	- \beta_{j}P_{j}(t)\sqrt{2\beta_{k}P_{k}(t)}\, dB_{k}(t)\right).
\end{multline}
where
\[
	b_{i}(\bm{p}) \defn p_{i}\left(s_{i}(I_{e}(\bm{p}) \bm{p}) 
		- \sum_{m=1}^{d} s_{m}(I_{e}(\bm{p}) \bm{p}) p_{m}\right),
\]
for $\bm{s}(\bm{x})$ as defined by \eqref{SELECTION} and \eqref{SELECTION2},
and where, if $\bm{p} \in \Delta_{d}$ corresponds to the point $\bm{x} \in \Omega$, \ie  
\[
	p_{i} \defn \frac{x_{i}}{\sum_{j = 1}^{d} x_{j}}, 
\]
then
\[
	I_{e}(\bm{x}) = \sum_{j = 1}^{d} x_{j}.
\]
Writing 
\[
	p_{i} \defn \frac{x_{i}}{\sum_{j = 1}^{d} x_{j}} = \frac{x_{i}}{I_{e}(\bm{x})}, 
\]
 and recalling that 
\[
	\lambda (R_{0}^{\star} - 1) = 
	\sum_{i=1}^{d} (\beta_{i}-\alpha_{i}) x_{i} = 
	\sum_{i=1}^{d} (\beta_{i}-\alpha_{i}) I_{e}(\bm{x}) p_{i}
\]
we see that we can explicitly express $I_{e}$ as a function of $\bm{p}$:
\[
	I_{e}(\bm{p}) 
	= \frac{\lambda (R_{0}^{\star} - 1)}{\sum_{i=1}^{d} (\beta_{i}-\alpha_{i}) p_{i}}.
\]

\begin{rem}
The notation above has been deliberately chosen to recall the Wright-Fisher diffusion in population genetics, with $s_{i}(\bm{p})$ and $I_{e}(\bm{p})$ a frequency-dependent selection coefficient and an effective population size, respectively.  

To understand the motivation for the notation $I_{e}$, which is meant to recall the effective population sizes used in population genetics, we note that 
\[
	I_{e} =  \lim_{n \to \infty} \frac{1}{n} \sum_{j = 1}^{d} I^{(n)}_{j}(t),
\]
so that $I_{e}(\bm{x})$ is approximately the total density of infected individuals when the diffusion limit is at the point $\bm{x} \in \Omega$, or, equivalently, when the frequencies of the various strains is $\bm{p}$. 
\end{rem}

\begin{rem} 
As before, vector $\bm{b}(\bm{p})$ may be interpreted geometrically as the projection of the vector 
\[
	\begin{bmatrix} s_{1}(I_{e}(\bm{p}) \bm{p}) p_{1} \\ \vdots \\ s_{d}(I_{e}(\bm{p}) \bm{p}) p_{d} 
		\end{bmatrix}
\]
onto the simplex $\Delta_{d}$. 
\end{rem}

%\[
%	\sum_{j = 0}^{d+1} \frac{\partial\pi_{i}}{\partial x_{j}} f_{j}(\hat{\bm{E}})
%	+ \frac{1}{2} \sum_{j = 0}^{d+1} \sum_{k=0}^{d+1} a_{jk}(\hat{\bm{E}})
%		\frac{\partial^{2}\pi_{i}}{\partial x_{j}\partial x_{k}} 
%\]

\subsubsection{Results for $d = 2$}

If we have $d = 2$ strains, then, since $P_{1}(t) + P_{2}(t) = 1$, it is sufficient to consider the frequency of the invading strain, strain 2.  Writing $P(t) \defn P_{2}(t)$, the results of the previous section tell us that the generator of $P(t)$\footnote{The generator of $\bm{P}(t)$ is the operator on the space of continuous functions on the $d$-simplex 
\[
	\Delta_{d} = \left\{\bm{p} : \sum_{i = 1}^{d} p_{i} = 1\right\}
\]
defined by 
\[
	\mathcal{L}f(\bm{p}) \defn 
	\lim_{t \downarrow 0} 
		\frac{\mathbb{E}\left[ f(\bm{P}(t)) \middle\vert \bm{P}(0) = \bm{p}\right] - f(\bm{p})}{t}.
\]
We recall that if the diffusion process $\bm{P}(t)$ has SDE
\[
	d\bm{P}(t) = \bm{b}(\bm{P}(t)))\, dt + \bm{\varsigma}(\bm{P}(t)))\, d\bm{B}(t),
\]
then 
\[
	\mathcal{L}f(\bm{p})  = \sum_{i = 1}^{d} b_{i}(\bm{p}) \frac{\partial f}{\partial p_{i}} 
	+ \frac{1}{2} \sum_{i = 1}^{d} \sum_{i = 1}^{d} a_{ij}(\bm{p}) 
		\frac{\partial^{2} f}{\partial p_{i} \partial p_{j}} 
\]
where $\bm{a}(\bm{x}) = \bm{\varsigma}(\bm{x})\bm{\varsigma}(\bm{x})^{\top}$ is the variance-covariance matrix for $d\bm{P}(t)$, and the probability density function for $\bm{P}(t)$, say $f(t,\bm{p})$, satisfies the Kolmogorov backward equation
\[
	\frac{\partial}{\partial t} f(t,\bm{p}) = \mathcal{L} f(t,\bm{p}).
\]}
is 
\[
	\mathcal{L}f(p) = b(b) f'(p) + \frac{1}{2} a(p) f''(p),
\]
where
\begin{multline*}
	b(p) \defn %\frac{1}{R_{0}^{\star}}
%(r_{2}-r_{1})\frac{\beta_{2}\beta_{1}\left((\beta_{2}-\alpha_{2})p+(\beta_{1}-\alpha_{1})(1-p)\right)}
%{\left((\beta_{2}-\alpha_{2})\beta_{2}p+(\beta_{1}-\alpha_{1})\beta_{1}(1-p)\right)}p(1-p)\\
%- \frac{1}{R_{0}^{\star}} \frac{1}{I_{e}(p)} 
%\frac{\beta_{2}\beta_{1}\left((\beta_{2}-\alpha_{2})p+(\beta_{1}-\alpha_{1})(1-p)\right)
%\left((\beta_{2}-\alpha_{2})\beta_{2}-(\beta_{1}-\alpha_{1})\beta_{1}\right)}
%{\left((\beta_{2}-\alpha_{2})\beta_{2}p+(\beta_{1}-\alpha_{1})\beta_{1}(1-p)\right)^{2}} p(1-p)\\
%- \frac{1}{R_{0}^{\star}} \frac{1}{I_{e}(p)} 
%\frac{\beta_{2}^{2}\beta_{1}^{2}
%\left((\beta_{2}-\alpha_{2})-(\beta_{1}-\alpha_{1})\right)
%\left((\beta_{2}-\alpha_{2})p+(\beta_{1}-\alpha_{1})(1-p)\right)^{2}}
%{\left((\beta_{2}-\alpha_{2})\beta_{2}p+(\beta_{1}-\alpha_{1})\beta_{1}(1-p)\right)^{3}} p(1-p)\\
%+ \frac{2}{R_{0}^{\star}}\frac{1}{I_{e}(p)} 
%\frac{\beta_{2}\beta_{1}\left((\beta_{2}-\alpha_{2})-(\beta_{1}-\alpha_{1})\right)
%\left((\beta_{2}-\alpha_{2})p+(\beta_{1}-\alpha_{1})(1-p)\right)(\beta_{2}p+\beta_{1}(1-p))}
%{\left((\beta_{2}-\alpha_{2})\beta_{2}p+(\beta_{1}-\alpha_{1})\beta_{1}(1-p)\right)^{2}}p(1-p)\\
	\frac{1}{R_{0}^{\star}}
	(r_{2}\!-\!r_{1})\frac{\beta_{2}\beta_{1}\left((\beta_{2}\!-\!\alpha_{2})p\!+\!(\beta_{1}\!-\!\alpha_{1})(1\!-\!p)\right)}
	{\left((\beta_{2}\!-\!\alpha_{2})\beta_{2}p\!+\!(\beta_{1}\!-\!\alpha_{1})\beta_{1}(1\!-\!p)\right)}p(1\!-\!p)\\
	\!-\! \frac{1}{R_{0}^{\star}} \frac{1}{I_{e}(p)}
	\frac{\beta_{2}\beta_{1}
	\left((\beta_{2}\!-\!\alpha_{2})p\!+\!(\beta_{1}\!-\!\alpha_{1})(1\!-\!p)\right)(\beta_{2}p\!+\!\beta_{1}(1\!-\!p))}
	{\left((\beta_{2}\!-\!\alpha_{2})\beta_{2}p\!+\!(\beta_{1}\!-\!\alpha_{1})\beta_{1}(1\!-\!p)\right)^{2}} 
	\times\left(\left((\beta_{2}\!-\!\alpha_{2})\beta_{2}\!-\!(\beta_{1}\!-\!\alpha_{1})\beta_{1}\right)\right.\\
	\left. \!-\!\left((\beta_{2}\!-\!\alpha_{2})\!-\!(\beta_{1}\!-\!\alpha_{1})\right)\left(\frac{\beta_{2}\beta_{1}
	\left((\beta_{2}\!-\!\alpha_{2})p\!+\!(\beta_{1}\!-\!\alpha_{1})(1\!-\!p)\right)}
	{\left((\beta_{2}\!-\!\alpha_{2})\beta_{2}p\!+\!(\beta_{1}\!-\!\alpha_{1})\beta_{1}(1\!-\!p)\right)}
	\!-\! 2(\beta_{2}p\!+\!\beta_{1}(1\!-\!p))\right)\right) 
	p(1\!-\!p)
\end{multline*}
and
\[
	a(p) \defn \frac{2}{R_{0}^{\star}}\frac{1}{I_{e}(p)} 
	\frac{\beta_{2}\beta_{1}\left((\beta_{2}-\alpha_{2})p+(\beta_{1}-\alpha_{1})(1-p)\right)^{2}
	(\beta_{2}p+\beta_{1}(1-p))}
	{\left((\beta_{2}-\alpha_{2})\beta_{2}p+(\beta_{1}-\alpha_{1})\beta_{1}(1-p)\right)^{2}}p(1-p).
\]

The generator allows us to compute many quantities of interest for the process $P(t)$.  In particular, if $h(p)$ is the probability of fixation of strain 1 given $P(0)=p$, then $h(p)$ satisfies the boundary problem
\begin{gather*}
	\mathcal{L}h(p) = 0\\
	h(0) = 0\\
	h(1) =1 
\end{gather*}
(see \eg \citet{Ewens1979,Durrett2009,Etheridge2011}).  This may be solved to give 
\begin{equation}\label{WEAKAPPROX}
	h(p) = \frac{\int_{0}^{p} e^{-2\int \frac{b(q)}{a(q)}\, dq}\, dq}
		{\int_{0}^{1} e^{-2\int \frac{b(q)}{a(q)}\, dq}\, dq}. 
\end{equation}

Let $\tilde{h}(p)$ be the numerator of this fraction.  Substituting the expressions for $a(p)$ and $b(p)$ and some simplification yields 
\begin{multline*}
	\tilde{h}(p) \defn \int_{0}^{p} \frac{(\beta_{2}q\!+\!\beta_{1}(1\!-\!q))^{\!-\!\frac{(\beta_{2}\!+\!\beta_{1})
		\left((\beta_{2}\!-\!\alpha_{2})\!-\!(\beta_{1}\!-\!\alpha_{1})\right)}
		{\beta_{2}\alpha_{1}\!-\!\beta_{1}\alpha_{2}}}
		\left((\beta_{2}\!-\!\alpha_{2})q\!+\!(\beta_{1}\!-\!\alpha_{1})(1\!-\!q)\right)
		^{2\!-\!\frac{\left((\beta_{2}\!-\!\alpha_{2})\beta_{2}\!-\!(\beta_{1}\!-\!\alpha_{1})\beta_{1}\right)}
		{\beta_{2}\alpha_{1}\!-\!\beta_{1}\alpha_{2}}}}
		{\left((\beta_{2}\!-\!\alpha_{2})\beta_{2}q\!+\!(\beta_{1}\!-\!\alpha_{1})\beta_{1}(1\!-\!q)\right)}\\
	\times e^{\!-\! (r_{2} \!-\! r_{1}) \int_{0}^{q} I_{e}(u)
	\frac{\left((\beta_{2}\!-\!\alpha_{2})\beta_{2}u\!+\!(\beta_{1}\!-\!\alpha_{1})\beta_{1}(1\!-\!u)\right)}
	{\left((\beta_{2}\!-\!\alpha_{2})u\!+\!(\beta_{1}\!-\!\alpha_{1})(1\!-\!u)\right)
	(\beta_{2}u\!+\!\beta_{1}(1\!-\!u))}\, du}\, dq,
\end{multline*}
whereas 
\[
	h(p) = \frac{\tilde{h}(p)}{\tilde{h}(1)}.
\]
For ease of notation, we will write
\[
	 \phi(q) =  \int_{0}^{q} I_{e}(u)
	\frac{\left((\beta_{2}-\alpha_{2})\beta_{2}u+(\beta_{1}-\alpha_{1})\beta_{1}(1-u)\right)}
	{\left((\beta_{2}-\alpha_{2})u+(\beta_{1}-\alpha_{1})(1-u)\right)
	(\beta_{2}u+\beta_{1}(1-u))}\, du
\]
and
\[
	g(q) = \frac{(\beta_{2}q+\beta_{1}(1-q))^{-\frac{(\beta_{2}+\beta_{1})
		\left((\beta_{2}-\alpha_{2})-(\beta_{1}-\alpha_{1})\right)}
		{\beta_{2}\alpha_{1}-\beta_{1}\alpha_{2}}}
		\left((\beta_{2}-\alpha_{2})q+(\beta_{1}-\alpha_{1})(1-q)\right)
		^{2-\frac{\left((\beta_{2}-\alpha_{2})\beta_{2}-(\beta_{1}-\alpha_{1})\beta_{1}\right)}
		{\beta_{2}\alpha_{1}-\beta_{1}\alpha_{2}}}}
		{\left((\beta_{2}-\alpha_{2})\beta_{2}q+(\beta_{1}-\alpha_{1})\beta_{1}(1-q)\right)},
\]
so that 
\[
	\tilde{h}(p) = \int_{0}^{p} g(q) e^{-(r_{2} - r_{1}) \phi(q)}\, dq.
\]

We can evaluate this expression numerically, but we will be particularly interested in a number of special cases, when we can obtain analytical approximations to $\tilde{h}(p)$.  

\begin{enumerate}[(i)]
\item When $r_{2} = r_{1}$ (or, more generally, when $R_{0,i} 
= R_{0}^{\star}\left(1+ o\left(\frac{1}{n}\right)\right)$) we can give an explicit closed form for $\tilde{h}(p)$, and thus $h(p)$:
\begin{multline}\label{WEAKNOSTRONG}
	\tilde{h}(p) \propto \frac{(\beta_{2}p\!+\!\beta_{1}(1\!-\!p))^{1\!-\!\frac{(\beta_{2}\!+\!\beta_{1})
		\left((\beta_{2}\!-\!\alpha_{2})\!-\!(\beta_{1}\!-\!\alpha_{1})\right)}
		{\beta_{2}\alpha_{1}\!-\!\beta_{1}\alpha_{2}}}
		\left((\beta_{2}\!-\!\alpha_{2})p\!+\!(\beta_{1}\!-\!\alpha_{1})(1\!-\!p)\right)
		^{3\!-\!\frac{\left((\beta_{2}\!-\!\alpha_{2})\beta_{2}\!-\!(\beta_{1}\!-\!\alpha_{1})\beta_{1}\right)}
		{\beta_{2}\alpha_{1}\!-\!\beta_{1}\alpha_{2}}}}
		{\left((\beta_{2}\!-\!\alpha_{2})p\!+\!(\beta_{1}\!-\!\alpha_{1})(1\!-\!p)\right)}\\
		\!-\! \frac{(\beta_{2}\!+\!\beta_{1})^{1\!-\!\frac{(\beta_{2}\!+\!\beta_{1})
		\left((\beta_{2}\!-\!\alpha_{2})\!-\!(\beta_{1}\!-\!\alpha_{1})\right)}
		{\beta_{2}\alpha_{1}\!-\!\beta_{1}\alpha_{2}}}
		\left((\beta_{2}\!-\!\alpha_{2})\!+\!(\beta_{1}\!-\!\alpha_{1})\right)
		^{3\!-\!\frac{\left((\beta_{2}\!-\!\alpha_{2})\beta_{2}\!-\!(\beta_{1}\!-\!\alpha_{1})\beta_{1}\right)}
		{\beta_{2}\alpha_{1}\!-\!\beta_{1}\alpha_{2}}}}
		{\left((\beta_{2}\!-\!\alpha_{2})\!+\!(\beta_{1}\!-\!\alpha_{1})\right)}.
\end{multline}
This expression is not, however, especially illuminating.  

\begin{figure}[h] 
    \centering
    \includegraphics[width=0.65\linewidth]{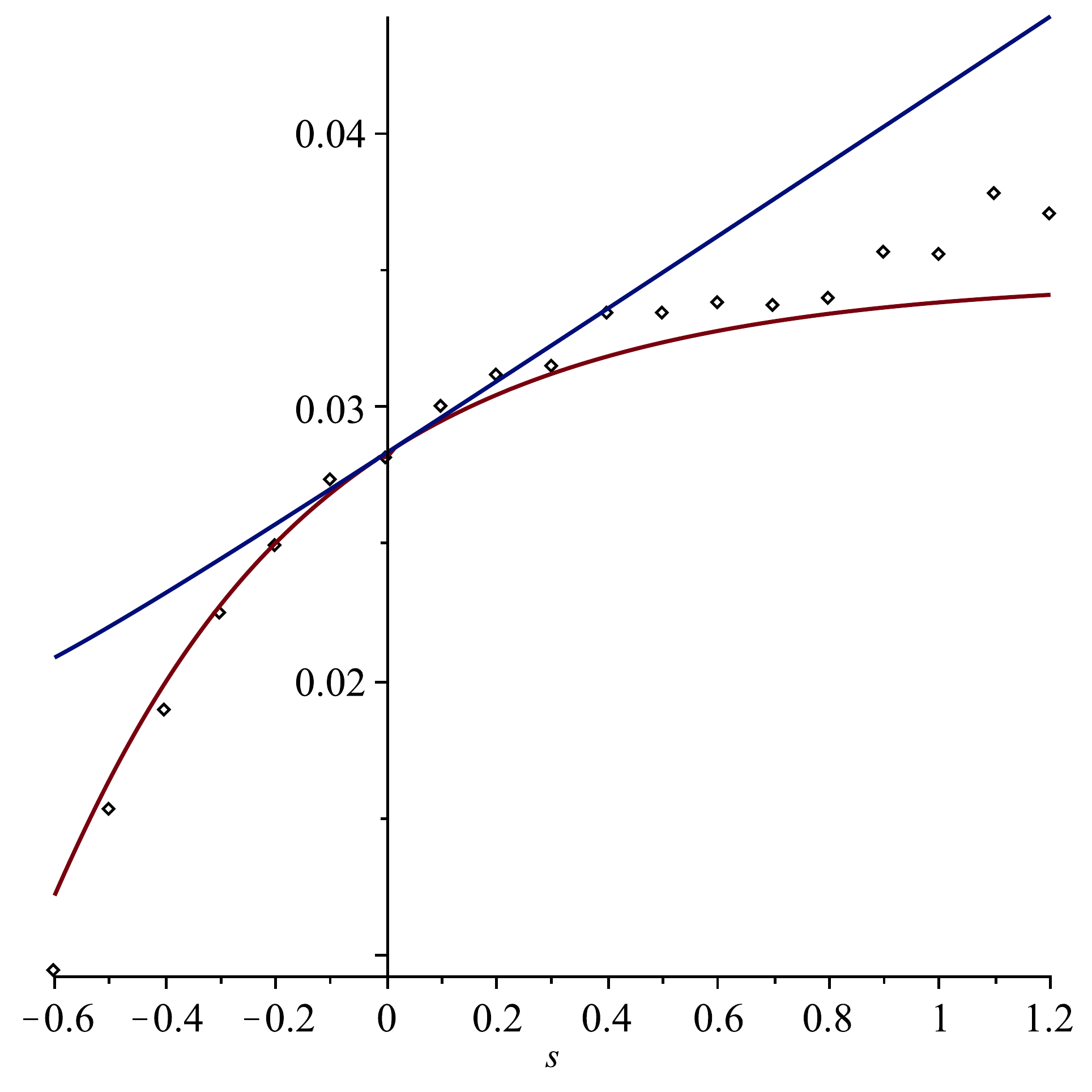} 
    \caption{We compare \eqref{WEAKNOSTRONG} (red curve) and its order $p^2$ approximation \eqref{WEAKAPPROX2} (blue line) with fixation probabilities for a single mutant invader obtained via simulating the Markov chain with rates given by Table \ref{RATES} (black diamonds).  We vary $\beta_{2}$ by setting $\beta_{1} = \beta_{2}(1+\sigma)$, while holding $R_{0,1} = R_{0,2}$ by setting $\alpha_{2} = \frac{\delta+\alpha_{1}+\gamma}{1 + \sigma} - \delta - \gamma$.  The other parameters are fixed at $n = 100$, $R_{0}^{\star} = 4$, $\lambda = 2$, $\delta = 1$, $\beta_{1} = 20$, and $\alpha_{1} = 3$.} 
   \label{WEAKAPPROXFIGURE} 
\end{figure}

\item We shall principally be interested in the case when $p$ is small, in which case we can Taylor expand $\tilde{h}(p)$ as 
\begin{align*}
	\tilde{h}(p) &= \tilde{h}(0) + \tilde{h}'(0) p + \frac{1}{2} \tilde{h}''(0) p^{2} + \BigO{p^{3}}\\
	&= g(0) p +  \frac{1}{2} (g'(0) + g(0) (r_{2} - r_{1}) \phi'(0)) p^{2} + \BigO{P^{3}}\\
	&= g(0)\left(p+\frac{1}{2} \left(\frac{g'(0)}{g(0)} + (r_{2} - r_{1}) I_{e}(0)\right) p^{2}\right) 
	+ \BigO{p^{3}}
\end{align*}

Unfortunately, this does not yield an estimate of the normalising constant, $\tilde{h}(1)$.  To obtain this, 
we consider the case when $\beta_{2} - \beta_{1}$ and $\alpha_{2} - \alpha_{1}$ are small.  While this is a restrictive assumption, it will allow us to consider the long-term evolution in the framework of adaptive dynamics, where mutational changes are assumed to be very small.  To this end, we introduce $\sigma$ and $\theta$ such that 
\[
	\beta_{2} = \beta_{1} (1+\sigma) \quad \text{and} \quad \alpha_{2} = \alpha_{1} (1+\theta).
\] 

We then have that 
\[
	\frac{g'(0)}{g(0)} = \sigma + \BigO{2},  
	\frac{g(p)}{g(0)} = 1 - \sigma p + \BigO{2},
	\phi'(p) = - I_{e}(p) + \BigO{2},
\]
and
\begin{align*}
	\frac{\tilde{h}(1)}{g(0)} 
	&= \int_{0}^{1} (1 - \sigma p) e^{-(r_{2}-r_{1}) \int_{0}^{p} I_{e}(q)\, dq}\, dp + \BigO{2}\\
	&= 1 - \frac{\sigma}{2} - (r_{2}-r_{1}) I_{e}(0)  +\BigO{2},
\end{align*}
where $\BigO{2}$ is used to denote terms of order $\BigO{\sigma^{2}}$, $\BigO{\theta^{2}}$, or 
$\BigO{\sigma\theta}$.  

Then, to order $p^{2}$, the fixation probability is 
\begin{equation}\label{WEAKAPPROX2}
	h(p) \defn \frac{\tilde{h}(p)}{\tilde{h}(1)} = p + \frac{1}{2}(\sigma + (r_{2}-r_{1}) I_{e}(0))p(1-p)
	+\BigO{2},
\end{equation}
which may be written informally in terms of the original parameters as  
\[
	h(p) \defn \frac{\tilde{h}(p)}{\tilde{h}(1)} = p + 
	\frac{1}{2} \left(\frac{\beta_{1}-\beta_{2}}{\beta_{2}} 
	+ \left(1 - \frac{R^{(n)}_{0,1}}{R^{(n)}_{0,2}}\right) n I_{e}(0)\right)p(1-p)
	+\BigO{2}
\]
to lowest order.  Here we have used 
\[
	r_{2}-r_{1} =  n \frac{R^{(n)}_{0,2}- R^{(n)}_{0,1}}{R_{0}^{\star}} + o(1)
	= n \frac{R^{(n)}_{0,2}- R^{(n)}_{0,1}}{R^{(n)}_{0,2}}
		\left(1+\frac{r_{2}}{n} + o(1)\right) + o(1).
\]

In practice, we are most interested in the case when a single individual carries the mutant strain, so 
$p = \frac{1}{n I_{e}(0)} = \frac{1}{I^{(n)}(0)}$.  While our proofs -- which assume that $p$ is independent of $n$, and thus that the number of invading individuals is proportional to $n$ -- do not justify taking this value for $p$, we find that the expression for the fixation probability obtained by taking $p = \frac{1}{I^{(n)}(0)}$, which to lowest order is
\begin{equation}\label{PFIXWEAK}
	\frac{1}{I^{(n)}(0)} + 
	\frac{1}{2} \left( \frac{1}{I^{(n)}(0)} \frac{\beta_{1}-\beta_{2}}{\beta_{2}} 
	+ 1 - \frac{R^{(n)}_{0,1}}{R^{(n)}_{0,2}}\right), 
\end{equation}
agrees extremely well with simulations (Figure \ref{STRONGAPPROXFIGURE:c}) -- another example of the so-called ``unreasonable effectiveness of mathematics' \citep{Wigner1960} -- and will use it to investigate the long term evolution of the virulence in Section \ref{AD}.

\end{enumerate}

%\textcolor{blue}{In the simulations I only tried to check $\eqref{PFIXWEAK}$ . But it may be good to compare simulation results with the above result. This may help obtain a better fit when the strength of selection is higher, right?}.

\subsubsection{Transient Dynamics}

Whilst in previous sections, we have considered the case when the invader arrives when the resident is near to its endemic equilibrium, the projection results in \citet{Katzenberger1991} allow us to consider the process given any initial condition in $(\mathbb{R}_{+}^{2})^{\circ}$: when started from an arbitrary point in this set, say $(x,\bm{y},z)$, the process jumps instantaneously to the point $\bm{\pi}(x,\bm{y},z) \in \Omega$, at which point the relative frequencies of the two strains evolve according to \eqref{REDUCEDFREQUENCY}. In particular, the fixation probability of the novel strain 2 is given by 
\[
	h(p(x,\bm{y},z)),
\]
where $p(x,\bm{y},z)$ is the ``post-projection'' frequency of strain 2,
\[
	p(x,\bm{y},z) = \frac{\pi_{2}(x,\bm{y},z)}{\pi_{2}(x,\bm{y},z)+\pi_{1}(x,\bm{y},z)}
	= \left(1+\frac{\pi_{1}(x,\bm{y},z)}{\pi_{2}(x,\bm{y},z)}\right)^{-1} 
\]

Unfortunately, a closed analytical expression for $\bm{\pi}(x,\bm{y},z)$ is not available, but we can nonetheless characterise it implicitly via \eqref{OMEGA} and \eqref{FLOWS}, which for $d = 2$
give us a pair of equations that may be solved numerically to yield $\bm{\pi}(x,\bm{y},z)$:
\[
	(\beta_{1}-\alpha_{1}) \pi_{1}(x,\bm{y},z) + (\beta_{2}-\alpha_{2}) \pi_{2}(x,\bm{y},z)
	= \lambda (R_{0}^{\star} - 1)
\]
and 
\[
		\frac{1}{\beta_{2}} \ln{\left(\frac{\pi_{2}(x,\bm{y},z)}{y_{2}}\right)} 
		= \frac{1}{\beta_{1}} \ln{\left(\frac{\pi_{1}(x,\bm{y},z)}{y_{1}}\right)}.
\]
Rewriting this as 
\[
	\frac{\pi_{2}(x,\bm{y},z)^{\beta_{1}}}{\pi_{1}(x,\bm{y},z)^{\beta_{2}}} 
		= \frac{y_{2}^{\beta_{1}}}{y_{1}^{\beta_{2}}},
\]
we note that there are at least two cases when one can solve these analytically.  If $\beta_{2} = \beta_{1}$, then we have 
\[
	\frac{\pi_{2}(x,\bm{y},z)}{\pi_{1}(x,\bm{y},z)} = \frac{y_{2}}{y_{1}},
\]
so the relative frequencies of remain unchanged in the jump to the diffusion:
\[
	\frac{y_{2}}{y_{2}+y_{1}} = \left(1+\frac{y_{1}}{y_{2}}\right)^{-1} 
	= \left(1+\frac{\pi_{1}(x,\bm{y},z)}{\pi_{2}(x,\bm{y},z)}\right)^{-1} 
	= p(x,\bm{y},z)
\]	

Alternately, if $\beta_{2} = 2\beta_{1}$, then 
\begin{equation}\label{RATIO2}
	\pi_{1}(x,\bm{y},z) = \frac{y_{1}^{2}}{y_{2}} \pi_{2}(x,\bm{y},z)^{2},
\end{equation}
and substituting this into \eqref{OMEGA} gives a quadratic equation for $\pi_{2}(x,\bm{y},z)$:
\[
	 (\beta_{1}-\alpha_{1}) y_{1}^{2} \pi_{2}(x,\bm{y},z)^{2}
	 + (\beta_{2}-\alpha_{2}) y_{2} \pi_{2}(x,\bm{y},z)
	 - \lambda (R_{0}^{\star} - 1) y_{2} = 0,
\]
which may be solved to yield
\[
	\pi_{2}(x,\bm{y},z) = \frac{-(\beta_{2}-\alpha_{2}) y_{2} 
		+ \sqrt{(\beta_{2}-\alpha_{2})^{2} y_{2}^{2} + 4\lambda (R_{0}^{\star} - 1) 
		(\beta_{1}-\alpha_{1}) y_{2} y_{1}^{2}}}{2 (\beta_{1}-\alpha_{1}) y_{1}^{2}},
\]
and thus, using \eqref{RATIO2}
\[
	\frac{\pi_{1}(x,\bm{y},z)}{\pi_{2}(x,\bm{y},z)}
	= \frac{y_{1}^{2}}{y_{2}} \pi_{2}(x,\bm{y},z)	
	= \frac{-(\beta_{2}-\alpha_{2}) y_{2} 
		+ \sqrt{(\beta_{2}-\alpha_{2})^{2} y_{2}^{2} + 4\lambda (R_{0}^{\star} - 1) 
		(\beta_{1}-\alpha_{1}) y_{2} y_{1}^{2}}}{2 (\beta_{1}-\alpha_{1}) y_{2}}
\]
and 
\[
	 p(x,\bm{y},z) = \left(1-\frac{(\beta_{2}-\alpha_{2}) y_{2} 
		- \sqrt{(\beta_{2}-\alpha_{2})^{2} y_{2}^{2} + 4\lambda (R_{0}^{\star} - 1) 
		(\beta_{1}-\alpha_{1}) y_{2} y_{1}^{2}}}{2 (\beta_{1}-\alpha_{1}) y_{2}}\right)^{-1}.
\] 

Similarly, if $\beta_{1} = 2\beta_{2}$, we obtain a quadratic equation for $\pi_{1}(x,\bm{y},z)$, which we may solve to obtain 
\[
	 p(x,\bm{y},z) = 1-\left(1-\frac{(\beta_{1}-\alpha_{1}) y_{1} 
		- \sqrt{(\beta_{1}-\alpha_{1})^{2} y_{1}^{2} + 4\lambda (R_{0}^{\star} - 1) 
		(\beta_{2}-\alpha_{2}) y_{1} y_{2}^{2}}}{2 (\beta_{2}-\alpha_{2}) y_{1}}\right)^{-1}.
\] 

We use these latter expressions to observe the effect of population size on the probability of fixation; replacing $y_{1}$ and $y_{2}$ by $n_{0} (1-p_{0})$ and $n_{0} p_{0}$, we can vary the initial normalized (by $n$) number of infected individuals, 
\[
	n_{0} = y_{1} + y_{2}, 
\]
while holding the initial frequency of the mutant strain fixed.  This gives us the projected frequency as a function of the initial frequency, $ p(p_{0},n_{0})$. This, as we observed above, may in turn be used in conjunction with \eqref{WEAKAPPROX} to determine the fixation probability of the mutant.  In Figure \ref{WEAKAPPROXTRANSIENTFIGURE}, we plot the fixation probability for mutant strains (solid curves)  starting at $p_{0} = 0.1$ (green), $p_{0} = 0.5$ (red) and  $p_{0} = 0.9$ (blue) as a function of the initial population density $n_{0}$.  The correspondingly coloured dashed lines plot $p_{0}$, which would also be the mutant fixation probability in the absence of selective pressures.  These dashed lines intersect the black dash-dot line at $I_{e}(p(p_{0},n_{0}))$ \ie the point where the trajectory started from $p_{0},n_{0}$ intersects $\Omega$.  All values of $n_{0}$ to the left of this point correspond to initial points below $\Omega$ (\ie closer to the origin), whereas those to the right correspond to points above $\Omega$.  Figure \ref{WEAKAPPROXTRANSIENTFIGURE:a} shows the case $\beta_{1} = 2\beta_{2}$ whereas Figure \ref{WEAKAPPROXTRANSIENTFIGURE:b} shows the case $\beta_{2} = 2\beta_{1}$.  For all values of $p_{0}$, there is a cross-over point below $\Omega$ where the mutant goes from fixation probability that is less ($\beta_{1} = 2\beta_{2}$) or greater ($\beta_{1} = 2\beta_{1}$) than the neutral prediction to greater or less, respectively, than the neutral prediction.  Thus, more virulent/higher contact rate strategies are relatively advantageous in smaller populations, but become disadvantageous as population size approaches equilibrium, whilst conversely, less virulent/lower contact rate strategies are advantageous near equilibrium, when hosts are in limited supply, and disadvantageous in populations where hosts are abundant, consistent with the pattern we saw in the strong selection case, Section \ref{STRONGTRANSIENT} (as we discuss below, one must approach these figures with care for small values of $n_{0}$, as they neglect the possibility of one or other of the strains will go extinct before the population arrives near $\Omega$.)  

\begin{figure}[h] 
	\begin{subfigure}[t]{0.5\linewidth}
    		\centering
 	 		\includegraphics[width=0.85\linewidth]{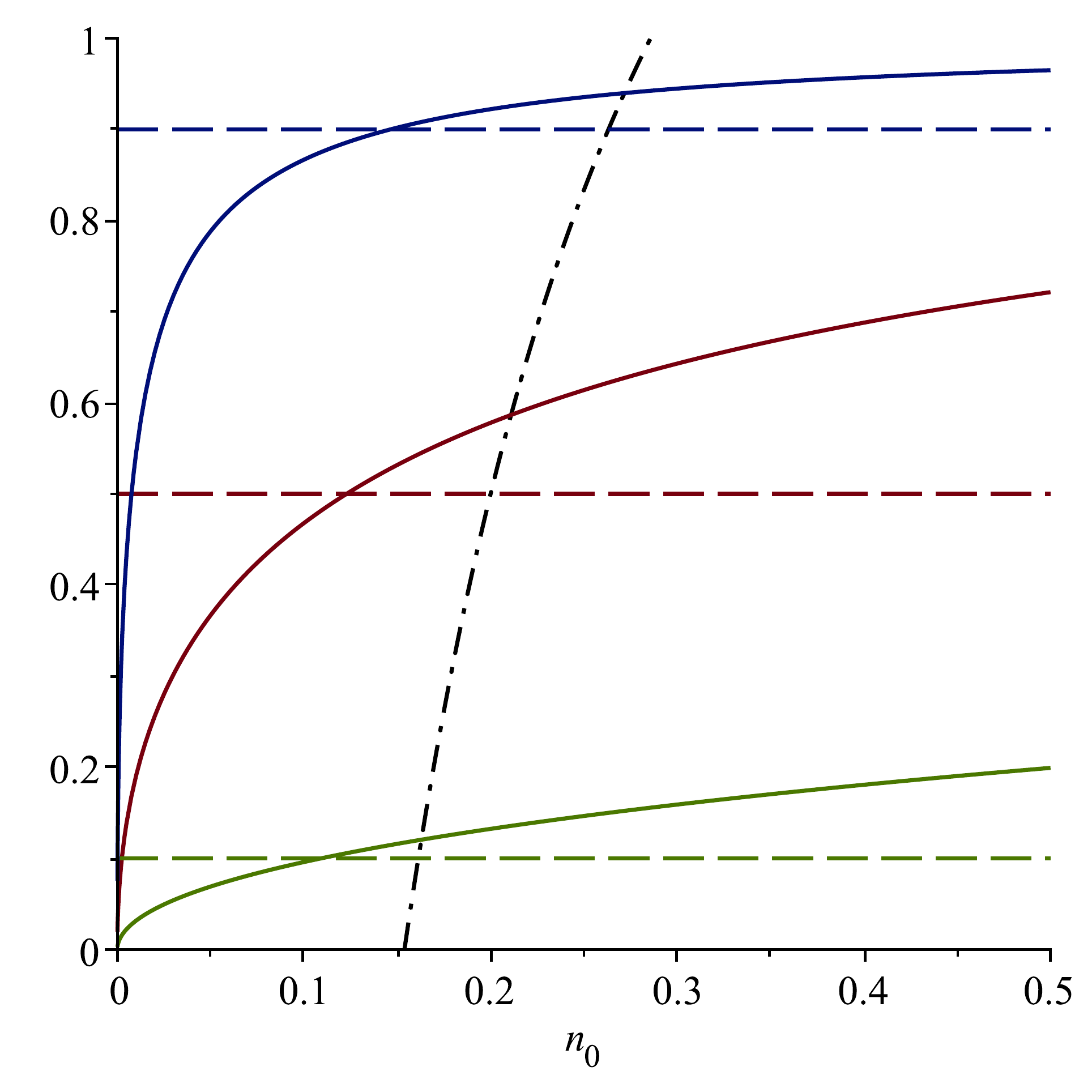} 
    			\caption{$\beta_{1} = 2\beta_{2}$} 
   			\label{WEAKAPPROXTRANSIENTFIGURE:a} 
    		%\vspace{4ex}
  		\end{subfigure}%% 
		\begin{subfigure}[t]{0.5\linewidth}
    		\centering
 	 		\includegraphics[width=0.85\linewidth]{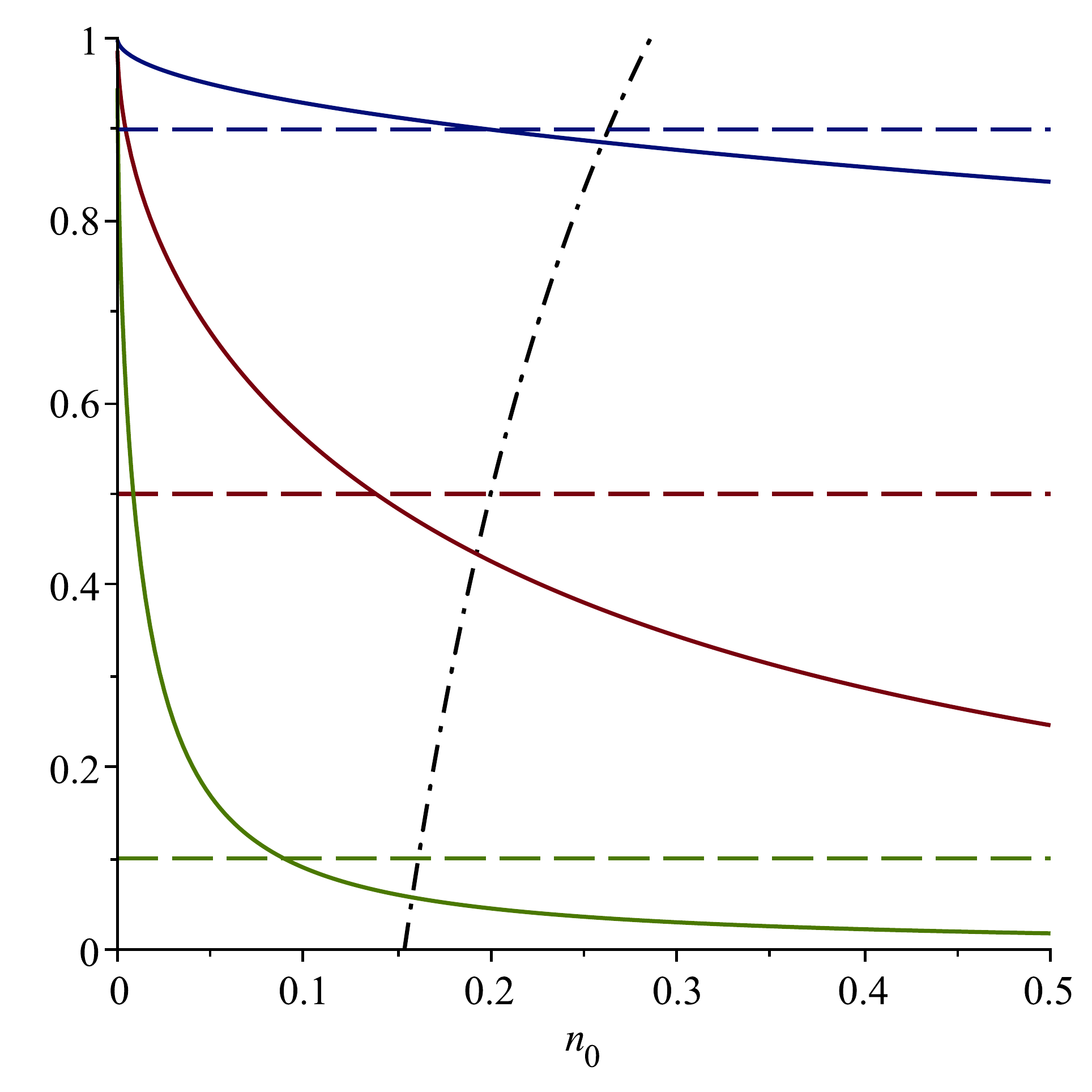} 
    			\caption{$\beta_{2} = 2\beta_{1}$} 
   			\label{WEAKAPPROXTRANSIENTFIGURE:b} 
    		%\vspace{4ex}
  		\end{subfigure}%% 

    \caption{We plot the fixation probability of mutant strains with initial frequencies $p_{0} = 0.1$ (green), $p_{0} = 0.5$ (red) and  $p_{0} = 0.9$ (blue) as a function of the initial population size.  The black dash-dot curve indicates the position of the manifold $\Omega$. We hold $R_{0,1} = R_{0,2}$ by setting $\alpha_{2} = \frac{\delta+\alpha_{1}+\gamma}{1 + \sigma} - \delta - \gamma$.  The other parameters are fixed at $R_{0}^{\star} = 2$, $\lambda = 2$, $\delta = 1$, $\beta_{1} = 10$, $\alpha_{1} = 3$, and $\gamma_{1}=\gamma_{2}=1$.} 
   \label{WEAKAPPROXTRANSIENTFIGURE} 
\end{figure}

\begin{rem}
We briefly note that this simple approach becomes invalid when $I^{(n)}_{i}(0) = o(n)$ for some $n$ 
(and thus $(x,\bm{y},z) \in \partial \mathbb{R}_{+}^{2}$); in this case, stochastic effects will lead to the rapid extinction of the rare strain.  Unfortunately, our previous branching process approach does not translate directly: the heuristic approximation gives an asymptotically critical branching process, whereas our upper and lower bounds become supercritical and subcritical branching processes respectively, so that the resulting ``sandwiching'' of the actual process is uninformative. We are currently considering a refined approach for this case.
\end{rem}

\subsection{Reconciling the Strong and Weak Selection Results}

On first inspection, our expressions for the strong and weak selection limits have little in common.  In Section \ref{STRONG}, we saw that if $R^{(n)}_{0,i}\to R_{0,i}$ and $R_{0,2}
\neq R_{0,1}$, then, if $I^{(n)}_{2}(0) \to I_{2}(0)$, the fixation probability was 
\[
	\begin{cases}
		1-\left(\frac{R_{0,1}}{R_{0,2}}\right)^{I_{2}(0)} 
			& \text{if $R_{0,2} > R_{0,1}$, and}\\
		0 & \text{otherwise}
	\end{cases}
\]
On the other hand, in Section \ref{WEAK}, we assume that 
\[
	R^{(n)}_{0,i}= R_{0}^{\star}
		\left(1+\frac{r_{i}}{n}\right) + {\textstyle o\left(\frac{1}{n}\right)}
\] 
and $\bar{I}^{(n)}(0) = \frac{1}{n} I^{(n)}_{2}(0) \to I_{2}(0)$, and derive a quite different appearing expression for the fixation probability.

More generally, we might consider the intermediary scalings: let $1 \ll \kappa_{n} \ll n$, and suppose that 
\[
	R^{(n)}_{0,i}= R_{0}^{\star}
		\left(1+\frac{r_{i}}{\kappa_{n}}\right) + {\textstyle o\left(\frac{1}{\kappa_{n}}\right)}
\] 
whilst 
\[
	\lim_{n \to \infty} \frac{I^{(n)}_{2}(0)}{\kappa_{n}} = \iota.  
\]
Substituting these into our expression for strong selection, we see that, provided $r_{2} > r_{1}$, we have that the probability of fixation is 
\[
	1-\left(\frac{\left(1+\frac{r_{1}}{\kappa_{n}}\right) + {\textstyle o\left(\frac{1}{\kappa_{n}}\right)}}
	{\left(1+\frac{r_{2}}{\kappa_{n}}\right) + {\textstyle o\left(\frac{1}{\kappa_{n}}\right)}}\right)
	^{\kappa_{n} \frac{I^{(n)}(0)}{\kappa_{n}}} 
		\to 1- e^{-(r_{2}-r_{1})\iota}
\]
as $n \to \infty$ (\nb that when $\kappa_{n} \gg n$, $\frac{I^{(n)}_{2}(0)}{\kappa_{n}} \to 0$, so trivially the probability of fixation is 0).

On the other hand, we can begin with our expression for the fixation probability under weak selection, which written in terms of the original parameters $R^{(n)}_{0,2}$ and $R^{(n)}_{0,1}$, was proportional to 
\begin{multline*}
	\tilde{h}(p) = \int_{0}^{p} \frac{(\beta_{2}q\!-\!\beta_{1}(1\!-\!q))^{\!-\!\frac{(\beta_{2}\!+\!\beta_{1})
		\left((\beta_{2}\!-\!\alpha_{2})\!-\!(\beta_{1}\!-\!\alpha_{1})\right)}
		{\beta_{2}\alpha_{1}\!-\!\beta_{1}\alpha_{2}}}
		\left((\beta_{2}\!-\!\alpha_{2})q\!+\!(\beta_{1}\!-\!\alpha_{1})(1\!-\!q)\right)
		^{2\!-\!\frac{\left((\beta_{2}\!-\!\alpha_{2})\beta_{2}\!-\!(\beta_{1}\!-\!\alpha_{1})\beta_{1}\right)}
		{\beta_{2}\alpha_{1}\!-\!\beta_{1}\alpha_{2}}}}
		{\left((\beta_{2}\!-\!\alpha_{2})\beta_{2}q\!+\!(\beta_{1}\!-\!\alpha_{1})\beta_{1}(1\!-\!q)\right)}\\
	\times e^{\!-\! n\left(\frac{R^{(n)}_{0,2}\!-\! R^{(n)}_{0,1}}{R_{0}^{\star}} 
	 \!+\! {\textstyle o\left(\frac{1}{n}\right)}\right)
	 	\frac{\left((\beta_{2}\!-\!\alpha_{2})\beta_{2}q\!+\!(\beta_{1}\!-\!\alpha_{1})\beta_{1}(1\!-\!q)\right)}
	{\beta_{2}\beta_{1}\left((\beta_{2}\!-\!\alpha_{2})q\!+\!(\beta_{1}\!-\!\alpha_{1})(1\!-\!q)\right)
	(\beta_{2}q\!+\!\beta_{1}(1\!-\!q))}\, dq}\, dq
\end{multline*}

Replacing $R^{(n)}_{0,2}$ and $R^{(n)}_{0,1}$  with the intermediary scalings, our expression becomes
\begin{multline*}
	 \int_{0}^{p} \frac{(\beta_{2}q\!-\!\beta_{1}(1\!-\!q))^{\!-\!\frac{(\beta_{2}\!+\!\beta_{1})
		\left((\beta_{2}\!-\!\alpha_{2})\!-\!(\beta_{1}\!-\!\alpha_{1})\right)}
		{\beta_{2}\alpha_{1}\!-\!\beta_{1}\alpha_{2}}}
		\left((\beta_{2}\!-\!\alpha_{2})q\!+\!(\beta_{1}\!-\!\alpha_{1})(1\!-\!q)\right)
		^{2\!-\!\frac{\left((\beta_{2}\!-\!\alpha_{2})\beta_{2}\!-\!(\beta_{1}\!-\!\alpha_{1})\beta_{1}\right)}
		{\beta_{2}\alpha_{1}\!-\!\beta_{1}\alpha_{2}}}}
		{\left((\beta_{2}\!-\!\alpha_{2})\beta_{2}q\!+\!(\beta_{1}\!-\!\alpha_{1})\beta_{1}(1\!-\!q)\right)}\\
	\times e^{\!-\!\frac{n}{\kappa_{n}}(r_{2} \!-\! r_{1} \!+\!o(1)) \int I_{e}(q)
	\frac{\left((\beta_{2}\!-\!\alpha_{2})\beta_{2}q\!+\!(\beta_{1}\!-\!\alpha_{1})\beta_{1}(1\!-\!q)\right)}
	{\beta_{2}\beta_{1}\left((\beta_{2}\!-\!\alpha_{2})q\!+\!(\beta_{1}\!-\!\alpha_{1})(1\!-\!q)\right)
	(\beta_{2}q\!+\!\beta_{1}(1\!-\!q))}\, dq}\, dq.
\end{multline*}

We can find a large $n$ asymptotic expression for this probability using a pair of lemmas:

\begin{lem}\label{LAPLACE1}
Suppose that $\phi(x)$ and $g(x)$ are an increasing continuously differentiable function and a continuous function on $[a,b]$ ($-\infty < a < b < \infty$) respectively, and that $\phi'(a) \neq 0$.  Then,
\[
	\lim_{M \to \infty}  
	\frac{\int_{a}^{b} g(x) e^{-M \phi(x)}\, dx}{\frac{g(a) e^{-M \phi(a)}}{M \phi'(a)}} = 1.
\]
\end{lem}

\begin{proof}
Fix $\varepsilon > 0$ such that $\phi'(a) > \varepsilon$.  Using Taylor's theorem, we may write 
\[
	\phi(x) = \phi(a) + \phi'(a)(x-a) + R(x)(x-a),
\]
where $R(x) \to 0$ as $x \to a$.  Fix $\delta > 0$ such that 
\[
	|R(x)| < \varepsilon \quad \text{and} \quad |g(x) - g(a)| < \varepsilon 
\]
for all $x$ such that $x - a < \delta$. Finally, choose $\eta > 0$ such that $\phi(x) > \phi(a) + \eta$ for all 
$x$ such that $x - a \geq \delta$ and $B$ such that $|g(x)| < B$ for all $x \in [a,b]$.  Then, 
\begin{align*}
	\int_{a}^{b} g(x) e^{-M \phi(x)}\, dx 
	&= e^{-M \phi(a)}\int_{a}^{b} g(x) e^{-M (\phi(x)-\phi(a))}\, dx\\
	&= e^{-M \phi(a)}\left(\int_{a}^{a+\delta} g(x) e^{-M (\phi(x)-\phi(a))}\, dx
		+ \int_{a+\delta}^{b} g(x) e^{-M (\phi(x)-\phi(a))}\, dx\right).
\end{align*}
Now,
\[
	\abs{\int_{a+\delta}^{b} g(x) e^{-M (\phi(x)-\phi(a))}\, dx}
		\leq \int_{a+\delta}^{b} B e^{-M \eta}\, dx \to 0
\]
as $M \to \infty$, whereas
\begin{multline*}
	(g(a) - \varepsilon)  \int_{a}^{a+\delta} e^{-M (\phi'(a) + \varepsilon)(x-a)}\, dx 
	\leq \int_{a}^{b} g(x) e^{-M (\phi(x)-\phi(a))}\, dx  \\
	\leq (g(a) + \varepsilon)  \int_{a}^{a+\delta} e^{-M (\phi'(a) - \varepsilon)(x-a)}\, dx. 
\end{multline*}
Now, letting $y = M(x-a)$, we have
\[
	 \int_{a}^{a+\delta} e^{-M (\phi'(a) - \varepsilon)(x-a)}\, dx 
	 =  \frac{1}{M}\int_{0}^{M \delta} e^{-(\phi'(a) - \varepsilon)y}\, dy
\]
whilst 
\[
	\int_{0}^{M \delta} e^{-(\phi'(a) - \varepsilon)y}\, dy 
	= \frac{1}{\phi'(a) - \varepsilon}\left(1 - e^{-(\phi'(a) - \varepsilon) M\delta} \right) 
	\to \frac{1}{\phi'(a) - \varepsilon}\
\]
as $M \to \infty$, and similarly for the lower bound.  

Since $\varepsilon > 0$ can be chosen arbitrarily small, the result follows.
\end{proof}

\begin{lem}\label{LAPLACE2}
Let $\phi(x)$, $g(x)$, and $[a,b]$ be as above.  Then,
\[
	\lim_{M \to \infty}  \frac{\int_{a}^{a+\frac{X}{M}} g(x) e^{-M \phi(x)}\, dx}
		{\frac{g(a) e^{-M \phi(a)}}{M \phi'(a)}}
		= 1 - e^{-\phi'(a)X}.
\]
\end{lem}

\begin{proof}
By direct computation, we have
\begin{align*}
\int_{a}^{a+\frac{X}{M}} g(x) e^{-M \phi(x)}\, dx	
	&= \int_{0}^{X} g(a+\frac{y}{M}) e^{-M \phi(a+\frac{y}{M})}\, dy\\
	&= \frac{1}{M} \int_{0}^{X} g(a+\frac{y}{M}) 
		e^{-M (\phi(a)+\phi'(a)\frac{y}{M}+R(a+\frac{y}{M}) \frac{y}{M}}\, dy
\end{align*}
so that as $M \to \infty$, 
\[
	\frac{\int_{a}^{a+\frac{X}{M}} g(x) e^{-M \phi(x)}\, dx}
		{\frac{g(a) e^{-M \phi(a)}}{M \phi'(a)}}
	\to \phi'(a) \int_{0}^{X} e^{-\phi'(a)y}\, dy. 
\]
The result follows.
\end{proof}

To apply the lemmas here, we take $a = 0$ and $b = 1$, $M_{n} = \frac{n}{\kappa_{n}}$, and, as before,
\[
	g(p) =\frac{(\beta_{2}p\!-\!\beta_{1}(1\!-\!p))^{\!-\!\frac{(\beta_{2}\!+\!\beta_{1})
		\left((\beta_{2}\!-\!\alpha_{2})\!-\!(\beta_{1}\!-\!\alpha_{1})\right)}
		{\beta_{2}\alpha_{1}\!-\!\beta_{1}\alpha_{2}}}
		\left((\beta_{2}\!-\!\alpha_{2})p\!+\!(\beta_{1}\!-\!\alpha_{1})(1\!-\!p)\right)
		^{2\!-\!\frac{\left((\beta_{2}\!-\!\alpha_{2})\beta_{2}\!-\!(\beta_{1}\!-\!\alpha_{1})\beta_{1}\right)}
		{\beta_{2}\alpha_{1}\!-\!\beta_{1}\alpha_{2}}}}
		{\left((\beta_{2}\!-\!\alpha_{2})\beta_{2}p\!+\!(\beta_{1}\!-\!\alpha_{1})\beta_{1}(1\!-\!p)\right)}
\]
whereas we now take a slightly different definition for $\phi(p)$, which now has an $o(1)$ correction:
\[
	 \phi(p) = (r_{2} - r_{1} +o(1)) \int I_{e}(p)
	\frac{\left((\beta_{2}-\alpha_{2})\beta_{2}p+(\beta_{1}-\alpha_{1})\beta_{1}(1-p)\right)}
	{\beta_{2}\beta_{1}\left((\beta_{2}-\alpha_{2})p+(\beta_{1}-\alpha_{1})(1-p)\right)
	(\beta_{2}p+\beta_{1}(1-p))}\, dp
\]
so that
\[
	\phi'(0) \to I_{e}(0) (r_{2} - r_{1}) 
\]
as $n \to \infty$.  Then, using Lemma \ref{LAPLACE1} we conclude that $\tilde{h}(1)$ is asymptotically equivalent to 
\[
	\frac{g(0) e^{-M_{n} \phi(0)}}{M_{n} \phi'(0)}.
\]

Now, to consider the numerator when we start with $I^{(n)}_{2}(0) \sim \iota \kappa_{n}$ individuals of the invading strain, we recall that 
\[
	p = \lim_{n \to \infty} P^{(n)}_{2}(0) 
	= \lim_{n \to \infty}\frac{I^{(n)}_{2}(0)}{I^{(n)}_{2}(0) + I^{(n)}_{1}(0)}
	= \lim_{n \to \infty} \frac{I^{(n)}_{2}(0)}{I_{e}(P^{(n)}_{2}(0)) n}  
\]
so, to apply Lemma \ref{LAPLACE2}, we will take 
\[
	X \defn X^{(n)} = M_{n} \frac{I^{(n)}_{2}(0)}{I_{e}(P^{(n)}_{2}(0)) n}
	= \frac{I^{(n)}_{2}(0)}{I_{e}(P^{(n)}(0)) \kappa_{n}} 
\]
\ie so that $\frac{X^{(n)}}{M_{n}} \sim p$.

We note that $P^{(n)}(0) \propto \frac{\kappa_{n} \iota}{n} \to 0$ as $n \to \infty$, so
\[
	\lim_{n \to \infty} X^{(n)} = \frac{\iota}{I_{e}(0)}.
\]
Thus, applying Lemma \ref{LAPLACE2}, we have 
\[
	{\textstyle \tilde{h}\left( \frac{I^{(n)}_{2}(0)}{I_{e}(P^{(n)}_{2}(0)) n}\right)} \sim 
		\frac{g(0) e^{-M_{n} \phi(0)}}{M_{n} \phi'(0)}\left(1 - e^{-(r_{2}-r_{1})\iota}\right),
\]
and the probability of fixation obtained from the weak selection expression is again asymptotic to
\[
	1 - e^{-(r_{2}-r_{1})\iota}. 
\]

While this is not a rigorous proof, it does demonstrate heuristically that the weak and strong selection expressions for the fixation probability agree to first order when applied across the intermediate selective regimes.  In particular, we can use the method of matched asymptotic expansions (see \eg \citet{Hinch91,Kevorkian96}) to combine our two solutions into a single expression valid across all scales, by summing the expressions for strong and weak selection and subtracting their common limit, where all are expressed in the unscaled (\ie strong selection parameters):
\begin{equation}\label{MATCHEDAPPROX}
	\left[1-\left(\frac{R^{(n)}_{0,1}}{R^{(n)}_{0,2}}\right)^{I^{(n)}_{2}(0)}
	+ {\textstyle h\left(\frac{I^{(n)}_{2}(0)}{I_{e}(0) n}\right)} 
	- \left(1 - e^{\left(\frac{R^{(n)}_{0,1}}{R^{(n)}_{0,2}}-1\right)I^{(n)}_{2}(0)}
	\right)\right]^{+},
\end{equation}
where $[x]^{+} = \max\{x,0\}$ and we have used that
\[
	\frac{R^{(n)}_{0,1}}{R^{(n)}_{0,2}}  - 1 
	= \frac{1+\frac{r_{1}}{\kappa_{n}}}{1+\frac{r_{2}}{\kappa_{n}}} - 1
	= \frac{1}{\kappa_{n}} \frac{r_{1}-r_{2}}{1+\frac{r_{2}}{\kappa_{n}}} 
	= \frac{1}{\kappa_{n}} (r_{1}-r_{2}) + \BigO{\frac{1}{\kappa_{n}^{2}}},
\]
so that
\[ 
	(r_{1}-r_{2}) \iota \sim  (r_{1}-r_{2}) \frac{I^{(n)}_{2}(0)}{\kappa_{n}}
		\sim \left(\frac{R^{(n)}_{0,1}}{R^{(n)}_{0,2}}-1\right)I^{(n)}_{2}(0).
\]

We illustrate how these approximations compare to a simulated epidemic in Figure \ref{STRONGAPPROXFIGURE}.

%\textcolor{blue}{What I don't understand here is how do we deal with the case where $R_{0,2} < R_{0,1}$. The first term will become negative... Should we assume that this strong selection term becomes zero (as in the strong selection case) in this situation?} 

\begin{figure}
  \begin{subfigure}[b]{0.5\linewidth}
    \centering
    \includegraphics[width=0.85\linewidth]{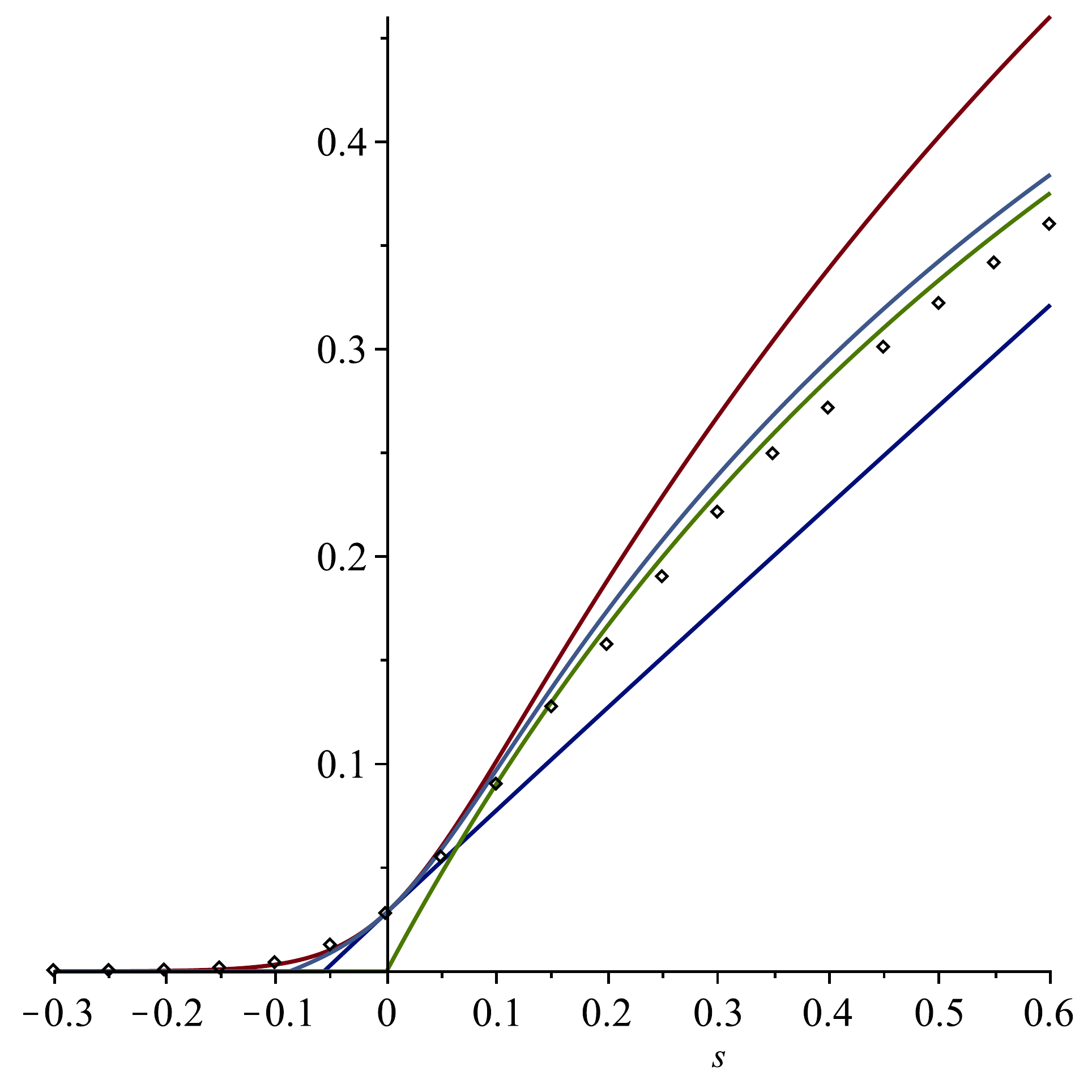} 
    \caption{all approximations} 
   \label{STRONGAPPROXFIGURE:a} 
    \vspace{4ex}
  \end{subfigure}%% 
  \begin{subfigure}[b]{0.5\linewidth}
    \centering
    \includegraphics[width=0.85\linewidth]{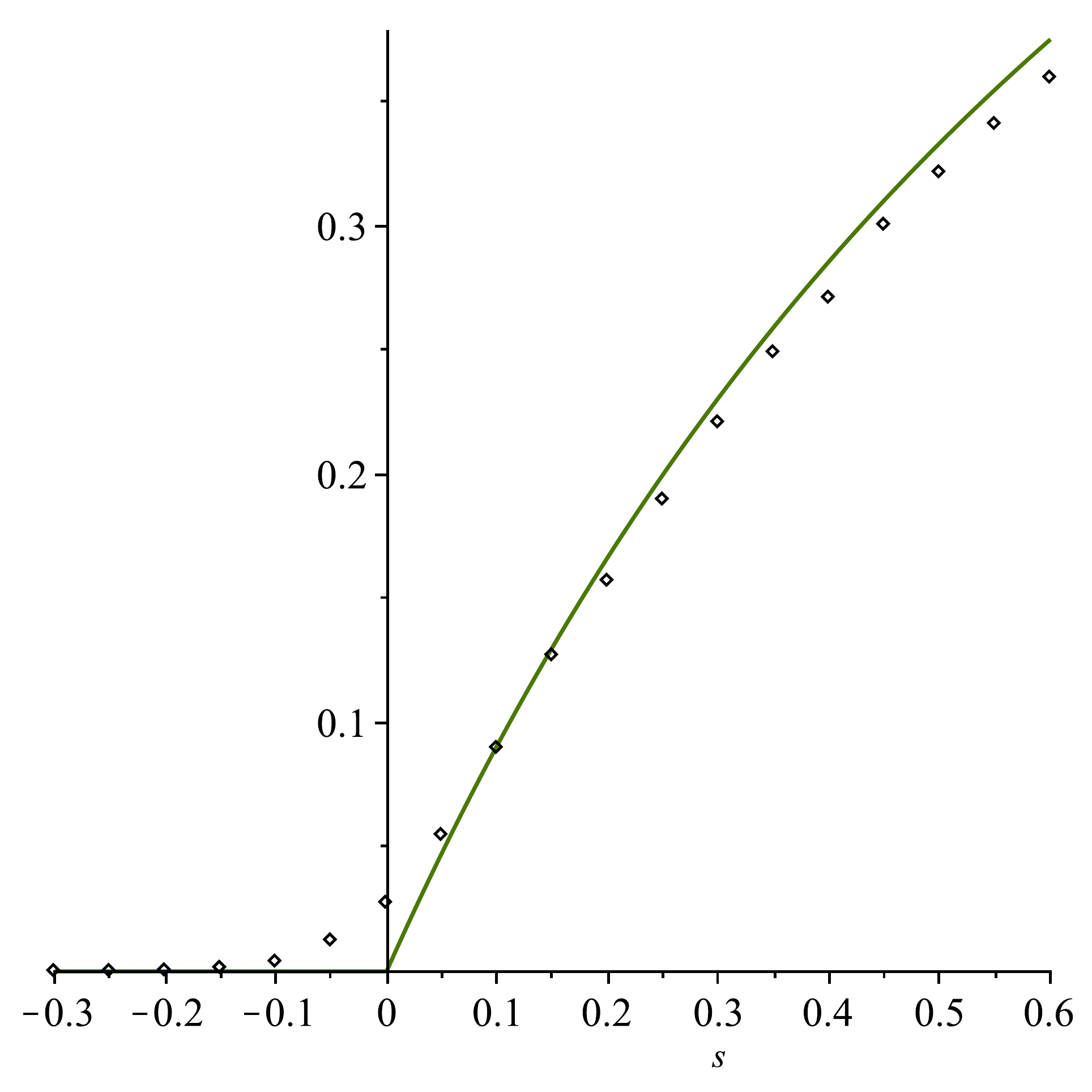} 
    \caption{strong selection} 
    \label{STRONGAPPROXFIGURE:b} 
    \vspace{4ex}
  \end{subfigure} 
  \begin{subfigure}[b]{0.5\linewidth}
    \centering
    \includegraphics[width=0.85\linewidth]{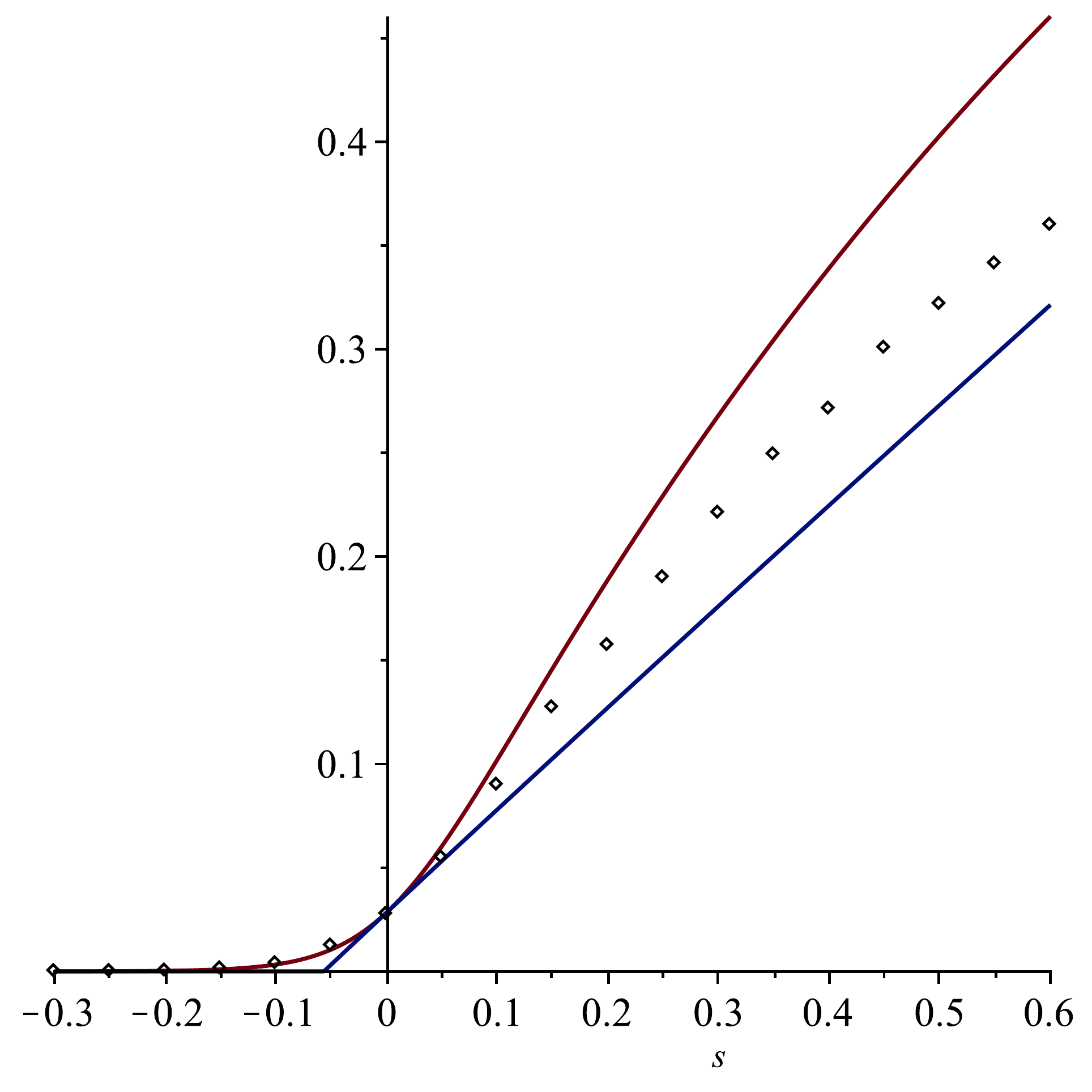} 
    \caption{weak selection} 
    \label{STRONGAPPROXFIGURE:c} 
  \end{subfigure}%%
  \begin{subfigure}[b]{0.5\linewidth}
    \centering
    \includegraphics[width=0.85\linewidth]{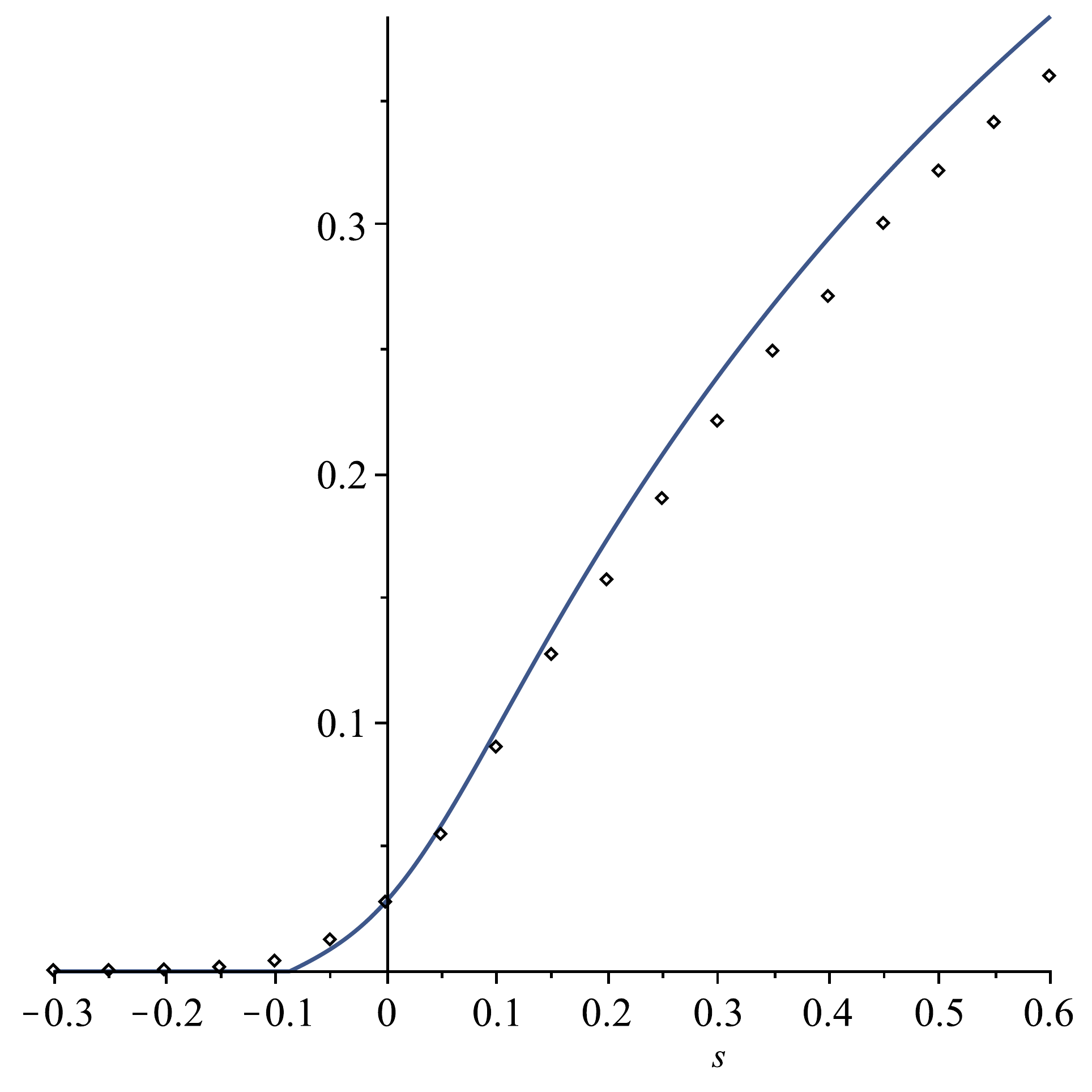} 
    \caption{matched} 
    \label{STRONGAPPROXFIGURE:d} 
  \end{subfigure} 
  \caption{We compare our various approximations to fixation probabilities obtained by Markov chain simulations (black diamonds), where we assume a single mutant invading a resident population at equilibrium.  We have implemented selection on the invading  strain by setting $\beta_{2} = \beta_{1}(1+s)$, so that $R^{(n)}_{0,2}= R^{(n)}_{0,1}(1+s)$.  The other parameters are fixed at $n = 100$, $R^{(n)}_{0,1}= R_{0}^{\star} = 4$, $\lambda = 2$, $\delta = 1$, $\beta_{1} = 20$,  and $\alpha_{1} = 3$.   (b) shows strong selection approximation \eqref{STRONGAPPROX} (green curve), (c) shows the weak selection approximation \eqref{WEAKAPPROX} (red curve) and its second order approximation \eqref{WEAKAPPROX} (blue line), (d) shows the matched  asymptotic approximation \eqref{MATCHEDAPPROX}, whereas (a) overlays all the approximations for comparison.  As would be expected, the strong and weak approximations do well in their corresponding parameter regimes, but poorly elsewhere, whereas the matched asymptotic provides a compromise, performing worse than the weak or strong approximations at the respective extremes, but interpolating between them for intermediate values of $s$.}
  \label{STRONGAPPROXFIGURE} 
\end{figure}

\clearpage

\section{Adaptive Dynamics}\label{AD}

Using the expressions for the fixation probability derived above, we can use the framework of adaptive dynamics to investigate the long-term evolution of strains.  In what follows, we give an informal discussion of the derivation of the canonical \textit{diffusion} for the process, a generalisation of the canonical equation of adaptive dynamics which allows us to consider the influence of random drift on phenotypic evolution.  We refer the reader to \citet{Otto2007} for a more extensive discussion aimed at a biological audience and to \citet{Champagnat+Lambert07} for a mathematically rigorous derivation of the canonical equation.

We briefly recall the assumptions of adaptive dynamics in the context of our epidemic models; throughout, we assume that a novel mutant strain with virulence $\alpha'$ invades a population which is at the endemic equilibrium with a resident strain $\alpha$.   

We assume a tradeoff between transmissibility and virulence, so that the contact rate of a strain depends on its virulence according to some fixed function $\beta(\alpha)$.  For our numerical investigations, we take 
\[
	\beta(\alpha) = (\delta+\alpha+\gamma)(R_{0,\text{max}} - w(\alpha-\alpha_{0})^{2}),
\]
where $w$ is a parameter that determines the ``flatness'' of the fitness landscape.

The reproductive number is a function of the virulence, 
\[
	R_{0}(\alpha) = \frac{\beta(\alpha)}{\delta+\alpha+\gamma}.
\]
We will  assume that there is a value, $\alpha_{0}$, for the virulence that maximises $R_{0}(\alpha)$.

Under these assumptions, the density of ndividuals infected with the resident strain at the endemic equilibrium is
\[
	I_{\text{eq}}(\alpha) \sim \frac{\lambda(R_{0}(\alpha)-1)}{\beta(\alpha)-\alpha}
\]
$I_{\text{eq}}(\alpha)$ is non-zero on a range $(\alpha_{min},\alpha_{max})$; outside of this range, the pathogen goes extinct.

We then have, using \eqref{PFIXWEAK}, that the fixation probability of a mutant strain of virulence $\alpha'$ arising in a single individual in a population in which a strain of virulence $\alpha'$, is
\[
	S(\alpha,\alpha') \sim \frac{1}{n I_{\text{eq}}(\alpha)} + 
	\frac{1}{2} \left(\frac{1}{n I_{\text{eq}}(\alpha)} \frac{\beta(\alpha)-\beta(\alpha')}{\beta(\alpha')} 
	+ 1 - \frac{R_{0}(\alpha)}{R_{0}(\alpha')}\right) + \BigO{\abs{\alpha-\alpha'}^{2}}.
\]

To introduce the evolutionary dynamics, we assume that mutations occur in individuals with virulence $\alpha$ at a per-capita rate $\epsilon \eta(\alpha)$, where $\epsilon > 0$ is a dimensionless parameter that we will take to 0.  This will ensure that, with high probability, fixation occurs before a second novel mutation can arise.  The population is thus assumed to be monomorphic (\ie all individuals have the same virulence) between invasion events.  

Finally, we assume mutations have small effects, and are unbiased in direction, so that a mutation in a strain of virulence $\alpha$ gives rise to a new strain of virulence $\alpha' \approx \alpha$.   We will assume that given a mutation occurs in an individual  $\alpha$, the offspring has virulence $\alpha'$ with probability $K(\alpha,\alpha')$; we will further assume that 
\[
	\int_{\alpha_{min}}^{\alpha_{max}} (\alpha - \alpha')^{k} K(\alpha,\alpha')\, d\alpha'
	 	= \begin{cases}
			0 & \text{if $k = 1$,}\\
			\varepsilon \nu(\alpha) & \text{if $k = 2$, and,}\\
			o(\varepsilon) & \text{otherwise,} 
		\end{cases}
\] 
where $\varepsilon$ is a dimensionless parameter which we will take to 0; this limit of small mutational effects allows us to ignore terms of order $\BigO{\abs{\alpha-\alpha'}^{2}}$ in the fixation probability.

We now pass from the individual based model to the trait substitution sequence \citep{Metz1992,Dieckmann1996}: we have seen that whenever a new strain arises, either the  mutant or resident strain will rapidly go extinct.  Until a new mutant arises, the population will be composed entirely of individuals of the surviving strain.  Let $A_{\epsilon}(t)$ be a random variable giving the virulence of the strain that survived the last competition event prior to time $t$.  The population is thus entirely composed of the strain $A_{\epsilon}(t)$ except for times $t$ in the short intervals when two strains are competing. If we pass to a ``mutational time scale'', $\frac{t}{\epsilon}$, as $\epsilon \to 0$ the duration of these intervals shrinks to 0, and we are left with a process in which novel mutations either fix or disappear instantly, so that the population is only observed with a single strain at equilibrium.  

Formally, as $\epsilon \to 0$, $A_{\epsilon}(\frac{t}{\epsilon}) \Rightarrow A(t)$, a continuous time Markov chain that jumps from virulence $\alpha$ to $\alpha'$ when a strain of virulence $\alpha'$ successfully invades a population with resident virulence $\alpha$.    The process $A(t)$ has generator (recall, the generator is the operator $\mathcal{L}$ defined by 
$\mathcal{L}f(\alpha) \defn \frac{d}{dt}\big\vert_{t=0} \mathbb{E}[f(A(t)) \vert A(0) = \alpha]$):
\[
	\mathcal{L}f(\alpha) = n \eta(\alpha) I_{\text{eq}}(\alpha)
		\int_{\alpha_{min}}^{\alpha_{max}} K(\alpha,\alpha') s(\alpha, \alpha')
			\left(f(\alpha')-f(\alpha)\right)\, d\alpha'
\]

Now, consider the time rescaled process $\hat{A}_{\varepsilon} \defn A\left(\frac{t}{\varepsilon}\right)$; 
this has generator
\[
	\hat{\mathcal{L}}_{\varepsilon}f(\alpha) = \frac{n}{\varepsilon} \eta(\alpha) I_{\text{eq}}(\alpha)
		\int_{\alpha_{min}}^{\alpha_{max}} K(\alpha,\alpha') s(\alpha, \alpha')
			\left(f(\alpha')-f(\alpha)\right)\, d\alpha'.
\]
Taylor expanding $s(\alpha, \alpha') \left(f(\alpha')-f(\alpha)\right)$ in $\alpha'$ about $\alpha$, this is equal to 
\begin{multline*}
	\frac{n}{\varepsilon} \eta(\alpha) I_{\text{eq}}(\alpha)
		\int_{\alpha_{min}}^{\alpha_{max}} K(\alpha,\alpha') \left[\left(
		\frac{\partial s}{\partial \alpha'}(\alpha, \alpha) f'(\alpha) + \frac{1}{2}
		s(\alpha, \alpha') f''(\alpha)\right)(\alpha - \alpha')^{2} + \cdots\right]\, d\alpha'\\
	= n \eta(\alpha) I_{\text{eq}}(\alpha) \nu(\alpha) \frac{\partial s}{\partial \alpha'}(\alpha, \alpha) f'(\alpha)
	 + \frac{n}{2} \eta(\alpha) I_{\text{eq}}(\alpha) \nu(\alpha) s(\alpha, \alpha) f''(\alpha) + o(\varepsilon).
\end{multline*}
Thus, as $\varepsilon \to 0$, $\hat{A}_{\varepsilon}(t)$ converges to a limiting diffusion $\hat{A}(t)$ with advective coefficient
\begin{equation}\label{DRIFT}
	\mu(\alpha) = n \eta(\alpha) I_{\text{eq}}(\alpha) \nu(\alpha) 
		\frac{\partial}{\partial \alpha'} s(\alpha, \alpha)
	= \frac{\eta(\alpha) \nu(\alpha)}{2} 
		\left(n I_{\text{eq}}(\alpha)\frac{R_{0}'(\alpha)}{R_{0}(\alpha)} 
		- \frac{\beta'(\alpha)}{\beta(\alpha)}\right)
\end{equation}
and diffusion coefficient 
\begin{equation}\label{DIFFUSION}
	\sigma^{2}(\alpha) = n \eta(\alpha) I_{\text{eq}}(\alpha) \nu(\alpha) s(\alpha, \alpha) 
		=  \eta(\alpha)\nu(\alpha).
\end{equation}
and generator
\[
	\hat{\mathcal{L}}f(\alpha) = \mu(\alpha)f'(\alpha) + \frac{1}{2} \sigma^{2}(\alpha) f''(\alpha).
\]
The process $\hat{A}(t)$ is our canonical diffusion.

\subsection{Stationary Distribution}

Using \eqref{DRIFT} and \eqref{DIFFUSION}, we may compute the stationary distribution 
$\psi$ of $\hat{A}(t)$; this stationary distribution describes the long-term behaviour of the virulence, after any ``memory'' of the initial state has been lost.    Given any subset $\mathcal{A} \subset (\alpha_{min},\alpha_{max})$, $\psi(\mathcal{A})$ gives the proportion of time that the virulence is in the set $\mathcal{A}$, or equivalently, the probability that at some random sampling time $t$, the virulence takes a value in  $\mathcal{A}$. $\psi$ is characterised by the relation 
\[
	\int_{\alpha_{min}}^{\alpha_{max}} \hat{\mathcal{L}}f(\alpha) \psi(d\alpha) = 0.
\]
In particular, if $\psi(d\alpha)$ has a density, which, in a slight abuse of notation, we write as $\psi(\alpha)$,
\[
	\hat{\mathcal{L}}^{*} \psi = 0,
\]
where $\hat{\mathcal{L}}^{*}$, defined by
\[
	\hat{\mathcal{L}}^{*} f(\alpha) = -\frac{d}{d\alpha} \left[\mu(\alpha)f(\alpha)\right]
	+ \frac{1}{2} \frac{d^{2}}{d\alpha^{2}}\left[\sigma^{2}(\alpha)f(\alpha)\right],
\]
is the adjoint operator to $\hat{\mathcal{L}}$ (and thus, 
\[
	\frac{d}{dt} f(\alpha,t) = \hat{\mathcal{L}}^{*} f(\alpha,t)
\]
is the Fokker-Planck equation for the probability density of $\hat{A}(t)$, $f(\alpha,t)$).

Thus,
\begin{equation}\label{STATIONARY}
	\psi(\alpha) = \frac{1}{Z}
	\frac{1}{\sigma^{2}(\alpha)} e^{\int \frac{2\mu(\alpha)}{\sigma^{2}(\alpha)}\, d\alpha}\, d\alpha
\end{equation}
where $Z = \int_{\alpha_{min}}^{\alpha_{max}}\frac{1}{\sigma^{2}(\alpha)} 
e^{\int \frac{2\mu(\alpha)}{\sigma^{2}(\alpha)}\, d\alpha}\, d\alpha$ is a normalising constant.

From the previous, we have
\begin{equation}\label{INTEGRAND}
%	\begin{aligned}
	\frac{2\mu(\alpha)}{\sigma^{2}(\alpha)}
%	&= n I_{\text{eq}}(\alpha) \left(\frac{\beta'(\alpha)}{\beta(\alpha)}-\frac{1}{\delta+\alpha+\gamma}\right) - \frac{\beta'(\alpha)}{\beta(\alpha)}\\
	= n I_{\text{eq}}(\alpha)\frac{R_{0}'(\alpha)}{R_{0}(\alpha)} - \frac{\beta'(\alpha)}{\beta(\alpha)}.
%	\end{aligned}
\end{equation}

Unfortunately, we can only compute its integral analytically in the case of a ``flat landscape'', when $R_{0}(\alpha) \equiv R_{0}$, independently of $\alpha$.  In this case, we have $\beta(\alpha) = R_{0}(\delta+\alpha+\gamma)$, $\beta'(\alpha) = R_{0}$, and
\[
	\frac{2\mu(\alpha)}{\sigma^{2}(\alpha)}
	= - \frac{\beta'(\alpha)}{\beta(\alpha)} 
	= - \frac{1}{\delta+\alpha+\gamma}
\]
which has integral $-\ln{(\delta+\alpha+\gamma)}$, so that 
\[
	\frac{1}{\sigma^{2}(\alpha)} e^{\int \frac{2\mu(\alpha)}{\sigma^{2}(\alpha)}\, d\alpha}\, d\alpha = \frac{1}{\sigma^{2}(\alpha)} \frac{1}{\delta+\alpha+\gamma}.
\]
In particular, in the case when $\eta(\alpha)$ and $\nu(\alpha)$ (and thus $\sigma^{2}(\alpha)$) are constants independent of $\alpha$, then we can integrate this to obtain $Z$ and thus a closed expression for the stationary distribution:
\begin{equation}\label{FLATLANDSCAPE}
	\psi(\alpha) =  
	 	\frac{1}{\ln{\left(\frac{\delta+\alpha_{max}+\gamma}{\delta+\alpha_{min}+\gamma}\right)}}
		\frac{1}{(\delta+\alpha+\gamma)}.
\end{equation}
	
In the next section we will show how one may obtain an analytical approximation in the large $n$ limit.   

\subsection{Asymptotic Approximation to Stationary Distribution}

The stationary distribution \eqref{STATIONARY} lends itself to an approximation by Laplace's method (see \eg \citet{Erdelyi:1956yu}), which tells us that if $\phi(x)$ is a twice differentiable function with a unique local maximum attained at $x_{0} \in (a,b)$, and $g(x)$ is continuous, then 
\[
	\int_{a}^{b} g(x) e^{n \phi(x)}\, dx = \sqrt{-\frac{2\pi}{n \phi''(x_{0})}} g(x_{0}) e^{n \phi(x_{0})}
	\left(1+ \BigO{\frac{1}{n}}\right).
\]
(To order $\frac{1}{n}$, one has
\begin{multline}\label{HIGHERORDER}
	\int_{a}^{b} g(x) e^{n \phi(x)}\, dx = \sqrt{-\frac{2\pi}{n \phi''(x_{0})}}  e^{n \phi(x_{0})}
	\Bigg(g(x_{0})\\
		+ \frac{1}{n}\left(\frac{1}{2} \frac{g''(x_{0})}{\phi''(x_{0})} 
		+ \frac{1}{8}\frac{g^{(4)}(x_{0})}{(\phi''(x_{0}))^{2}} 	
		+ \frac{1}{2} \frac{g'(x_{0})\phi^{(3)}(x_{0})}{(\phi''(x_{0}))^{2}}
		+ \frac{5}{24} \frac{g(x_{0})(\phi^{(3)}(x_{0}))^{2}}{(\phi''(x_{0}))^{3}}\right)
	+ \BigO{\frac{1}{n^{2}}}\Bigg);
\end{multline}
see \eg \citet{Bender+Orszag78}).

We can apply this to our stationary distribution by first observing from \eqref{INTEGRAND} that
\[
	 \frac{1}{\sigma^{2}(\alpha)} e^{\int \frac{2\mu(\alpha)}{\sigma^{2}(\alpha)}\, d\alpha}
	= g(\alpha) e^{n \phi(\alpha)}
	%n \frac{\lambda(R_{0}(\alpha)-1)}{\beta(\alpha) - \alpha}
		%\left(\frac{\beta'(\alpha)}{\beta(\alpha)}-\frac{1}{\delta+\alpha+\gamma}\right) 
		%-  \frac{\beta'(\alpha)}{\beta(\alpha)}
\]
%A(x) = 2n \left(\frac{\delta+x}{\beta(x)} - \ln{\left(\frac{\delta+x}{\beta(x)}\right)}\right) -  2\ln{\beta(x)}, 
% \frac{d}{d\alpha} 
for
\begin{equation}\label{PHI}
	\phi(\alpha) = \int I_{\text{eq}}(\alpha) \frac{R_{0}'(\alpha)}{R_{0}(\alpha)}\, d\alpha
\end{equation}
and 
\[
	g(\alpha) = \frac{1}{\sigma^{2}(\alpha)\beta(\alpha)},
\]
and we can thus apply Laplace's formula to compute $Z$. 

To find our $\alpha_{0}$, we note that by assumption, we have chosen $\alpha_{min}$  and $\alpha_{max}$ so that $\frac{\lambda(R_{0}(\alpha)-1)}{\beta(\alpha) - \alpha} > 0$ for all $\alpha \in [\alpha_{min},\alpha_{max}]$.  Thus, all values of  $\alpha_{0}$ such that $\phi'(\alpha_{0}) = 0$ satisfy $R_{0}'(\alpha) = 0$.

On the other hand, we recall that $R_{0}(\alpha) = \frac{\beta(\alpha)}{\delta+\alpha+\gamma}$, so that
\begin{equation}\label{RDER}
	R_{0}'(\alpha) = \frac{\beta'(\alpha)}{\delta+\alpha+\gamma}
		- \frac{\beta(\alpha)}{(\delta+\alpha+\gamma)^{2}},
\end{equation}
and thus also $R_{0}(\alpha_{0}) = \beta'(\alpha_{0})$.

We next observe that 
\begin{equation}\label{DDR0}
	R''(\alpha_{0}) = \frac{\beta'(\alpha_{0})\beta''(\alpha_{0})}{\beta(\alpha_{0})}
	= \frac{\beta''(\alpha_{0})}{\delta+\alpha_{0}+\gamma}.
\end{equation}
which depends on the choice of tradeoff function $\beta(\alpha)$.  We note briefly that if we assume $\beta(\alpha)$ is increasing, then this is a local maximum if and only if $\beta''(\alpha_{0}) < 0$.   In particular, if we assume that $R_{0}(\alpha)$ has a unique global maximum, then it must occur at $\alpha_{0}$.   We will henceforth make this assumption (\nb we don't have to assume this in general, but to apply Laplace's method, we require that $\alpha_{0}$ be a global maximum).

We then have 
\[
	\phi''(\alpha) =
	I_{\text{eq}}'(\alpha)
		 \frac{R_{0}'(\alpha)}{R_{0}(\alpha)} + I_{\text{eq}}(\alpha)
		\left( \frac{R_{0}''(\alpha)}{R_{0}(\alpha)}-\left( \frac{R_{0}'(\alpha)}{R_{0}(\alpha)}\right)^{2}\right)
\]
so that 
\[
	\phi''(\alpha_{0}) 
%		= n \frac{\lambda(R_{0}(\alpha_{0})-1)}{\beta(\alpha_{0}) - \alpha_{0}}
%		 \frac{R_{0}'(\alpha)}{R_{0}(\alpha)}
		= I_{\text{eq}}(\alpha_{0}) \frac{R_{0}''(\alpha_{0})}{R_{0}(\alpha_{0})}.
\]
Thus,
\[
	Z = \int_{\alpha_{min}}^{\alpha_{max}}\frac{1}{\sigma^{2}(\alpha)} 
		e^{\int \frac{2\mu(\alpha)}{\sigma^{2}(\alpha)}\, d\alpha}\, d\alpha
	= \sqrt{\frac{2\pi}{n I_{\text{eq}}(\alpha_{0}) \frac{\abs{R_{0}''(\alpha_{0})}}{R_{0}(\alpha_{0})}}} 
		e^{n \phi(\alpha_{0})}
		\left(\frac{1}{\sigma^{2}(\alpha_{0})\beta(\alpha_{0})} + \BigO{\frac{1}{n}}\right).
\]

We then have
\[
	\frac{1}{Z} \frac{1}{\sigma^{2}(\alpha)} e^{\int \frac{2\mu(\alpha)}{\sigma^{2}(\alpha)}\, d\alpha}
	= \frac{\sigma^{2}(\alpha_{0})\beta(\alpha_{0})}{\sigma^{2}(\alpha)\beta(\alpha)} \frac{1}{\sqrt{\frac{2\pi}{n I_{\text{eq}}(\alpha_{0})\frac{\abs{R_{0}''(\alpha_{0})}}{R_{0}(\alpha_{0})}}}}
		e^{n(\phi(\alpha)-\phi(\alpha_{0}))}\left(1 + \BigO{\frac{1}{n}}\right).
\]
Next, recalling that $\phi'(\alpha_{0}) = 0$, a Taylor expansion gives
\[
	n(\phi(\alpha)-\phi(\alpha_{0})) 
		= \frac{1}{2} n \left(\phi''(\alpha_{0})(\alpha-\alpha_{0})^{2} 
		+ \BigO{(\alpha-\alpha_{0})^{3}}\right).
\]
Now, for $\alpha$ close to $\alpha_{0}$, $(\alpha-\alpha_{0})^{3}$ will be quite small, so we can locally approximate our full stationary distribution by a process that is almost a Gaussian with mean $\alpha_{0}$ and variance 
\[
	\frac{1}{n I_{\text{eq}}(\alpha_{0}) \frac{\abs{R_{0}''(\alpha_{0})}}{R_{0}(\alpha_{0})}},
\]
except for a pre-factor of $\frac{\sigma^{2}(\alpha_{0})\beta(\alpha_{0})}{\sigma^{2}(\alpha)\beta(\alpha)}$, which skews the distribution:
\begin{multline*}
	\frac{1}{Z} \frac{1}{\sigma^{2}(\alpha)} e^{\int \frac{2\mu(\alpha)}{\sigma^{2}(\alpha)}\, d\alpha}
	 \approx 
	 \frac{\sigma^{2}(\alpha_{0})\beta(\alpha_{0})}{\sigma^{2}(\alpha)\beta(\alpha)} 
	 \frac{1}{\sqrt{\frac{2\pi}{n I_{\text{eq}}(\alpha_{0}) \frac{\abs{R_{0}''(\alpha_{0})}}{R_{0}(\alpha_{0})}}}}
	e^{-\frac{1}{2} n I_{\text{eq}}(\alpha_{0})\frac{\abs{R_{0}''(\alpha_{0})}}{R_{0}(\alpha_{0})}(\alpha-\alpha_{0})^{2}}\\
	= \frac{\eta(\alpha_{0})}{\eta(\alpha)} \frac{\nu(\alpha_{0})}{\nu(\alpha)} 	 \frac{\beta(\alpha_{0})}{\beta(\alpha)}\
	  \frac{1}{\sqrt{\frac{2\pi}{n I_{\text{eq}}(\alpha_{0}) \frac{\abs{R_{0}''(\alpha_{0})}}{R_{0}(\alpha_{0})}}}}
	e^{-\frac{1}{2} n I_{\text{eq}}(\alpha_{0})\frac{\abs{R_{0}''(\alpha_{0})}}{R_{0}(\alpha_{0})}(\alpha-\alpha_{0})^{2}}.
\end{multline*}
Again assuming that $\eta(\alpha)$ and $\nu(\alpha)$ are constants independent of $\alpha$, this simplifies to 
\begin{equation}\label{LAPLACEAPPROX}
	\psi_{\text{approx}}(\alpha) = \frac{\beta(\alpha_{0})}{\beta(\alpha)}\
	  \frac{1}{\sqrt{\frac{2\pi}{n I_{\text{eq}}(\alpha_{0}) \frac{\abs{R_{0}''(\alpha_{0})}}{R_{0}(\alpha_{0})}}}}
	e^{-\frac{1}{2} n I_{\text{eq}}(\alpha_{0})\frac{\abs{R_{0}''(\alpha_{0})}}{R_{0}(\alpha_{0})}(\alpha-		\alpha_{0})^{2}}.
\end{equation}

\begin{rem}
Applying \eqref{HIGHERORDER} with $f(\alpha) g(\alpha)$ in place of $g(\alpha)$ for an arbitrary differentiable function $f(\alpha)$, we may estimate the error in integrating $f(\alpha)$ versus the true and approximate densities for the stationary distribution:
\begin{multline*}
	\abs{\frac{1}{Z}  \int_{\alpha_{min}}^{\alpha_{max}} f(\alpha) g(\alpha) 
		e^{-n \int \phi(\alpha)\, d\alpha}\, d\alpha
	-  \frac{1}{\sqrt{\frac{2\pi}{n |\phi''(\alpha_{0})|}}g(\alpha_{0})}
	\int_{\alpha_{min}}^{\alpha_{max}} f(\alpha) g(\alpha) 
		e^{-n\frac{\phi''(\alpha_{0})}{2}(\alpha-\alpha_{0})^{2}}\, d\alpha}\\
	= \frac{1}{n} \abs{\frac{1}{2} \frac{f'(\alpha_{0}) \phi'''(\alpha_{0})}{(\phi''(\alpha_{0}))^{2}}
	- \frac{1}{2} \frac{f(\alpha_{0}) g''(\alpha_{0})}{g(\alpha_{0}) |\phi''(\alpha_{0})|}
	}%- \frac{1}{8} \frac{f'(\alpha_{0}) g(\alpha_{0}) \phi^{(4)}(\alpha_{0})}{(\phi'''(\alpha_{0}))^{2}}}
	+ \BigO{\frac{1}{n^{2}}}
\end{multline*}
From this, see that in the bounded Lipschitz metric on probability measures, the difference between the true stationary approximation $\psi(d\alpha)$ and the Laplace approximation $\psi_{\text{approx}}(d\alpha)$ satisfies
\[
	\lim_{n \to \infty} d_{\text{BL}}(\psi,\psi_{\text{approx}})
	\leq  \frac{1}{2} \abs{\frac{\phi'''(\alpha_{0})}{(\phi''(\alpha_{0}))^{2}}} + 
	\frac{1}{2} \abs{\frac{g''(\alpha_{0})}{g(\alpha_{0}) |\phi''(\alpha_{0})|}}.
	%+ \BigO{\frac{1}{n^{2}}}
\]
Unfortunately, we cannot similarly bound the total variation distance: for any $M > 0$ the function $f(\alpha) := e^{-M (\alpha-\alpha_{min})^{2}}$ satisfies 
\[
	\sup_{\alpha \in [\alpha_{min},\alpha_{max}]} \abs{f(\alpha)} \leq 1,
\]
but, since 
\[
	\sup_{\alpha \in [\alpha_{min},\alpha_{max}]} \abs{f(\alpha)} = \sqrt{2M} e^{-\frac{1}{2}},
\]
the bound 
\[
	\frac{1}{n} \abs{\frac{1}{2} \frac{f'(\alpha_{0}) \phi'''(\alpha_{0})}{(\phi''(\alpha_{0}))^{2}}
	- \frac{1}{2} \frac{f(\alpha_{0}) g''(\alpha_{0})}{g(\alpha_{0}) |\phi''(\alpha_{0})|}}
\] 
may be made arbitrarily large. 
\end{rem}

\subsection{Mean \& Mode of the Stationary Distribution}

We observe that, whilst the stationary distribution is closely related to a Gaussian centred at 
$\alpha_{0}$, the value of the virulence that maximises $R_{0}(\alpha)$, the full stationary distribution does not have mean $\alpha_{0}$, and $\alpha_{0}$ is not the most probable value of the virulence.

\subsubsection{Estimating the Mean}

Applying Laplace's method with $g(\alpha)$ replaced by $\alpha g(\alpha)$ allows us to estimate the mean of the stationary distribution; to lowest order, we have 
\[
	\int_{\alpha_{min}}^{\alpha^{max}} \alpha g(\alpha) e^{n \phi(\alpha)}\, d\alpha 
	= \sqrt{-\frac{2\pi}{n \phi''(\alpha_{0})}} \alpha_{0} g(\alpha_{0}) e^{n \phi(\alpha_{0})}
	\left(1+ \BigO{\frac{1}{n}}\right),
\]
so that, normalising by our prior estimate of $Z$, we find that the mean is $\alpha_{0}$ to order $\BigO{\frac{1}{n}}$.  To observe the effects of a finite population size, we can the use higher order corrections to Laplace's method to obtain the $\BigO{\frac{1}{n}}$ terms in both the integral above and in $Z$; we omit the calculations, but remark that the mean can then be shown to be
\begin{multline*}
	\alpha_{0} - \frac{1}{n} \left(\frac{g'(\alpha_{0})}{g(\alpha_{0})\phi''(\alpha_{0})}
		- \frac{\phi'''(\alpha_{0})}{2 (\phi''(\alpha_{0}))^{2}}\right) 
		+ o{\textstyle\left(\frac{1}{n}\right)}\\
		= \alpha_{0} + \frac{1}{n I_{\text{eq}}(\alpha_{0}) \frac{\abs{R_{0}''(\alpha_{0})}}{R_{0}(\alpha_{0})}} 
		\left(\frac{g'(\alpha_{0})}{g(\alpha_{0})}
		- \frac{\phi'''(\alpha_{0})}{2 \phi''(\alpha_{0})}\right) 
		+ o{\textstyle\left(\frac{1}{n}\right)}
\end{multline*}
where
\begin{align*}
	\frac{g'(\alpha_{0})}{g(\alpha_{0})} 
%	&= \frac{d}{d\alpha}\bigg\vert_{\alpha = \alpha_{0}} \ln{g(\alpha)}\\
	&=  -\frac{\eta'(\alpha_{0})}{\eta(\alpha_{0})} - \frac{\nu'(\alpha_{0})}{\nu(\alpha_{0})} 
	- \frac{\beta'(\alpha_{0})}{\beta(\alpha_{0})} \\
	&=  -\frac{\eta'(\alpha_{0})}{\eta(\alpha_{0})} - \frac{\nu'(\alpha_{0})}{\nu(\alpha_{0})} 
	- \frac{1}{\delta+\alpha_{0}+\gamma},
	%+ \frac{\lambda(R_{0}(\alpha_{0})-1)}{\beta(\alpha_{0}) - \alpha_{0}}
	%\frac{\nu'(\alpha_{0})}{\nu(\alpha_{0})} - \frac{1}{\nu(\alpha_{0})}\\
	%&= 2\frac{\beta'(\alpha_{0})}{\beta(\alpha_{0})}  + \frac{1}{n} I_{\text{eq}}(\alpha_{0}) 
%		\frac{\nu'(\alpha_{0})}{\nu(\alpha_{0})} - \frac{1}{\nu(\alpha_{0})}
\end{align*}
which simplifies to $- \frac{1}{\delta+\alpha_{0}+\gamma}$ in the case when $\eta$ and $\nu$ are independent of $\alpha$,
and
\[
	\frac{\phi'''(\alpha_{0})}{\phi''(\alpha_{0})} 
	=  \frac{R_{0}'''(\alpha_{0})}{R_{0}''(\alpha_{0})} 
	- \frac{\beta'(\alpha_{0})-1}{\beta(\alpha_{0}) - \alpha_{0}} = \frac{I_{\text{eq}}'(\alpha_{0})}{I_{\text{eq}}(\alpha_{0})} - \frac{R_{0}'''(\alpha_{0})}{\abs{R_{0}''(\alpha_{0})}}.
\]
We note that this order $\BigO{\frac{1}{n}}$ term is proportional to the variance of the best-fit Gaussian of the previous section.

One may similarly show that the variance and skewness of the stationary distribution are
\[
	\frac{1}{n} \frac{1}{\phi''(\alpha_{0})} + o{\textstyle\left(\frac{1}{n}\right)}
	= \frac{1}{n I_{\text{eq}}(\alpha_{0}) \frac{\abs{R_{0}''(\alpha_{0})}}{R_{0}(\alpha_{0})}} 
		+ o{\textstyle\left(\frac{1}{n}\right)}
\]
and
\[
	 \frac{1}{n} \frac{3 g'(\alpha_{0})}{g(\alpha_{0})(\phi''(\alpha_{0}))^{2}} 
	 + o{\textstyle\left(\frac{1}{n}\right)}
\]
respectively.  We note that if $\eta'(\alpha_{0}) \geq 0$ and $\nu'(\alpha_{0}) \geq 0$ (for example, if both rates are independent of $\alpha$), then the stationary distribution has negative skew.

\subsubsection{Estimating the Mode}

We begin by observing the density function of the stationary distribution \eqref{STATIONARY} may be written as
\[
	 e^{\int \frac{2\mu(\alpha)}{\sigma^{2}(\alpha)}\, d\alpha - 
		\ln{\frac{\sigma^{2}(\alpha)}{Z}}} \, d\alpha
\]
and thus has its maximum where 
\[
	-\frac{2\mu(\alpha)}{\sigma^{2}(\alpha)} 
		= \frac{d}{d\alpha} \ln{\frac{\sigma^{2}(\alpha)}{Z}},
\]
or equivalently,  for the value of $\alpha^{\star}$ such that 
\[
	n\phi'(\alpha^{\star}) =  -\frac{g'(\alpha^{\star})}{g(\alpha^{\star})} 
\]
where $\phi(\alpha)$ is given by \eqref{PHI}.   

When $\sigma^{2}(\alpha)$ is constant (\ie the variance of the mutation kernel is independent of the resident variance) then this reduces to 
\[
	\frac{2\mu(\alpha^{\star})}{\sigma^{2}} = 0,
\]
and thus $\mu(\alpha^{\star}) = 0$.  Recalling \eqref{DRIFT}, this tells us that 
\[
	\frac{1}{\delta+\alpha^{\star}+\gamma} - \left(1-\frac{1}{n I_{\text{eq}}(\alpha^{\star})} \right)
	\frac{\beta'(\alpha^{\star})}{\beta(\alpha^{\star})} = 0,
\]
and thus
\begin{equation}\label{MVT}
	\beta'(\alpha^{\star}) = R_{0}(\alpha^{\star})\left(1+\frac{1}{n I_{\text{eq}}(\alpha^{\star})-1}\right).
\end{equation}
On the other hand, \eqref{RDER} tells us that
\[
		\frac{R_{0}'(\alpha)}{R_{0}(\alpha)} = \frac{\beta'(\alpha)}{\beta(\alpha)}
			- \frac{1}{\delta+\alpha+\gamma},
\]
so that 
\[
	\frac{R_{0}'(\alpha^{\star})}{R_{0}(\alpha^{\star})} 
		= \frac{1}{n I_{\text{eq}}(\alpha^{\star})} \frac{\beta'(\alpha^{\star})}{\beta(\alpha^{\star})}
		> 0,
\]
since $\beta(\alpha)$ is increasing.  In particular, since $R_{0}'(\alpha)$ is maximized at $\alpha_{0}$, we see immediately that $\alpha^{\star} < \alpha_{0}$. 

Even in this special case, we cannot solve for $\alpha^{\star}$ exactly.   Instead we will seek a perturbative solution to 
\begin{equation}\label{PHIPERTURB}
	\phi'(\alpha^{\star}) = -\frac{1}{n} \frac{g'(\alpha^{\star})}{g(\alpha^{\star})} 
\end{equation}
in the general case. We already know that $\phi'(\alpha_{0}) = 0$; we thus seek a solution of the form
\[
	\alpha^{\star} = \alpha_{0} + \sum_{i=1}^{\infty} \frac{\alpha_{i}}{n^{i}}.
\]
Substituting this into \eqref{PHIPERTURB} and Taylor expanding right and left, we find that
\[
	\phi'(\alpha_{0}) + \frac{1}{n} \phi''(\alpha_{0})  \alpha_{1}
	= -\frac{1}{n} \frac{g'(\alpha_{0})}{g(\alpha_{0})} 
		+ o{\textstyle\left(\frac{1}{n}\right)}
\]
\ie that 
\[%begin{multline*}
	\alpha^{\star} =  \alpha_{0} 
		- \frac{1}{n} \frac{g'(\alpha_{0})}{g(\alpha_{0})\phi''(\alpha_{0})}	
		+ o{\textstyle\left(\frac{1}{n}\right)}
	= \alpha_{0} + \frac{1}{n I_{\text{eq}}(\alpha_{0}) \frac{\abs{R_{0}''(\alpha_{0})}}{R_{0}(\alpha_{0})}} 
		\frac{g'(\alpha_{0})}{g(\alpha_{0})}
		+ o{\textstyle\left(\frac{1}{n}\right)}
\]%end{multline*}
which may be expanded using the expression for $\frac{g'(\alpha_{0})}{g(\alpha_{0})}$ given in the previous sections.  In particular, when when $\eta(\alpha)$ and $\nu(\alpha)$ (and thus $\sigma^{2}(\alpha)$) are constants independent of $\alpha$, we have that 
\begin{equation}\label{MODE}
	\alpha^{\star} =  \alpha_{0} 
		- \frac{1}{n I_{\text{eq}}(\alpha_{0}) \frac{\abs{R_{0}''(\alpha_{0})}}{R_{0}(\alpha_{0})}} 
		\frac{1}{\delta+\alpha_{0}+\gamma} + o{\textstyle\left(\frac{1}{n}\right)},
\end{equation}
so that to first order, the modal virulence is the  virulence maximizing $R_{0}(\alpha)$ less the product of the variance of the best-fit Gaussian and the expected infectious period when the virulence is $\alpha_{0}$. 

\section{Some Rigorous Demonstrations}

\subsection{Proof of Proposition \ref{prop:Pi_i}}\label{BOUNDING}

Substituting $f$ with $\Pi_i$ in It\^o's formula \eqref{ITOP} for jump processes yields
\begin{multline}
\label{ITOP2}
	P^{(n)}_{i}(t) = P^{(n)}_{i}(0)+\int_0^t \sum_{j = 1}^{d} 
	\frac{\partial \Pi_{i}}{\partial x_{j}}(\bar{\bm{I}}^{(n)}(s)) F^{(n)}_{j}(\bar{\bm{E}}^{(n)}(s))
	+ \frac{1}{2} \sum_{j,k = 1}^{d} a^{(n)}_{jk}(\bar{\bm{E}}^{(n)}(s))
		\frac{\partial \Pi_{i}}{\partial x_{j}\partial x_{k}} (\bar{\bm{I}}^{(n)}(s))
	 \, ds\\
	+ \frac{1}{n} \int_{0}^{t}  \sum_{j = 1}^{d} \frac{\partial \Pi_{i}}{\partial x_{j}}(\bar{\bm{I}}^{(n)}(s))\, 
		dM^{(n)}_{j}(s) + \varepsilon^{(n)}_{i}(t),
\end{multline}
where $\varepsilon^{(n)}_{i}(t)$ can be expressed thanks to Equation \eqref{ERROR} as
\begin{multline*}
	\varepsilon^{(n)}_{i}(t) =  \sum_{s < t} \Pi_{i}(\bar{\bm{I}}^{(n)}(s)) - \Pi_{i}(\bar{\bm{I}}^{(n)}(s-)) 
	- \sum_{j = 1}^{d} \frac{\partial \Pi_{i}}{\partial x_{j}} (\bar{\bm{I}}^{(n)}(s-)) \Delta \bar{I}^{(n)}_{j}(s) \\
	- \frac{1}{2} \sum_{j,k = 1}^{d}  \frac{\partial \Pi_{i}}{\partial x_{j}\partial x_{k}}(\bar{\bm{I}}^{(n)}(s-))
	\Delta \bar{I}^{(n)}_{j}(s) \Delta \bar{I}^{(n)}_{k}(s).
\end{multline*}

%\begin{rem}
%Here, the integrals with respect to the $M^{(n)}_{j}(t)$ are Riemann-Stieltjes integrals, where, as with the It\^o integral with respect to Brownian motion, the integrand is always evaluated at the left endpoint of each interval in the partition, and the limit over approximating sums is in probability.
%\end{rem}

Now some elementary computations yield
\[
	\frac{\partial \Pi_{i}}{\partial x_{j}} 
	= \frac{1}{\sum_{l=1}^{d} x_{l}}\left(\mathbbm{1}_{\{i=j\}} - \frac{x_{i}}{\sum_{l=1}^{d} x_{l}}\right)
	\quad \text{and} \quad
	\frac{\partial \Pi_{i}}{\partial x_{j}\partial x_{k}}
	 = -\frac{1}{\left(\sum_{l=1}^{d} x_{l}\right)^{2}}\left(\mathbbm{1}_{\{i=j\}} 
	 	+ \mathbbm{1}_{\{i=k\}} - 2\frac{x_{i}}{\sum_{l=1}^{d} x_{l}}\right),
\]
where $\mathbbm{1}_{\{i=j\}}$ is equal to 1 if $i=j$ and 0 otherwise.
Substituting these and \eqref{AA} into \eqref{ITOP2} yields after some simplification Equation \eqref{PIF} in Proposition \ref{prop:Pi_i}.

Now it remains to prove that $n^2\varepsilon^{(n)}_{i}(t)$ is uniformly bounded with high probability.  
Taylor's theorem tells us that
\begin{multline*}
 	\Pi_{i}(\bar{\bm{I}}^{(n)}(s)) - \Pi_{i}(\bar{\bm{I}}^{(n)}(s-)) - \sum_{j = 1}^{d}  
	\frac{\partial \Pi_{i}}{\partial x_{j}}(\bar{\bm{I}}^{(n)}(s-)) \Delta \bar{I}^{(n)}_{j}(s) 
	- \frac{1}{2} \sum_{j,k = 1}^{d}  
	\frac{\partial \Pi_{i}}{\partial x_{j} \partial x_{k}}(\bar{\bm{I}}^{(n)}(s-))
		\Delta \bar{I}^{(n)}_{j}(s) \Delta \bar{I}^{(n)}_{k}(s)\\
	= \sum_{j,k=1}^{d} g_{ij}(\bar{\bm{I}}^{(n)}(s),\bar{\bm{I}}^{(n)}(s-))
		\Delta \bar{I}^{(n)}_{j}(s) \Delta \bar{I}^{(n)}_{k}(s),
\end{multline*}
where the functions $g_{ij}(\bm{x},\bm{y})$ satisfy 
\[
	\lim_{\bm{x} \to \bm{y}} g_{ij}(\bm{x},\bm{y}) = 0 
\]
uniformly on compact sets. 

Now, recalling 
\[
	\bar{I}^{(n)}_{i}(t) = \begin{multlined}[t] \bar{I}^{(n)}_{i}(0) + \frac{1}{n} P_{-\bm{e}_{0} + \bm{e}_{i}}\left(n \int_{0}^{t} \frac{\beta^{(n)}_{i} \bar{S}^{(n)}(s) \bar{I}^{(n)}_{i}(s) }{\bar{N}^{(n)}(s)}\, ds\right) \\
		- \frac{1}{n} P_{-\bm{e}_{i} - \bm{e}_{d+1}}\left(n \int_{0}^{t}  (\delta^{(n)} + \alpha^{(n)}_{i}) \bar{I}^{(n)}_{i}(s)\, ds\right)
			 - \frac{1}{n} P_{-\bm{e}_{i}}\left(n \int_{0}^{t}  \gamma^{(n)}_{i} \bar{I}^{(n)}_{i}(s)\, ds\right), \end{multlined}
\]
we see that $\Delta \bar{I}^{(n)}_{i}(s)$ is non-zero only at the jump-times of the Poisson processes and are always of magnitude $\frac{1}{n}$.  In particular, since $f_{i}(\bm{x})$ is smooth outside of a neighbourhood of $\bm{0}$, we can conclude that $g_{ij}$ is bounded above by a constant multiple of $\norm{\bm{x}-\bm{y}}$; this allows us to conclude that $|g_{ij}(\bar{\bm{I}}^{(n)}(s),\bar{\bm{I}}^{(n)}(s-))| \leq \frac{C}{n}$ and 
\[
	|\varepsilon^{(n)}_{i}(t)| 
	\leq \frac{C}{n} \sum_{s < t} \sum_{j,k=1}^{d} |\Delta \bar{I}^{(n)}_{j}(s)||\Delta \bar{I}^{(n)}_{k}(s)|.
\]
Further, $|\Delta \bar{I}^{(n)}_{j}(s)| |\Delta \bar{I}^{(n)}_{k}(s)|$ is non-zero if some pair of processes $P_{j,\cdot}$ and $P_{k,\cdot}$ jump simultaneously; if $j \neq k$, the Poisson processes are independent, and probability of such an event in an interval $[t,t+\Delta t)$ is $\BigO{\Delta t^2}$, and thus tends to 0 if as $\Delta t \to 0$ \ie  $|\Delta \bar{I}^{(n)}_{j}(s)| |\Delta \bar{I}^{(n)}_{k}(s)| \neq 0$ if and only if $j \neq k$.   Moreover, the processes $P_{S,j}$, $P_{j,-}$ and $P_{-\bm{e}_{i}}$ are also independent, and thus cannot jump simultaneously, so that $|\Delta \bar{I}^{(n)}_{j}(s)| |\Delta \bar{I}^{(n)}_{k}(s)| \neq 0$ (and is thus equal to $\frac{1}{n^2}$) at exactly the jump times of these Poisson processes; \ie
\[
	\sum_{s < t} \sum_{j,k=1}^{d} |\Delta \bar{I}^{(n)}_{j}(s)| |\Delta \bar{I}^{(n)}_{k}(s)| 
	= \begin{multlined}[t]  \frac{1}{n^{2}} P_{-\bm{e}_{0} + \bm{e}_{i}}\left(n \int_{0}^{t} \frac{\beta^{(n)}_{i} \bar{S}^{(n)}(s) \bar{I}^{(n)}_{i}(s) }{\bar{N}^{(n)}(s)}\, ds\right)\\
	+ \frac{1}{n^{2}} P_{-\bm{e}_{i} - \bm{e}_{d+1}}\left(n \int_{0}^{t}  (\delta^{(n)} + \alpha^{(n)}_{i}) \bar{I}^{(n)}_{i}(s)\, ds\right)\\
	+ \frac{1}{n^{2}} P_{-\bm{e}_{i}}\left(n \int_{0}^{t}  \gamma^{(n)}_{i} \bar{I}^{(n)}_{i}(s)\, ds\right)
		\end{multlined}
\]
We seek an upper bound on this quantity.  To that end, we begin by observing that $\bar{S}^{(n)}(t)$ and each $\bar{I}^{(n)}_{i}(t)$ is bounded above by $\bar{N}^{(n)}(t)$, and that
\[
	\bar{N}^{(n)}(t) \leq \bar{N}^{(n)}(0) + \frac{1}{n} P_{\bm{e}_{0} + \bm{e}_{d+1}}(n \lambda^{(n)} t),
\]
and, since this Poisson process is increasing in $t$, we have that for $t \leq T$,  
\[
	\bar{N}^{(n)}(t) \leq \bar{N}^{(n)}(0) + \frac{1}{n} P_{\bm{e}_{0} + \bm{e}_{d+1}}(n \lambda^{(n)} T).
\]
Now, 
\[
	\mathbb{E}\left[\frac{1}{n} P_{\bm{e}_{0} + \bm{e}_{d+1}}(n \lambda^{(n)} T)\right] = \lambda^{(n)} T,
\]
and, applying Chebyshev's inequality, we see that for any $C > 0$, 
\begin{multline*}
	\mathbb{P}\left\{\abs{\frac{1}{n} P_{\bm{e}_{0} + \bm{e}_{d+1}}(n \lambda^{(n)} T) - \lambda^{(n)} T} > C\right\}
	= \mathbb{P}\left\{|P_{\bm{e}_{0} + \bm{e}_{d+1}}(n \lambda^{(n)} T) - n\lambda^{(n)} T| > C n\right\}\\
	\leq \frac{\mathbb{E}\left[\left(P_{\bm{e}_{0} + \bm{e}_{d+1}}(n \lambda^{(n)} T) - n\lambda^{(n)} T \right)^{2}\right]}{C^{2}n^{2}} =  \frac{\lambda^{(n)} T}{C^{2} n} \to 0
\end{multline*}
as $n \to \infty$. Thus, for any fixed $T > 0$,  $\bar{N}^{(n)}(t)$ is bounded above and below on $[0,T]$ by \eg $\bar{N}^{(n)}(0) + \lambda^{(n)} T \pm 1 $, with probability that approaches one as $n$ tends to infinity.  

Thus, for example, 
\[
	\int_{0}^{t} \frac{\beta^{(n)}_{i} \bar{S}^{(n)}(s) \bar{I}^{(n)}_{i}(s) }{\bar{N}^{(n)}(s)}\, ds \leq \frac{\beta^{(n)}_{i} (\bar{N}^{(n)}(0) + \lambda^{(n)} T + 1)^{2} T}{\bar{N}^{(n)}(0) + \lambda^{(n)} T - 1},
\]
and we may proceed exactly as above to conclude that fo $t \leq T$,
\[	
	\frac{1}{n} P_{-\bm{e}_{0} + \bm{e}_{i}}\left(n \int_{0}^{t} \frac{\beta^{(n)}_{i} \bar{S}^{(n)}(s) \bar{I}^{(n)}_{i}(s) }{\bar{N}^{(n)}(s)}\, ds\right)
\]
is bounded above with probability approaching 1 as $n \to \infty$, and similarly for the other Poisson processes, from which we conclude that there exists some constant $C'$ such that 
\[
	|\varepsilon^{(n)}_{i}(t)| \leq \frac{C'}{n^{2}}
\]
with high probability.

\subsection{Proof of Proposition \ref{SSPF}}\label{PROOF}

In this section, we will make the heuristic argument of Section \ref{STRONG} rigorous using the technique of coupling (see \citet{Ball1995} for a very good introduction):  we start by constructing birth and death processes that bound $I^{(n)}_{1}(t)$ above and below provided $\bar{S}^{(n)}(t)$ and $\bar{N}^{(n)}(t)$ remain within $\varepsilon$ of the endemic equilibrium, and such that the upper and lower bounds approach one another as $\varepsilon \to 0$.  Finally, we show that the probability that $\bar{S}^{(n)}(t)$ and $\bar{N}^{(n)}(t)$ depart a $\varepsilon$-neighbourhood of the endemic equilibrium before strain $2$ has either successfully invaded or gone extinct goes to 0 as $n \to \infty$.  Since $\varepsilon$ is arbitrary, we recover the na\"\i ve branching process result.

\subsubsection{Macroscopic Initial Frequencies}
In this section we prove the following extinction of part (i) of Proposition \ref{SSPF}, where the population is infected with $d \geq 2$ strains.
\begin{prop}
Suppose that $R_{0,1} > R_{0,i}$ for all $i > 1$.  If $\bar{I}^{(n)}_{1}(0) \to I_{1}(0) > 0$, then all strains $i > 1$ will go extinct with high probability.   
\end{prop}
In light of the results in Section \ref{ASYMPT}, we might, without loss of generality, assume that at time $t = 0$, the process is in some neighbourhood of endemic fixed point $\bar{\bm{E}}^{\star,1} = (\bar{S}^{\star,1},\bar{I}^{\star,1}_{1},\ldots,\bar{I}^{\star,1}_{d},\bar{N}^{\star,1})$, as defined by \eqref{EQ}.
%where 
%\[
%	\bar{S}^{\star} = \frac{\lambda}{\delta R_{0,1}}
%		\left(1-\frac{\alpha_{1}(R_{0,1}-1)}{\beta_{1}-\alpha_{1}}\right), \quad
%	\bar{I}^{\star}_{i} = \begin{cases} 
%		\frac{\lambda(R_{0,1}-1)}{\beta_{1}-\alpha_{1}} & i = 1,\\ 
%		0 &  i > 1,
%	\end{cases} \quad \text{and} \quad
%	\bar{N}^{\star} = R_{0,1}\bar{S}^{\star}.
%\]
In particular, fix $\varepsilon > 0 $ such that 
\[
	1 < \frac{\bar{S}^{\star,1} - \varepsilon}{\bar{N}^{\star,1} + \varepsilon} 
	< \frac{1}{R_{0,1}}  = \frac{\bar{S}^{\star,1}}{\bar{N}^{\star,1}} 
	< \frac{\bar{S}^{\star,1} + \varepsilon}{\bar{N}^{\star,1} - \varepsilon} 
	< \frac{1}{R_{0,2}}
	< \frac{1}{R_{0,3}} 
	< \cdots 
	< \frac{1}{R_{0,d}}.
\]
For reasons that will become transparent below, we will assume that 
\[
	\|\bar{\bm{E}}^{(n)}(0) - \bar{\bm{E}}^{\star,1}\| < B \varepsilon.
\]
for some constant $0 < B < 1$ that will be determined later.

Let 
\[
	\tau^{(n)}_{\varepsilon} 
	\defn \inf\left\{ t : \|\bar{\bm{E}}^{(n)}(t) - \bar{\bm{E}}^{\star,1}\| > \varepsilon\right\},
\]
where we adopt the convention that $\tau^{(n)}_{\varepsilon} = \infty$ if $\|\bar{\bm{E}}^{(n)}(t) - \bar{\bm{E}}^{\star,1}\| < \varepsilon$ for all $t$.   

Provided $t < \tau^{(n)}_{\varepsilon}$, we have that 
\[
	\frac{\bar{S}^{\star,1} - \varepsilon}{\bar{N}^{\star,1} + \varepsilon} 
	< \frac{\bar{S}^{(n)}(t)}{\bar{N}^{(n)}(t)} 
	< \frac{\bar{S}^{\star,1} + \varepsilon}{\bar{N}^{\star,1} - \varepsilon}.
\]
Next, fix $\eta > 0$ sufficiently small that 
\[
	 (\beta_{i} + \eta) \frac{\bar{S}^{\star,1} + \varepsilon}{\bar{N}^{\star,1} - \varepsilon} 
	 	- (\delta+\alpha_{i}+\gamma_{i} - 3 \eta)
	 = \beta_{i} \left(\frac{\bar{S}^{\star,1} + \varepsilon}{\bar{N}^{\star,1} - \varepsilon} 
	 	- \frac{1}{R_{0,i}}\right)  
		+ \eta \left(3 + \frac{\bar{S}^{\star,1} + \varepsilon}{\bar{N}^{\star,1} - \varepsilon} \right)
	 < 0 
\]

Since $I^{(n)}_{i}(0) \to I_{i}(0)$ for $i > 1$ and $\beta^{(n)}_{i} \to \beta_{i}$, \etc we can assume that $n$ is sufficiently large that
\[
	|\beta^{(n)}_{i} - \beta_{i}| < \eta, 
	\quad |\delta^{(n)} - \delta| < \eta, 
	\quad |\alpha^{(n)}_{i} - \alpha_{i}| < \eta, 
	\quad \text{and} \quad
	|\gamma^{(n)}_{i} - \gamma_{i}| < \eta
\]
for all $i > 1$.

Thus, for $t < \tau^{(n)}_{\varepsilon}$, the per-infective transmission rate for strain $i$ satisfies
\[
	(\beta_{i} - \eta) \frac{\bar{S}^{\star,1} - \varepsilon}{\bar{N}^{\star,1} + \varepsilon}
	 < \frac{\beta^{(n)}_{i} S^{(n)}(t-)}{N^{(n)}(t-)} 
	 < (\beta_{i} + \eta) \frac{\bar{S}^{\star,1} + \varepsilon}{\bar{N}^{\star,1} - \varepsilon}
\]
whereas the total per-infective rate of removal of strain $i$ satisfies 
\[
	\delta+\alpha_{i}+\gamma_{i} + 3 \eta > \delta^{(n)}+\alpha^{(n)}_{i}+\gamma^{(n)}_{i} 
		> \delta+\alpha_{i}+\gamma_{i} - 3 \eta.
\]

\begin{lem}
Provided $t < \tau^{(n)}_{\varepsilon}$, the number of infectives of strain $i$ is stochastically smaller
\footnote{Given random variables $X$ and $X'$, we say that $X$ is \textit{stochastically smaller} than $X'$, denoted $X \preceq X'$ if 
\[
	\mathbb{P}\{X' \geq x\} \geq  \mathbb{P}\{X \geq x\}.
\]
Similarly, a stochastic process $X$ is stochastically smaller than the process $X'$ if $S(t) \preceq X'(t)$ for all $t \geq 0$.  One defines \textit{stochastically greater} analogously.}
than the birth and death process $Z^{+}_{i}(t)$ with $Z^{+}_{i}(0) = I^{(n)}_{i}(0)$ and birth and death rates 
\[
	 (\beta_{i} + \eta) \frac{\bar{S}^{\star,1} + \varepsilon}{\bar{N}^{\star,1} - \varepsilon} 
	 \quad \text{and} \quad
	 \delta+\alpha_{i}+\gamma_{i} - 3 \eta,	
\]
and stochastically greater than the birth and death process $Z^{-}_{i}(t)$ with 
$Z^{-}_{i}(0) = I^{(n)}_{i}(0)$ and birth and death rates 
\[
	 (\beta_{i} - \eta) \frac{\bar{S}^{\star,1} - \varepsilon}{\bar{N}^{\star,1} + \varepsilon} 
	 \quad \text{and} \quad
	 \delta+\alpha_{i}+\gamma_{i} + 3 \eta.	
\]
\end{lem}

\begin{proof}
It suffices to construct coupled versions of $I^{(n)}_{i}(t)$, $Z^{+}_{i}(t)$ and $Z^{-}_{i}(t)$ such that 
\[
	Z^{+}_{i}(t) \geq I^{(n)}_{i}(t) \geq Z^{-}_{i}(t).
\]
We will do so inductively, at each step constructing the processes up to the next among the 
aggregated jump times of all three processes, which we denote 
\[	
	0 = \tau_{0} < \tau_{1} < \tau_{2} < \cdots.
\]
For our underlying probability space, we assume sequences of independent rate 1 exponential random variables $\mathcal{E}_{k}$ and independent uniformly distributed random variables $\mathcal{U}_{k}$ on $[0,1]$, for $k=1,2,\ldots$.

Suppose that $I^{(n)}_{i}(t)$, $Z^{+}_{i}(t)$ and $Z^{-}_{i}(t)$ have been constructed up to $\tau_{k}$ (trivially true for $k = 0$).  Note that 
\[
	\rho_{k+1} \defn (\beta_{i} + \eta) \frac{\bar{S}^{\star,1} + \varepsilon}{\bar{N}^{\star,1} - \varepsilon} 
		Z^{+}_{k}(\tau_{k})
		+ (\delta+\alpha_{i}+\gamma_{i} + 3) Z^{+}_{k}(\tau_{k})
\]
is an upper bound on the combined rate of all transitions for all three processes.  Set 
\[
	\tau_{k+1} = \frac{\mathcal{E}_{k+1}}{\rho_{k+1}},
\]
so $\tau_{k+1}$ is a rate $\rho_{k+1}$ exponential random variable. Next, for $t < \tau_{k+1}$ we set
\[
	I^{(n)}_{i}(t) = I^{(n)}_{i}(\tau_{k}), \quad
	Z^{+}_{i}(t) = Z^{+}_{i}(\tau_{k}) \quad \text{and} \quad
	Z^{-}_{i}(t) = Z^{-}_{i}(\tau_{k}).
 \]
Finally, we set
\begin{multline*}
(I^{(n)}_{i}(\tau_{k\!+\!1}),Z^{\!+\!}_{i}(\tau_{k\!+\!1}),Z^{\!-\!}_{i}(\tau_{k\!+\!1}))\\
= \begin{cases}
(I^{(n)}_{i}(\tau_{k})\!+\!1,Z^{\!+\!}_{i}(\tau_{k})\!+\!1,Z^{\!-\!}_{i}(\tau_{k})\!+\!1)
& \text{if $\mathcal{U}_{k\!+\!1} 
\leq \frac{(\beta_{i} \!-\! \eta) \frac{\bar{S}^{\star,1} \!-\! \varepsilon}{\bar{N}^{\star,1} \!+\! \varepsilon}}{\rho_{k\!+\!1}}$}\\
(I^{(n)}_{i}(\tau_{k})\!+\!1,Z^{\!+\!}_{i}(\tau_{k})\!+\!1,Z^{\!-\!}_{i}(\tau_{k}))
& \text{if $\frac{(\beta_{i} \!-\! \eta) \frac{\bar{S}^{\star,1} \!-\! \varepsilon}{\bar{N}^{\star,1} \!+\! \varepsilon}}{\rho_{k\!+\!1}}
\leq \mathcal{U}_{k\!+\!1} < \frac{\beta_{i}\frac{\bar{S}^{(n)}(t)}{\bar{N}^{(n)}(t)}}{\rho_{k\!+\!1}}$}\\
(I^{(n)}_{i}(\tau_{k})\!+\!1,Z^{\!+\!}_{i}(\tau_{k}),Z^{\!-\!}_{i}(\tau_{k}))
& \text{if $\frac{\beta_{i}\frac{\bar{S}^{(n)}(t)}{\bar{N}^{(n)}(t)}}{\rho_{k\!+\!1}} 
< \mathcal{U}_{k\!+\!1} 
\leq  \frac{(\beta_{i} \!+\! \eta) \frac{\bar{S}^{\star,1} \!+\! \varepsilon}{\bar{N}^{\star,1} \!-\! \varepsilon}}{\rho_{k\!+\!1}}$}\\
(I^{(n)}_{i}(\tau_{k})\!-\!1,Z^{\!+\!}_{i}(\tau_{k})\!-\!1,Z^{\!-\!}_{i}(\tau_{k})\!-\!1)
& \text{if $\frac{(\beta_{i} \!+\! \eta) \frac{\bar{S}^{\star,1} \!+\! \varepsilon}{\bar{N}^{\star,1} \!-\! \varepsilon}}{\rho_{k\!+\!1}}
< \mathcal{U}_{k\!+\!1} 
\leq  \frac{(\beta_{i} \!+\! \eta) \frac{\bar{S}^{\star,1} \!+\! \varepsilon}{\bar{N}^{\star,1} \!-\! \varepsilon}
\!+\! \delta\!+\!\alpha_{i}\!+\!\gamma_{i} \!-\! 3\eta}{\rho_{k\!+\!1}}$}\\
(I^{(n)}_{i}(\tau_{k}),Z^{\!+\!}_{i}(\tau_{k})\!-\!1,Z^{\!-\!}_{i}(\tau_{k})\!-\!1)
& \text{if $\frac{(\beta_{i} \!+\! \eta) \frac{\bar{S}^{\star,1} \!+\! \varepsilon}{\bar{N}^{\star,1} \!-\! \varepsilon}
\!+\! \delta\!+\!\alpha_{i}\!+\!\gamma_{i} \!-\! 3\eta}{\rho_{k\!+\!1}}
< \mathcal{U}_{k\!+\!1} 
\leq \frac{(\beta_{i} \!+\! \eta) \frac{\bar{S}^{\star,1} \!+\! \varepsilon}{\bar{N}^{\star,1} \!-\! \varepsilon}
\!+\! \delta^{(n)}\!+\!\alpha^{(n)}_{i}\!+\!\gamma^{(n)}_{i}}{\rho_{k\!+\!1}}$}\\
(I^{(n)}_{i}(\tau_{k}),Z^{\!+\!}_{i}(\tau_{k}),Z^{\!-\!}_{i}(\tau_{k})\!-\!1)
& \text{if $\frac{(\beta_{i} \!+\! \eta) \frac{\bar{S}^{\star,1} \!+\! \varepsilon}{\bar{N}^{\star,1} \!-\! \varepsilon}
\!+\! \delta^{(n)}\!+\!\alpha^{(n)}_{i}\!+\!\gamma^{(n)}_{i}}{\rho_{k\!+\!1}}
< \mathcal{U}_{k\!+\!1} \leq 1$}.	
\end{cases}
\end{multline*}
It is readily verified that the resulting processes have the correct jump rates.
\end{proof}

\begin{rem} Note that given $\varepsilon > \varepsilon_{1} > 0$, choosing $\eta > 0$ as before and $\eta_{1} > 0$ analogously, we can similarly construct processes $Z^{+,1}_{i}(t)$ and $Z^{-,1}_{i}(t)$ such that 
\[
	Z^{+}_{i}(t) \succeq Z^{+,1}_{i}(t) \succeq I^{(n)}_{i}(t) \succeq Z^{-,1}_{i}(t) 
		\succeq Z^{-}_{i}(t),
\] 
\etc We shall apply this with a decreasing sequence of values $\varepsilon_{n} > 0$ below.
\end{rem}

Now, our choice of $\eta$ ensures that $Z^{+}_{i}(t)$ is subcritical for all $i > 1$.  In particular, setting 
\[
	\mu^{+}_{i} 
	\defn  (\beta_{i} + \eta) \frac{\bar{S}^{\star,1} + \varepsilon}{\bar{N}^{\star,1} - \varepsilon} 
	 	- (\delta+\alpha_{i}+\gamma_{i} - 3 \eta) < 0,
\]
we have 
\[
	\mathbb{E}\left[Z^{+}_{i}(t)\right] = Z^{+}_{i}(0) e^{\mu^{+}_{i} t},
\]
and, for any sequence $t_{n} > \frac{1}{|\mu^{+}_{i}|} \ln{n}$, for all $i > 1$, we have, using Markov's inequality
\[
	\mathbb{P}\left\{Z^{+}_{i}(t_{n}) \geq 1\right\}
	\leq \mathbb{E}\left[Z^{+}_{i}(t_{n})\right]
	\leq B \varepsilon n e^{\mu^{+}_{i} t_{n}} \to 0 
\]
as $n \to \infty$.  Thus,  if we show that $\mathbb{P}\left\{\tau^{(n)}_{\varepsilon} > t_{n}\right\} \to 1$ as $n \to \infty$, we can conclude that all strains $i > 1$ vanish after before $t_{n}$ with high probability.

To this end, we start by defining
\[
	\tau^{(n)}_{\varepsilon,i} 
	\defn \inf\left\{ t : |\bar{I}^{(n)}_{i}(t) - \bar{I}^{\star,1}_{i}| \geq \varepsilon n\right\},
\]
for $i = 1,\ldots,d$. We define $\tau^{(n)}_{\varepsilon,0}$ and $\tau^{(n)}_{\varepsilon,d+1}$ similarly, replacing $\bar{I}_{i}$ by $\bar{S}$ or $\bar{N}$ respectively in the above definition (again, $\tau^{(n)}_{\varepsilon,i} = \infty$ should the respective process never exceed $\varepsilon n$).  We then have  
\[
	\tau^{(n)}_{\varepsilon} = \min_{i} \tau^{(n)}_{\varepsilon,i}.
\] 

We continue with a classical result for birth and death processes: for $i > 1$, let 
\[
a_{i} = \frac{\delta+\alpha_{i}+\gamma_{i} - 3 \eta}{(\beta_{i} + \eta) \frac{\bar{S}^{\star,1} + \varepsilon}{\bar{N}^{\star,1} - \varepsilon}} > 1.  
\]
Then, a simple calculation shows that 
\[
	\mathbb{E}\left[a_{i}^{Z^{+}_{i}(t)} \middle\vert Z^{+}_{i}(s)\right] 
		= a_{i}^{Z^{+}_{i}(s)}
\]
\ie $a_{i}^{Z^{+}_{i}(t)}$ is a martingale.  Let 
\[
	\tau^{(n)}_{0,i} = \inf\left\{ t : Z^{+}_{i}(t) = 0\right\}
	\quad \text{and} \quad 
	\tau^{(n)}_{i} = \min\{\tau^{(n)}_{0,i},\tau^{(n)}_{\varepsilon,i}\}.
\]
We saw above that $\tau^{(n)}_{0,i}$ and thus $\tau^{(n)}_{i}$ are with high probability bounded above by any sequence $t_{n} > \frac{1}{\min_{i>1} |\mu^{+}_{i}|} \ln{n}$.  Now,
\[
	a_{i}^{Z^{+}_{i}(0)} = \mathbb{E}\left[a_{i}^{Z^{+}_{i}(\tau^{(n)}_{i})}\right] 
	= a_{i}^{\lfloor\varepsilon n\rfloor} 
		\mathbb{P}\left\{\tau^{(n)}_{0,i} > \tau^{(n)}_{\varepsilon,i}\right\}
	+ \left(1-\mathbb{P}\left\{\tau^{(n)}_{0,i} > \tau^{(n)}_{\varepsilon,i}\right\}\right)
\]
\ie
\[
	\mathbb{P}\left\{\tau^{(n)}_{0,i} > \tau^{(n)}_{\varepsilon,i}\right\}
	= \frac{a_{i}^{Z^{+}_{i}(0)} - 1}{a_{i}^{\lfloor\varepsilon n\rfloor}-1}
	\geq \frac{a_{i}^{\lfloor B \varepsilon n\rfloor} - 1}{a_{i}^{\lfloor\varepsilon n\rfloor}-1},
\]
which converges to 0 as $n \to \infty$ for any $B < 1$.  We thus have 
\[
	\mathbb{P}\left\{\tau^{(n)}_{\varepsilon,i} > t_{n} \right\} \to 1 
\]
as $n \to \infty$.

For the remaining three values $\tau^{(n)}_{\varepsilon,i}$, $i = 0,1,d+1$, we take a different approach, as the values $\bar{S}^{\star,1}$, $\bar{I}^{\star,1}_{1}$, and $\bar{N}^{\star,1}$ are all non-zero and a branching process approach is no longer appropriate.  Instead, we recall the SDE representation of our process (Proposition \ref{prop:SDE}).

Now $\bar{\bm{E}}^{\star,1}$ is a stable fixed point for the dynamical system $\dot{\bar{\bm{E}}} = 
\bm{F}(\bar{\bm{E}})$, so we may write 
\[
	\bm{F}(\bar{\bm{E}}) = A(\bar{\bm{E}}-\bar{\bm{E}}^{\star,1}) 
		+ \bm{G}(\bar{\bm{E}}-\bar{\bm{E}}^{\star,1}),
\]
where $A \defn \bm{D} \bm{F}(\bar{\bm{E}}^{\star,1})$, the Jacobian of $\bm{F}(\bm{x})$ evaluated at the resident endemic equilibrium, is a stable matrix and 
\[
	\norm{\bm{G}(\bar{\bm{E}})} \leq M \norm{\bar{\bm{E}}}^{2}
\]
for some fixed $M > 0$.

Now, let
\[
	\Xi^{(n)}(t) \defn \bar{\bm{E}}^{(n)}(t)-\bar{\bm{E}}^{\star,1}.
\]
Then, using Duhamel's principle, we have that
\begin{multline}\label{Xisde}
	\Xi^{(n)}(t) = e^{tA} \Xi^{(n)}(0)
	+ \int_{0}^{t} e^{(t-s)A} \bm{G}\left(\Xi^{(n)}(s)\right)\, ds\\
	+ \frac{1}{n} \int_{0}^{t} e^{(t-s)A} (\bm{e}_{0} + \bm{e}_{d+1})\, dM^{(n)}_{\bm{e}_{0} + \bm{e}_{d+1}}(t)
	- \frac{1}{n}  \int_{0}^{t} e^{(t-s)A}(\bm{e}_{0} + \bm{e}_{d+1})\, dM^{(n)}_{-\bm{e}_{0} - \bm{e}_{d+1}}(t)\\
	+ \frac{1}{n} \sum_{i=1}^{d} \int_{0}^{t} e^{(t-s)A}(\bm{e}_{i} - \bm{e}_{0})\, dM^{(n)}_{-\bm{e}_{0} + \bm{e}_{i}}(t)
	- \frac{1}{n} \sum_{i=1}^{d} \int_{0}^{t} e^{(t-s)A}(\bm{e}_{i} + \bm{e}_{d+1})\, dM^{(n)}_{-\bm{e}_{i} - \bm{e}_{d+1}}(t)\\
	- \frac{1}{n} \sum_{i=1}^{d}  \int_{0}^{t} e^{(t-s)A}\bm{e}_{i}\, dM^{(n)}_{-\bm{e}_{i}}(t)
	- \frac{1}{n} \int_{0}^{t} e^{(t-s)A}\bm{e}_{d+1}\, dM^{(n)}_{-\bm{e}_{d+1}}(t).
\end{multline}

Let $\alpha_{1},\ldots,\alpha_{d}$ denote the eigenvalues of $A$.  It is a standard result (see \eg \citet{Teschl2012}) that, given any 
\[ 
	\alpha < \min\{ -\Re(\alpha_{j}) : \Re(\alpha_{j}) < 0\},
\]
there exists a constant $C$, depending on $\alpha$ such that, when restricted to $E_{s}$, we have
\[
	\|e^{(t-r)A}\| \leq C e^{-\alpha(t-r)}.
\]
Thus,
\begin{multline*}
	\|\Xi^{(n)}(t)\| \leq  C e^{-\alpha t} \|\Xi^{(n)}(0)\|
	+ \int_{0}^{t} C e^{-\alpha(t-s)} M \|\Xi^{(n)}(s)\|^{2}\, ds\\
	+ \frac{1}{n} \int_{0}^{t} C e^{-\alpha(t-s)}\, dM^{(n)}_{\bm{e}_{0} + \bm{e}_{d+1}}(t)
	+ \frac{1}{n}  \int_{0}^{t} C e^{-\alpha(t-s)}\, dM^{(n)}_{-\bm{e}_{0} - \bm{e}_{d+1}}(t)\\
	+ \frac{1}{n} \sum_{i=1}^{d} \int_{0}^{t} C e^{-\alpha(t-s)}\, dM^{(n)}_{-\bm{e}_{0} + \bm{e}_{i}}(t)
	+ \frac{1}{n} \sum_{i=1}^{d} \int_{0}^{t} C e^{-\alpha(t-s)}\, dM^{(n)}_{-\bm{e}_{i} - \bm{e}_{d+1}}(t)\\
	+ \frac{1}{n} \sum_{i=1}^{d}  \int_{0}^{t} C e^{-\alpha(t-s)}\, dM^{(n)}_{-\bm{e}_{i}}(t)
	+ \frac{1}{n} \int_{0}^{t} C e^{-\alpha(t-s)}\, dM^{(n)}_{-\bm{e}_{d+1}}(t).
\end{multline*}

Now, fix a sequence $\ln{n} \ll t_{n} \ll n$, we observe that 
\begin{multline*}
	\mathbb{P}\left\{\tau^{(n)}_{\varepsilon} < t_{n}\right\}
	= \mathbb{P}\left\{\sup_{t \leq \tau^{(n)}_{\varepsilon} \wedge t_{n}} 
		\|\Xi^{(n)}(t)\| \geq \varepsilon\right\}\\
	\leq \mathbb{P}\left\{C e^{-\alpha t} \|\Xi^{(n)}(0)\| \geq \frac{\varepsilon}{5+3d}\right\}
	+\mathbb{P}\left\{\sup_{t \leq \tau^{(n)}_{\varepsilon}} 
		\int_{0}^{t} C e^{-\alpha(t-s)} M \|\Xi^{(n)}(s)\|^{2}\, ds \geq \frac{\varepsilon}{5+3d}\right\}\\
	+ \mathbb{P}\left\{\frac{1}{n} \sup_{t \leq \tau^{(n)}_{\varepsilon} \wedge t_{n}} 
		\int_{0}^{t} C e^{-\alpha(t-s)}\, dM^{(n)}_{\bm{e}_{0} + \bm{e}_{d+1}}(s) \geq \frac{\varepsilon}{5+3d}\right\}\\
	+ \mathbb{P}\left\{\frac{1}{n}  \sup_{t \leq \tau^{(n)}_{\varepsilon} \wedge t_{n}} 
		\int_{0}^{t} C e^{-\alpha(t-s)}\, dM^{(n)}_{-\bm{e}_{0} - \bm{e}_{d+1}}(s) \geq \frac{\varepsilon}{5+3d}\right\}\\	
	+ \sum_{i=1}^{d} \mathbb{P}\left\{\frac{1}{n} \sup_{t \leq \tau^{(n)}_{\varepsilon} \wedge t_{n}} 
		\int_{0}^{t} C e^{-\alpha(t-s)}\, dM^{(n)}_{-\bm{e}_{0} + \bm{e}_{i}}(s) \geq \frac{\varepsilon}{5+3d}\right\}\\
	+ \sum_{i=1}^{d} \mathbb{P}\left\{\frac{1}{n} \sup_{t \leq \tau^{(n)}_{\varepsilon} \wedge t_{n}} 
		\int_{0}^{t} C e^{-\alpha(t-s)}\, dM^{(n)}_{-\bm{e}_{i} - \bm{e}_{d+1}}(s) \geq \frac{\varepsilon}{5+3d}\right\}\\
	+ \sum_{i=1}^{d} \mathbb{P}\left\{\frac{1}{n} \sup_{t \leq \tau^{(n)}_{\varepsilon} \wedge t_{n}}
		\int_{0}^{t} C e^{-\alpha(t-s)}\, dM^{(n)}_{-\bm{e}_{i}}(s) \geq \frac{\varepsilon}{5+3d}\right\}\\
	+ \mathbb{P}\left\{\frac{1}{n} \sup_{t \leq \tau^{(n)}_{\varepsilon} \wedge t_{n}} 
		\int_{0}^{t} C e^{-\alpha(t-s)}\, dM^{(n)}_{-\bm{e}_{d+1}}(s) \geq \frac{\varepsilon}{5+3d}\right\}
\end{multline*}
Now,
\[
	C e^{-\alpha t} \|\Xi^{(n)}(0)\| \leq C B \varepsilon
\]
and 
\[
	\sup_{t \leq \tau^{(n)}_{\varepsilon}} \int_{0}^{t} C e^{-\alpha(t-s)} M \|\Xi^{(n)}(s)\|^{2}\, ds
	\leq \frac{C M}{\alpha} \varepsilon^{2}.
\]
Thus, provided we choose 
\[
	B < \frac{1}{C(5+3d)} \quad \text{and} \quad \varepsilon < \frac{\alpha}{CM(5+3d)},
\]
then
\[
	\mathbb{P}\left\{C e^{-\alpha t} \|\Xi^{(n)}(0)\| \geq \frac{\varepsilon}{5+3d}\right\}
	= \mathbb{P}\left\{\sup_{t \leq \tau^{(n)}_{\varepsilon}} 
		\int_{0}^{t} C e^{-\alpha(t-s)} M \|\Xi^{(n)}(s)\|^{2}\, ds \geq \frac{\varepsilon}{5+3d}\right\}
	= 0,
\]

Finally, we turn to the integrals $\int_{0}^{t} e^{-\alpha(t-s)}\, dM^{(n)}_{\bm{l}}(s)$.  Recall that each of the integrators $M^{(n)}_{\bm{l}}(t)$ takes the form 
\[
	\tilde{P}_{\bm{l}}\left(n \int \Lambda_{\bm{l}}(\bar{\bm{E}}^{(n)}(t))\, ds\right),
\]
for some continuous function $\Lambda_{\bm{l}}$.  We will thus prove the generic lemma:

\begin{lem}\label{intest}
Let $P$ be a Poisson process, $\Lambda: \mathbb{R}^{d+2} \to \mathbb{R}$ be continuous, and let
\[
	M^{(n)}(t) = \tilde{P}\left(n \int \Lambda(\bar{\bm{E}}^{(n)}(t))\, ds\right).
\]
then for any $\alpha > 0$, any constant $C$, any sequence $t_{n} \ll n$, and $\tau^{(n)}_{\varepsilon}$ as above, we have 
\[
	\mathbb{P}\left\{\frac{C}{n}  \sup_{t \leq \tau^{(n)}_{\varepsilon} \wedge t_{n}}  
		\int_{0}^{t} e^{-\alpha(t-s)}\, dM^{(n)}(s) > R\right\} \to 0 
\]
as $n \to \infty$, for any fixed $R > 0$.
\end{lem}	

\begin{proof}
We start by observing that for $t \leq \tau^{(n)}_{\varepsilon}$, $\bar{\bm{E}}^{(n)}(t)) \in \overline{B_{\varepsilon}(\bar{\bm{E}}^{\star,1})}$, a compact set, and thus 
\[
	\Lambda(\bar{\bm{E}}^{(n)}(t)) \leq \overline{\Lambda}
\]
for a constant $\overline{\Lambda} > 0$ depending on $\Lambda$, $\varepsilon$, and $\bar{\bm{E}}^{\star,1}$.
\begin{multline*}
	\mathbb{P}\left\{\frac{C}{n}  \sup_{t \leq \tau^{(n)}_{\varepsilon} \wedge t_{n}}  
		\int_{0}^{t} e^{-\alpha(t-s)}\, dM^{(n)}(s) > R\right\}
	= \mathbb{P}\left\{\sup_{t \leq \tau^{(n)}_{\varepsilon} \wedge t_{n}}
		e^{-\alpha t} \int_{0}^{t} e^{\alpha s}\, dM^{(n)}(s) > \frac{n R}{C} \right\}\\
	\leq \mathbb{P}\left\{\sup_{t \leq t_{n}}  e^{-\alpha t} \int_{0}^{t \wedge \tau^{(n)}_{\varepsilon}}
		e^{\alpha s}\, dM^{(n)}(s) >  \frac{n R}{C} \right\} \\
	\leq \sum_{k=0}^{t_{n}-1} \mathbb{P}\left\{\sup_{k < t \leq k+1}  e^{-\alpha t}
		 \int_{0}^{t \wedge \tau^{(n)}_{\varepsilon}} e^{\alpha s}\, dM^{(n)}(s) > \frac{n R}{C}\right\}\\
	\leq \sum_{k=0}^{t_{n}-1} \mathbb{P}\left\{\sup_{k < t \leq k+1}  e^{-\alpha k}
		\int_{0}^{t \wedge \tau^{(n)}_{\varepsilon}} e^{\alpha s}\, dM^{(n)}(s) > \frac{n R}{C}\right\}\\
	= \sum_{k=0}^{t_{n}-1} \mathbb{P}\left\{\sup_{k < t \leq k+1} 
		\int_{0}^{t \wedge \tau^{(n)}_{\varepsilon}} e^{\alpha s}\, dM^{(n)}(s) 
		> \frac{n R}{C} e^{\alpha k}\right\}\\
	\leq \sum_{k=0}^{t_{n}-1} \mathbb{P}\left\{\sup_{t \leq k+1} 
		\int_{0}^{t \wedge \tau^{(n)}_{\varepsilon}} e^{\alpha s}\, dM^{(n)}(s) 
		> \frac{n R}{C} e^{\alpha k}\right\}.
\end{multline*}
Now, applying Doob's inequality,
\begin{multline*}
	\mathbb{P}\left\{\sup_{t \leq k+1} 
		\int_{0}^{t \wedge \tau^{(n)}_{\varepsilon}} e^{\alpha s}\, dM^{(n)}(s) 
		> \frac{n R}{C} e^{\alpha k}\right\}
	\leq \frac{C^{2}}{n^{2}R^{2}} e^{-2\alpha k}\mathbb{E}\left[\left(
		\int_{0}^{k+1 \wedge \tau^{(n)}_{\varepsilon}} e^{\alpha s}\, dM^{(n)}(s)\right)^{2}\right]\\
	= \frac{C^{2}}{n^{2}R^{2}} e^{-2\alpha k} 
		\mathbb{E}\left[\int_{0}^{k+1 \wedge \tau^{(n)}_{\varepsilon}} e^{2\alpha s} 
			n \Lambda(\bar{\bm{E}}^{(n)}(s))\, ds\right]
	\leq \frac{C^{2}}{n R^{2}} e^{-2\alpha k} 
		\mathbb{E}\left[\int_{0}^{k+1} e^{2\alpha s} \overline{\Lambda}\, ds\right]\\
	\leq \frac{C^{2}\overline{\Lambda}}{2\alpha n R^{2}} e^{-2\alpha k} 
		\left(e^{2\alpha (k+1)}-1\right)
	\leq \frac{C^{2}\overline{\Lambda}}{2\alpha n R^{2}} e^{2\alpha} 
\end{multline*}
Thus,
\[
	\mathbb{P}\left\{\frac{C}{n}  \sup_{t \leq \tau^{(n)}_{\varepsilon} \wedge t_{n}}  
		\int_{0}^{t} e^{-\alpha(t-s)}\, dM^{(n)}(s) > R\right\}
		\leq \frac{C^{2}\overline{\Lambda}}{2\alpha n R^{2}} e^{2\alpha} t_{n} \to 0
\]
as $n \to \infty$.
\end{proof}

Applying the lemma, we have 
\[
	\mathbb{P}\left\{\tau^{(n)}_{\varepsilon} < t_{n}\right\} \to 0
\]
and, with high probability, all strains $i > 1$ will vanish before $t_{n}$, and moreover, during this time, the process  $\bar{\bm{E}}^{(n)}(t))$ will remain in $\overline{B_{\varepsilon}(\bar{\bm{E}}^{\star,1})}$.
		
\subsubsection{Novel Strain in Small Number of Copies}

We now consider the possibility that a new strain invades an established population.  From the results of the previous section, we see that generically, the population will eventually be in an $\varepsilon$- neighbourhood of the fixed point $\bar{\bm{E}}^{\star,1}$ for arbitrary $\varepsilon > 0$, and that $I^{(n)}_{i}(t) \equiv 0$ for all $i >1$.

If the new strain has reproductive number less than $R_{0,1}$, then the arguments of the preceding section apply directly, and we can conclude that the novel strain will very rapidly go extinct.

If, however, it has higher reproductive number, the invading strain now has a non-zero probability of establishing itself and replacing the resident strain.  In this section, we adapt the techniques above to deal with this case (by the above, we may take $d=2$).  

We start by defining the quantities $\tau^{(n)}_{\varepsilon,i}$ and $\tau^{(n)}_{\varepsilon}$ as before, with the understanding that $\tau^{(n)}_{\varepsilon,1} = \infty$ if the new strain goes extinct before hitting $\varepsilon n$. 

We now fix $\varepsilon > 0 $ such that 
\[
	\frac{1}{R_{0,1}}
	< \frac{\bar{S}^{\star,1} - \varepsilon}{\bar{N}^{\star,1} + \varepsilon} 
	< \frac{1}{R_{0,2}}  = \frac{\bar{S}^{\star,1}}{\bar{N}^{\star,1}} 
	< \frac{\bar{S}^{\star,1} + \varepsilon}{\bar{N}^{\star,1} - \varepsilon},
\]
and $\eta > 0$ sufficiently small that 
\begin{multline*}
	 \mu^{+}_{2} \defn (\beta_{2} + \eta) \frac{\bar{S}^{\star,1} 
	 	+ \varepsilon}{\bar{N}^{\star,1} - \varepsilon} - (\delta+\alpha_{2}+\gamma_{2} - 3 \eta)\\
	>  \mu^{-}_{2} 
	\defn (\beta_{2} - \eta) \frac{\bar{S}^{\star,1} - \varepsilon}{\bar{N}^{\star,1} + \varepsilon} 
	 	- (\delta+\alpha_{2}+\gamma_{2} + 3 \eta)\\
	 = \beta_{2} \left(\frac{\bar{S}^{\star,1} - \varepsilon}{\bar{N}^{\star,1} + \varepsilon} 
	 	- \frac{1}{R_{0,2}}\right)  
		+ \eta \left(\frac{\bar{S}^{\star,1} - \varepsilon}{\bar{N}^{\star,1} + \varepsilon} - 3\right)
	 > 0, 
\end{multline*}
and suppose that $n$ is sufficiently large that $|\beta^{(n)}_{2} - \beta_{2}| < \eta$, \etc

Again, provided $t < \tau^{(n)}_{\varepsilon}$, we have that 
\[
	\frac{\bar{S}^{\star,1} - \varepsilon}{\bar{N}^{\star,1} + \varepsilon} 
	< \frac{\bar{S}^{(n)}(t)}{\bar{N}^{(n)}(t)} 
	< \frac{\bar{S}^{\star,1} + \varepsilon}{\bar{N}^{\star,1} - \varepsilon},
\]
and thus, if $Z^{+}(t)$ and $Z^{-}(t)$ are birth and death processes with birth and death rates
\[
	 (\beta_{2} + \eta) \frac{\bar{S}^{\star,1} + \varepsilon}{\bar{N}^{\star,1} - \varepsilon} 
	 \quad \text{and} \quad
	 \delta+\alpha_{2}+\gamma_{2} - 3 \eta	
\]
and
\[
	 (\beta_{2} - \eta) \frac{\bar{S}^{\star,1} - \varepsilon}{\bar{N}^{\star,1} + \varepsilon} 
	 \quad \text{and} \quad
	 \delta+\alpha_{2}+\gamma_{2} + 3 \eta	
\]
respectively, then both $Z^{+}(t)$ and $Z^{-}(t)$ are supercritical with Malthusian parameters $\mu^{+}_{2}$ and $\mu^{-}_{2}$, respectively, and 
\[
	Z^{-}(t) < I^{(n)}_{2}(t) < Z^{+}(t)
\]
stochastically for $t < \tau^{(n)}_{\varepsilon}$.

Now, set 
\[
	\overline{q}_{2} \defn
	\frac{\delta+\alpha_{2}+\gamma_{2} - 3 \eta}
		{(\beta_{2} + \eta) \frac{\bar{S}^{\star,1} + \varepsilon}{\bar{N}^{\star,1} - \varepsilon}}
	< \underline{q}_{2} \defn
	\frac{\delta+\alpha_{2}+\gamma_{2} + 3 \eta}
		{(\beta_{2} - \eta) \frac{\bar{S}^{\star,1} - \varepsilon}{\bar{N}^{\star,1} + \varepsilon}} < 1.
\]
Classical results for birth and death processes (\eg \citet{Athreya+Ney1972}) tell us that 
$Z^{+}(t)$ and $Z^{-}(t)$ will hit 0 in finite time with probability 
$\overline{q}_{2}^{Z^{+}(0)}$ and
$\underline{q}_{2}$ respectively, and will grow indefinitely otherwise, and, moreover, that there exist random variables $\overline{W}$ and $\underline{W}$ taking values on $[0,\infty)$ such that
\[
	\mathbb{P}\{\overline{W} = 0\} = \overline{q}_{2}^{Z^{+}(0)}
	\quad \text{and} \quad 
	\mathbb{P}\{\underline{W} = 0\} = \underline{q}_{2}^{Z^{-}(0)} 
\]
and 
\begin{equation}\label{W}
	e^{-\mu^{+}_{2}t}Z^{+}(t) \to \overline{W} 
	\quad \text{and} \quad 
	e^{-\mu^{-}_{2}t}Z^{-}(t) \to \underline{W} 
\end{equation}
both almost surely and in $L^{2}$. %We note also that by virtue of our coupled construction of the processes $Z^{+}(t)$ and $Z^{-}(t)$, we have that 
%\[
%	\{\overline{W} = 0\} \subseteq \{\underline{W} = 0\}.
%\]

Now, as before, fix  $\ln{n} \ll t_{n} \ll n$. Taking logarithms in \eqref{W}, we see that for almost all 
$\omega \not\in \{\overline{W} = 0\}$, we have 
\[
	\frac{\ln Z^{-}(\omega,t_{n})}{t_{n}}  \to \mu^{-}_{2}.
\]
Now, if $Z^{-}(t_{n}) \leq \varepsilon n$ then the right hand side converges to 0.  Moreover, if $I^{(n)}_{2}(t) \leq \varepsilon n$, then necessarily $Z^{-}(t_{n}) \leq \varepsilon n$, and we conclude that
\[
	\limsup_{n \to \infty} \mathbb{P}\left\{t_{n} < \tau^{(n)}_{\varepsilon}\right\} \leq 
		\mathbb{P}\{\underline{W} = 0\} =  \underline{q}_{2}^{Z^{-}(0)} . 
\]
In particular, if we can establish that with high probability 
$\tau^{(n)}_{\varepsilon} = \tau^{(n)}_{\varepsilon,2}$, then this gives a lower bound on the probability that the novel strain invades, as the latter is the probability that $\tau^{(n)}_{\varepsilon}$ is finite.

Now, 
\[
	\mathbb{P}\left\{\tau^{(n)}_{\varepsilon} < \tau^{(n)}_{\varepsilon,2}\right\}
	 =  \mathbb{P}\left\{\tau^{(n)}_{\varepsilon} < \tau^{(n)}_{\varepsilon,2}; \tau^{(n)}_{\varepsilon} 
	 	< t_{n}\right\}
	+ \mathbb{P}\left\{\tau^{(n)}_{\varepsilon} < \tau^{(n)}_{\varepsilon,2}; \tau^{(n)}_{\varepsilon} 
	 	> t_{n}\right\}.
\]
We have already established that the latter is bounded above by $\underline{q}_{1}^{Z^{-}(0)} $ as $n \to \infty$.  The former follows almost exactly as the proof that $\mathbb{P}\left\{\tau^{(n)}_{\varepsilon} < t_{n}\right\} \to 0$ of the previous section;  $\bar{\bm{E}}^{\star,1}$ is now a hyperbolic fixed point rather than a stable fixed point, and $A$ has a positive eigenvalue, but the stable manifold of $\bar{\bm{E}}^{\star,1}$ coincides with the subset of $\mathbb{R}^{d+2}$ with $x_{1} = 0$.  In particular, we may decompose 
\[
	\mathbb{R}^{d+2} = E_{S} \oplus E_{X},
\]
where $E_{S} = \{x_{1} = 0\}$ and $E_{X}$ are $A$ invariant subspaces, corresponding to the sum of the generalised eigenspaces for eigenvalues with negative and positive real parts respectively.  We will write $P_{S}$ and $P_{X}$ for the corresponding projections (\ie $P_{S}$ has image $E_{S}$ and kernel $E_{X}$, and oppositely for $P_{X}$), and note that both $P_{S}$ and $P_{X}$ commute with $A$.  Proceeding exactly as previously, we define
\[
	\Xi^{(n)}(t) \defn \bar{\bm{E}}^{(n)}(t) - \bar{\bm{E}}^{\star,1}
\]
and let $\Xi^{(n)}_{S}(t) = P_{S} \Xi^{(n)}(t)$ denote it's projection onto the stable subspace.   Then,
if 
\[
	\tau^{(n)}_{\varepsilon,S} = \tau^{(n)}_{\varepsilon,0} \wedge \tau^{(n)}_{\varepsilon,2}
		\wedge \tau^{(n)}_{\varepsilon,3}
\]
then 
\[
	\mathbb{P}\left\{\tau^{(n)}_{\varepsilon} 
		< \tau^{(n)}_{\varepsilon,2};  t_{n} < \tau^{(n)}_{\varepsilon} \right\}
	= \mathbb{P}\left\{\tau^{(n)}_{\varepsilon,S} 
		= \tau^{(n)}_{\varepsilon};  t_{n} < \tau^{(n)}_{\varepsilon} \right\}
	= \mathbb{P}\left\{\sup_{t \leq \tau^{(n)}_{\varepsilon} \wedge t_{n}} 
		\|\Xi^{(n)}_{S}(t)\| \geq \varepsilon\right\}.
\]
Applying the arguments of the previous section, using the equation for $\Xi^{(n)}_{S}(t)$ obtained by letting $P_{S}$ act on both sides of \eqref{Xisde}, one obtains almost identically that
\[
	\mathbb{P}\left\{\sup_{t \leq \tau^{(n)}_{\varepsilon} \wedge t_{n}} 
		\|\Xi^{(n)}_{S}(t)\| \geq \varepsilon\right\} \to 0
\]
as $n \to \infty$.

Now, note that for any $\omega \in \{\overline{W} = 0\}$, we must have $Z^{+}(\omega,t) < \varepsilon n$ for all $t$, for some sufficiently large $n$, so $\omega \in \{\tau^{(n)}_{\varepsilon} < \tau^{(n)}_{\varepsilon,2}\}$.  Thus we have
\[
	\overline{q}_{1}^{Z^{+}(0)} = \mathbb{P}\{\overline{W} = 0\}
	\leq \liminf_{n \to \infty} \mathbb{P}\left\{\tau^{(n)}_{\varepsilon} < \tau^{(n)}_{\varepsilon,1}\right\}
\]

Finally, we notice that as $\varepsilon \to 0$ (and thus $\eta \to 0$ also,) both $\overline{q}_{2}$ and
$\underline{q}_{2}$ approach
\[
	q_{2} \defn \frac{\delta+\alpha_{2}+\gamma_{2}}{\beta_{2}
	\frac{\bar{S}^{\star,1}}{\bar{N}^{\star,1}}}
	= \frac{R_{0,1}}{R_{0,2}}.
\]
Since the choice of $\varepsilon$ was arbitrary in our definition of invasion, we conclude that the probability of successful invasion of the new strain 1 is 
\begin{equation}\label{strongfix}
	1-\left(\frac{R_{0,1}}{R_{0,2}}\right)^{I_{2}(0)}.
\end{equation}
Once this has happened, as we note above, Kurtz's deterministic approximation is applicable, and with high probability, the system will approach any arbitrarily small neighbourhood of the endemic fixed point for the new strain 1, at which point, by the argument above, the former resident strain, strain 2, goes extinct with probability approaching 1 as $n \to \infty$.

\section{Proof of Proposition \ref{prop:noneq}}
We will evaluate \eqref{QINTEGRAL} to order $o(\varepsilon)$ using the first and second order terms in the perturbative solution to \eqref{REDUCED}:
\[
	S(t) = S^{\star} + \varepsilon s(t) + o(\varepsilon), \quad
	I_{1}(t) = I_1^{\star} + \varepsilon i_{1}(t) + o(\varepsilon), \quad \text{and} \quad 
	N(0) = N^{\star} + \varepsilon n(t) + o(\varepsilon).
\]

Substituting these into \eqref{REDUCED} and gathering terms of order $\BigO{\varepsilon}$, we get the linear system
\[
	\frac{d}{dt}  \bm{e}(t) = A  \bm{e}(t),
\]
where $A = \bm{D}\bm{F}(\bm{E}^{\star,1})$, and
\[
	 \bm{e}(t) \defn \begin{bmatrix} s(t) \\ i_{1}(t) \\ n(t) \end{bmatrix}. 
\]
This has solution $\bm{e}(t) = e^{tA}  \bm{e}(0)$.

$A$ has eigenvalues $-\delta$ and
\[
	-\frac{1}{2} \frac{\delta R_{0,1}^{2} (\beta_{1} - \alpha_{1}) \pm 
	\sqrt{\delta R_{0,1}(\beta_{1} - \alpha_{1})\left(
	\delta R_{0,1}^{3}(\beta_{1} - \alpha_{1})
	-4\beta_{1}(R_{0,1}-1)(\beta_{1}-R_{0,1}\alpha_{1})\right)}}
	{R_{0,1}(\beta_{1}-R_{0,1}\alpha_{1})}
\]
all of which have negative real part, so $e^{tA}$ is bounded, and thus 	
\[
	\bm{E}(t) - \bm{E}^{\star} - \varepsilon \bm{e}(t)
\]
is bounded and of order $o(\varepsilon)$.  

On the other hand, we can write \eqref{QINTEGRAL} as
\begin{multline*}
	I=\int_{0}^{\infty} e^{-\int_{0}^{s} \beta_{2} \frac{S^{\star}}{N^{\star}} 
		- (\delta+\alpha_{2}+\gamma_{2})\, du}
		\left(1- \varepsilon \int_{0}^{s} \beta_{2}\left(\frac{s(u)}{N^{\star}} 
			- \frac{S^{\star}}{N^{\star}}\frac{n(u)}{N^{\star}}\right)\, du + o(\varepsilon)\right)
	 (\delta+\alpha_{2}+\gamma_{2})\, ds\\
	 = \int_{0}^{\infty} e^{-\int_{0}^{s} \beta_{2} \left(\frac{1}{R_{0,1}} 
	 	-\frac{1}{R_{0,2}}\right)\, du}
		\left(1- \varepsilon \int_{0}^{s} \frac{\beta_{2}}{R_{0,1}} \left(\frac{s(u)}{S^{\star}} 
			- \frac{n(u)}{N^{\star}}\right)\, du\right)
	 (\delta+\alpha_{2}+\gamma_{2})\, ds + o(\varepsilon)\\
	 = \frac{R_{0,2}}{R_{0,1}} - 1
	 -  \varepsilon \int_{0}^{\infty} e^{-\Lambda s} \int_{0}^{s} \bm{v}^{\top} e^{u A} \bm{e}(0)\, du\, ds
		+ o(\varepsilon)
\end{multline*}
for $\Lambda = \beta_{2} \left(\frac{1}{R_{0,1}} -\frac{1}{R_{0,2}}\right)$ and
\[
	\bm{v} = \frac{\beta_{2}^{2}}{R_{0,1}R_{0,2}}
		 \begin{bmatrix} \frac{1}{S^{\star}} \\ 0 \\ -\frac{1}{N^{\star}} \end{bmatrix},
\]
 provided $R_{0,2} > R_{0,1}$ (and thus $\Lambda > 0$). 
 
 Now, $A$ is a stable matrix and thus invertible, so 
 \[
 	\int_{0}^{s} e^{u A}\, du = A^{-1} (\bm{I} - e^{sA}),
\]
and the latter integral reduces to
\begin{align*}
	\int_{0}^{\infty} e^{-\Lambda s} \bm{v}^{\top} A^{-1} (\bm{I} - e^{sA}) \bm{e}(0)\, ds
	&= \frac{\bm{v}^{\top} A^{-1} \bm{e}(0)}{\Lambda}
	- \bm{v}^{\top} \left(\int_{0}^{\infty} e^{-s(\Lambda\bm{I} - A)}\, ds\right)  \bm{e}(0)\\
	&= \frac{\bm{v}^{\top} A^{-1} \bm{e}(0)}{\Lambda}
	- \bm{v}^{\top} A^{-1} (\Lambda\bm{I} - A)^{-1} \bm{e}(0),
\end{align*}
where we note that the inverse exists, as $\Lambda$, which is positive, is not an eigenvalue of $A$.  Evaluating the latter yields
\eqref{NONEQ}.

\putbib
\end{bibunit}

\end{document}